\documentclass[11pt]{article}

\usepackage{indentfirst}
\usepackage{amsfonts}
\usepackage{amsmath}
\usepackage{amsthm}
\usepackage{verbatim}
\usepackage[makeroom]{cancel}
\usepackage{rotating}
\usepackage{float}
\usepackage{amssymb}
\usepackage{mathrsfs}
\usepackage{epsfig}
\usepackage{graphicx}
\usepackage{indentfirst}
\usepackage{framed}
\usepackage[numbers]{natbib}
\usepackage{color}
\usepackage[backref=page]{hyperref}
\usepackage[capitalize]{cleveref}
\usepackage{enumitem}
\usepackage{float}
\usepackage{bm}
\usepackage{braket}
\newcommand{\ketbra}[2]{\lvert #1 \rangle \! \langle #2 \rvert} %
\newcommand{\kb}[1]{\ketbra{#1}{#1}} 

\usepackage[labelfont=bf]{caption}
\usepackage{soul}

\newcommand{\norm}[1]{\lVert#1\rVert}

\newcommand{\code}{\mathsf{code}}

\usepackage[margin=1in]{geometry}

\definecolor{corlinks}{RGB}{200,0,0}
\definecolor{cormenu}{RGB}{200,0,0}
\definecolor{corurl}{RGB}{200,0,0}

\definecolor{mygreen}{rgb}{0.0, 0.5, 0.0}

\def\01{\{0,1\}}

\hypersetup{
colorlinks=true,
urlcolor=corlinks,
linkcolor=corlinks,
menucolor=cormenu,
citecolor=corlinks,
pdfborder= 0 0 0
}

\newtheorem{question}{Question}
\newtheorem{theorem}{Theorem}
\numberwithin{theorem}{section}

\newtheorem{lemma}[theorem]{Lemma}
\newtheorem{corollary}[theorem]{Corollary}

\newtheorem{claim}[theorem]{Claim}

\newtheorem{definition}[theorem]{Definition}
\newtheorem{construction}[theorem]{Construction}
\newtheorem*{mainresult}{Main result}
\newtheorem{maintheorem}{Theorem}

\theoremstyle{definition}
\newtheorem{remark}[theorem]{Remark}

\DeclareMathOperator{\poly}{poly}

\def\colorful{1}

\ifnum\colorful=1

\fi
\ifnum\colorful=0

\fi

\newcommand{\negl}{\mathsf{negl}}
\newcommand{\calA}{\mathcal{A}}
\newcommand{\Is}{\mathcal{I}}
\newcommand{\calB}{\mathcal{B}}
\newcommand{\calS}{\mathcal{S}}
\newcommand{\calC}{\mathcal{C}}
\newcommand{\calP}{\mathcal{P}}

\newcommand{\calD}{\mathcal{D}}
\newcommand{\calG}{\mathcal{G}}
\newcommand{\De}{\mathcal{D}}

\newcommand{\calE}{\mathcal{E}}
\newcommand{\calU}{\mathcal{U}}
\newcommand{\nat}{\mathbb{N}}
\newcommand{\id}{\mathsf{Id}}
\newcommand{\NW}{\mathsf{NW}}
\newcommand{\PRG}{\mathsf{PRG}}

\newcommand{\eps}{\varepsilon}
\DeclareMathOperator*{\E}{\mathbb{E}}

\newcommand{\BQTIME}{\mathsf{BQTIME}}
\newcommand{\BQP}{\mathsf{BQP}}

\newcommand{\BQE}{\mathsf{BQE}}
\newcommand{\Lstar}{\ensuremath{L^\star}}

\newcommand{\eqdef}{\stackrel{\rm def}{=}}

\newcommand{\B}{\{-1,+1\}}
\newcommand{\NB}{\ensuremath{B' \setminus A}}
\newcommand{\NBp}{\ensuremath{B'' \setminus A}}

\newcommand{\correctCircuitA}{\widetilde{C}_A}

\begin{document}

\title{Quantum learning algorithms imply circuit lower bounds\vspace{0.5cm}}
%
%
%
%
%
%
%


\author{
Srinivasan Arunachalam \vspace{0.1cm} \\\small  IBM T.~J. Watson Research Center \\ \vspace{0.1cm} {\small  \texttt{srinivasan.arunachalam@ibm.com}} 
\and
Alex B. Grilo \vspace{0.1cm} \\\small Sorbonne Universit\'{e}, CNRS, LIP6 \\ \vspace{0.1cm} {\small  \texttt{Alex.Bredariol-Grilo@lip6.fr}}
\vspace{0.4cm}
\and
Tom Gur \vspace{0.1cm}
\\\small  University of Warwick\\ \vspace{0.1cm} {\small \texttt{tom.gur@warwick.ac.uk}}
\and
Igor C. Oliveira \vspace{0.1cm} \\\small  University of Warwick \\ \vspace{0.1cm} {\small \texttt{igor.oliveira@warwick.ac.uk}}
\and
Aarthi Sundaram \vspace{0.1cm} \\\small Microsoft Quantum \\ \vspace{0.1cm} {\small \texttt{aarthi.sundaram@microsoft.com}}\vspace{0.5cm}
}
\maketitle

\vspace{-0.6cm}

\begin{abstract}
We establish the first general connection between the design of quantum algorithms and  circuit lower bounds.  Specifically, let $\mathfrak{C}$ be a class of polynomial-size concepts, and suppose that $\mathfrak{C}$ can be PAC-learned with membership queries under the uniform distribution with error $1/2 - \gamma$ by a time $T$ quantum algorithm.
We prove that if $\gamma^2 \cdot T \ll 2^n/n$, then $\mathsf{BQE} \nsubseteq \mathfrak{C}$, where $\mathsf{BQE} = \mathsf{BQTIME}[2^{O(n)}]$ is an exponential-time analogue of~$\mathsf{BQP}$. This result is optimal in both $\gamma$ and $T$, since it is not hard to learn any class $\mathfrak{C}$ of functions in (classical) time $T = 2^n$  (with no error), or in quantum time $T = \mathsf{poly}(n)$ with error at most $1/2 - \Omega(2^{-n/2})$ via Fourier sampling. In other words, even a marginal quantum speedup over these generic learning algorithms would lead to major consequences in complexity lower bounds. As a consequence, our result shows that the study of quantum learning speedups is intimately connected to fundamental open problems about algorithms, quantum computing, and complexity theory.

\vspace{1.5mm}

Our proof builds on several works in learning theory, pseudorandomness, and computational complexity, and on a connection between non-trivial classical learning algorithms and circuit lower bounds established by Oliveira and Santhanam (CCC 2017). Extending their approach to quantum learning algorithms turns out to create significant challenges, since extracting computational hardness from a quantum computation is inherently more complicated. To achieve that, we show among other results how pseudorandom generators imply learning-to-lower-bound connections in a generic fashion, construct the first conditional pseudorandom generator secure against uniform quantum computations, and extend the local list-decoding algorithm of Impagliazzo, Jaiswal, Kabanets and Wigderson (SICOMP 2010) to quantum circuits via a delicate analysis. We believe that these contributions are of independent interest and might find other~applications.
\end{abstract}

\newpage

\setcounter{tocdepth}{3}
\tableofcontents

\newpage



%
%
%

\section{Introduction}
One of the salient goals of quantum computing is to understand which computational problems admit quantum speedups over classical algorithms. 
The canonical example is Shor's algorithm~\cite{DBLP:conf/focs/Shor94} for factoring, which achieves an exponential speedup over the best known classical algorithms. 
Another notable example is Grover's algorithm~\cite{DBLP:conf/stoc/Grover98}, which sheds light on quantum complexity theory by showing that expressive languages such as the quintessential
$\mathsf{NP}$-complete problem $\mathsf{Formula}$-$\mathsf{SAT}$ admits quantum algorithms of time complexity  $\widetilde{O}(2^{n/2})$, whereas in general, there are no known classical algorithms that outperform the $\widetilde{O}(2^n)$-time brute force search. This paper investigates the implications of quantum speedups within the setting of \emph{learning theory}.

Quantum learning theory forms the theoretical foundations which allow us to understand the potential power and limitations of quantum machine learning. At its core, this field studies quantum algorithms that are given quantum access (typically, quantum queries or quantum samples) to an unknown circuit $f$ from a fixed concept class $\mathfrak{C}$, and the goal is to output a hypothesis $h$ that well-approximates $f$, in which case we say that $\mathfrak{C}$ can be quantumly learned.
Due to the massive success of machine learning and the great potential of quantum computing, quantum learning received much attention over the last two decades~\cite{bshouty1998learning,servedio2004equivalences,AS05,ABG06,AS07,OW16,OW17,AW18,GriloKZ19,arunachalamtwo} (cf. survey~\cite{arunachalam2017guest} and references therein).%

In this setting, quantum algorithms have two main advantages over their classical counterparts: making queries in superposition, and using quantum computation to process the information obtained from these queries. Note that a large number of negative results about the power of \emph{classical} learning algorithms  do not extend to the quantum setting (e.g., \cite{kharitonov1993cryptographic,naor1999synthesizers}), since the underlying hardness assumptions, based on problems such as factoring and discrete logarithm, break for quantum computations. In addition, in some learning models and for some learning tasks, quantum algorithms are strictly faster than classical algorithms \citep{servedio2004equivalences}, under standard cryptographic assumptions.
While it is possible to rule out polynomial-time or quasi-polynomial time quantum learning algorithms for some concept classes using stronger cryptographic assumptions~\cite{AGS:shallow}, our understanding of the possibilities and limitations of \emph{sub-exponential} time quantum learnability is still limited. This motivates the following fundamental question:

\begin{center}
    \emph{Are quantum speedups for learning expressive concept classes possible?}
\end{center}

Our main result exposes an intrinsic connection between complexity theory and quantum learning theory, showing that obtaining any quantum learner that does marginally better than certain ``trivial'' learners (see \cref{sec:main}) would imply circuit lower bounds for languages computable in $\mathsf{BQE}$, the quantum analogue of $\mathsf{E}$.\footnote{Recall that $\mathsf{E} = \mathsf{DTIME}[2^{O(n)}]$, and analogously $\mathsf{BQE} = \mathsf{BQTIME}[2^{O(n)}]$.} To our knowledge, this is the first result connecting the design of general quantum algorithms to proving lower bounds:

\begin{mainresult}[Informal]
\label{mainresult}
If a class $\mathfrak{C}$ of polynomial-size concepts can be learned under the uniform distribution with membership queries and with error at most $1/2 - \gamma$ in quantum time $o(\gamma^2 \cdot 2^n/n)$, then $\mathsf{BQE} \nsubseteq \mathfrak{C}$.
\end{mainresult}

While it seems extremely unlikely that a large class such as $\mathsf{BQE}$ can be simulated by classical Boolean circuits of polynomial size, showing this and similar results for exponential time classes seems to be out of reach of current techniques. For comparison, the recent $\mathsf{NE} \nsubseteq \mathsf{ACC^0}$ lower bound due to Williams  \cite{DBLP:journals/jacm/Williams14} is widely recognized to be a milestone in complexity theory.

Our result admits two possible interpretations. For a pessimist, it explains the difficulty of designing provably correct quantum learning algorithms for expressive concept classes, given that establishing non-uniform circuit lower bounds is a notoriously difficult problem. 
Contrarily, for an optimist, it indicates a potential path to new lower bounds by exploring the power of quantum computation. For instance, if depth-$2$ threshold circuits of polynomial size can be learned in $o(2^n/n)$ quantum time under the uniform distribution, a new complexity lower bound would follow. %

Our starting point in the proof of this result is a connection of a similar nature between classical learning algorithms and circuit lower bounds due to Oliveira and Santhanam \cite{DBLP:conf/coco/OliveiraS17}. The extension to quantum learning algorithms turns out to require significant technical work. En route to that, we obtain new results concerning local list-decoding for quantum circuits, construct the first $\PRG$ secure against uniform quantum computations, and provide a new general method to establish learning-to-lower-bound connections. We next describe our contributions in more detail.
\subsection{Main result}
\label{sec:main}

Before we proceed to formally state our result, we first discuss the model of quantum learning and provide a brief overview of circuit lower bounds in complexity theory.

\paragraph{Quantum learning.} We consider the standard PAC learning model under the uniform distribution with quantum membership queries. Our main result derives a consequence from the \emph{existence} of learning algorithms, and restricting our model to the uniform distribution and allowing queries only makes it stronger (since learning algorithms are easier to design under these assumptions). Although we consider quantum learning algorithms, we emphasize that our results are concerned with the learnability of (classical) \emph{Boolean functions} $f \colon \{0,1\}^n \to \{0,1\}$ from a fixed concept class $\mathfrak{C}$. 
Here, we use $\mathfrak{C}[s(n)]$ to refer to concepts defined over $n$-bit inputs and of size at most  $s(n)$.

In more detail, a quantum learning algorithm $\mathcal{A}$ for $\mathfrak{C}$ running in time $T$ is described by a uniform sequence of quantum oracle circuits $Q_n$ 
of size at most $T(n)$. We say that $\mathcal{A}$ learns $\mathfrak{C}[s]$ with probability $\delta$ and error $\varepsilon$ if for every $n$ and every  $f\in \mathfrak{C}[s(n)]$, $Q_n$ with oracle access to $f$ outputs with probability at least $\delta$ the description of a (classical) Boolean circuit $C$ such that $\Pr_{x}[C(x) \neq f(x)] \leq \varepsilon$. Note that having oracle access to $f$ means that $Q_n$ can query $f$ in superposition via a unitary map $O_{f}$ whose action on basis states is defined as $O_{f}:\ket{x,b}\rightarrow \ket{x,b\oplus f(x)}$, for every $x\in \01^n$ and $b\in \01$. We stress here that we do not require the hypothesis circuit  $C$ to be from the same class, i.e, the learning algorithm can be improper (also known as representation independent). %
(Our result also holds for learners that output a \emph{quantum circuit} that approximates $f$ on average to error at most $\varepsilon$, and we refer to the body of the paper for details.)

It is instructive to contrast our main result with two learners that are ``trivial'' in a certain sense, i.e., they do not exploit the structure of the concept class. The first one is simply a brute-force (classical) learner that queries all possible inputs of the unknown function, and outputs a hypothesis consisting of its truth-table. This learner can be implemented in time $T = \widetilde{O}(2^n)$ and achieves optimal error parameter $\varepsilon = 0$. On the other hand, the second  algorithm explores the basic fact that any function $f \colon \{0,1\}^n \to \{0,1\}$ can be approximated by a parity function (or its negation) with  advantage $\gamma = \Omega(2^{-n/2})$ (i.e., with error $\varepsilon \leq 1/2 - \Omega(2^{-n/2})$). For this reason, any class~$\mathfrak{C}$ of Boolean concepts can be learned with probability $\delta = \Omega(1)$  and error $\varepsilon \leq 1/2 - \Omega(2^{-n/2})$ using Fourier sampling in time $T = \mathsf{poly}(n)$  (see \cref{app:Fsampling} for details). (We stress that this highly efficient learner is only available in the quantum setting, thanks to Fourier sampling.) 

\paragraph{Complexity Theory.} Establishing (non-uniform) circuit lower bounds is one of the holy grails of the theory of computation. Despite over 50 years of extensive research, we still have a very poor understanding of the limitations of Boolean circuits. While significant progress has been made with respect to very restricted circuit models, such as small-depth circuits with OR, AND, and NOT gates, the power of more expressive circuit classes remains mysterious.  

To give two concrete examples, firstly, we don't know how to rule out that every algorithm running in time $2^{O(n)}$ can be simulated by (classical) Boolean circuits of linear size.
Secondly, as mentioned above, it was not long ago that Williams \citep{DBLP:journals/jacm/Williams14} obtained the first separation between non-deterministic exponential time $\mathsf{NE} = \mathsf{NTIME}[2^{O(n)}]$ and the class $\mathsf{ACC}^0[m]$ of polynomial-size constant-depth circuits consisting of AND, OR, NOT, and MOD$_m$ gates. This illustrates
how difficult it is to understand the power of \emph{non-uniform} computations, even in the setting of exponential time classes such as $\mathsf{E}, \mathsf{NE}$ and $\BQE$.\footnote{Note that existing circuit lower bounds against circuits of size $n^k$, such as the result from \citep{DBLP:journals/siamcomp/Santhanam09} showing that $\mathsf{MA}/1 \nsubseteq \mathsf{SIZE}[n^k]$, do not provide a hard language in $\mathsf{BQE}$. The naive simulation of a language $L \in \mathsf{MA}$ in exponential time results in an algorithm that runs in time $2^{n^c}$ for some constant $c \geq 1$ that can be larger than $k$ (and we can easily diagonalize againts circuits of size $n^k$ in time $2^{n^c}$ when $c > k$). Additionally, we do not consider in this paper the computation of a hard problem with advice, as in the $\mathsf{MA}/1$ lower bound, which requires $1$ bit of non-uniform advice.}

Indeed, while the possibility that a large class such as $\mathsf{BQE}$ can be simulated by classical Boolean circuits of linear size seems unlikely, proving such statements remains notoriously hard. 
Given the scarcity of techniques to establish such 
complexity separations, it is of interest  to obtain new approaches for circuit lower bounds and to connect this question to other research areas in algorithms and~complexity.\\

In the following, we restrict our attention to reasonable classes $\mathfrak{C}$ of circuits %
that can be efficiently simulated by general Boolean circuits and are closed under restrictions, i.e., if  $f$ is computable by $\mathfrak{C}$-circuits of size $s(n)$, then $f$ is computable by general Boolean circuits of size $\mathsf{poly}(s(n))$, and the function obtained by restricting some of the variables of $f$ to constants $0$ and $1$ is also computable by $\mathfrak{C}$-circuits of size $s(n)$. Note that virtually any circuit class of interest, including depth-$2$ threshold circuits, $\mathsf{ACC}^0$, and polynomial-size formulas,  satisfies these~properties.

\begin{maintheorem}[Non-trivial quantum learning algorithms yield non-uniform complexity lower bounds]\label{thm:main} There is a universal constant $\lambda \geq 1$ for which the following holds. Let  $\mathfrak{C}$ be a concept class. Let $\gamma \colon \mathbb{N} \to [0,1/2]$ and  $T:\mathbb{N}\rightarrow \mathbb{N}$ satisfy $\gamma(n) \geq \lambda\cdot 2^{-n/2}$ and $T(n)\leq \gamma(n)^2\cdot 2^n/\lambda n$. Suppose that, for every $k \geq 1$, the class $\mathfrak{C}[n^k]$ can be learned in quantum time $T(n)$ with probability~$\geq~1/100$ and error $\eps \leq 1/2 - \gamma(n)$. Then, for every $k \geq 1$, we have $\mathsf{BQE} \nsubseteq \mathfrak{C}[n^k]$.
\end{maintheorem}

The confidence probability $1/100$ is not essential, and we adopted this constant in order to simplify the statement and focus on the trade-off between running time and accuracy. Two interesting settings of parameters are (1) $\gamma(n) = 0.49$ and $T(n) = o(2^n/n)$, and (2) $\gamma(n) = n^{\omega(1)} \cdot 2^{-n/2}$ and $T(n) = n^{\omega(1)}$.
The first case shows that strong learners that beat the trivial brute-force learning algorithm by a polynomial factor (with respect to the running time) imply lower bounds, while the second setting shows that polynomial-time learners that perform marginally better than a Fourier-sampling based learner (with respect to the error parameter) also imply lower bounds. For this reason, we view \cref{thm:main} as a result essentially stating that \emph{non-trivial quantum learnability of a class of polynomial-size circuits yields complexity lower bounds}.

As alluded to above, the connection established in \cref{thm:main} can be interpreted in two ways. On the one hand, it provides an explanation for why it is difficult to design provably correct
non-trivial quantum learners, as they would imply dramatic consequences to complexity theory, showing new circuit lower bounds that are notoriously hard to prove. On the other hand, this connection significantly strengthens the paradigm of proving circuit lower bounds via (classical) learning algorithms \cite{DBLP:conf/coco/OliveiraS17} by capitalizing on the power of quantum learning algorithms, which might be vastly stronger than their classical counterparts.\footnote{Even if in the short term we are not able to design new quantum learning algorithms, the mere existence of a \emph{connection} between learning and lower bounds has been used to establish \emph{unconditional} complexity lower bounds (see Section \ref{sec:related_work} and \citep[Section 3.1]{DBLP:conf/icalp/Oliveira19}). Thus, if the ``pessimist'' interpretation is the correct one, it is still possible that the connection established in Theorem \ref{thm:main} can be indirectly used as a key ingredient of a lower bound proof.}

In the subsequent section, we discuss our techniques in detail and explain our additional contributions.

\subsection{Techniques}
\label{sec:tech}
The first results showing that learning algorithms imply lower bounds appeared in the pioneering work of Fortnow and Klivans \citep{DBLP:journals/jcss/FortnowK09}. These initial results, however, required a strong assumption on the resources of the learning algorithm. For instance, for randomized learners, lower bound consequences could only be obtained from \emph{polynomial-time} learners. 
Different learning assumptions were explored in a sequence of subsequent works \citep{DBLP:journals/toct/HarkinsH13, DBLP:conf/coco/KlivansKO13, DBLP:conf/icalp/Volkovich14, DBLP:conf/colt/Volkovich16, DBLP:conf/coco/OliveiraS17, DBLP:conf/approx/OliveiraS18, CRTY20}. We review these works in detail in \cref{sec:related_work}, where we present a self-contained exposition of existing connections between learning algorithms and lower bounds. Here we focus on one of the strongest (and most relevant in our context) connections obtained before this work. 

Oliveira and Santhanam \citep[Theorem 14]{DBLP:conf/coco/OliveiraS17} showed that if a class $\mathcal{C}$ of polynomial-size concepts can be PAC-learned under the uniform distribution to error $\varepsilon \leq 1/2 - \gamma$ by a \emph{randomized} algorithm running in time $\gamma^2 \cdot 2^n/n^{\omega(1)}$, then for every $k \geq 1$ we have $\mathsf{BPE} \nsubseteq \mathcal{C}[n^k]$. In contrast, \cref{thm:main} can be seen as an analogue of this connection in the \emph{quantum} setting, where algorithms can be  more powerful. Intuitively, this means that the task of extracting a computational lower bound from a learning algorithm becomes more delicate. (Indeed, as mentioned above, there is even another ``trivial'' learner that is only available in the quantum setting and proceeds via Fourier sampling.) To accomplish that, we make additional conceptual and technical contributions of independent~interest:

\begin{itemize}
  \setlength{\itemsep}{1pt}
    \item[--] Our paper introduces a \emph{new paradigm} to establish learning-to-lower-bound connections that employs a pseudorandom generator ($\PRG$) in a black-box way. Thus, even without its quantum counterpart, we simplify and extend existing results. 
    \item[--] We propose and analyze the \emph{first $\PRG$} with sub-exponential stretch that is secure against uniform quantum computations. Moreover, we base its security on a weaker \emph{uniform} hardness~assumption. 
    \item[--] We prove a near-optimal uniform hardness amplification result for quantum circuits by a delicate extension of techniques and analysis from Impagliazzo, Jaiswal, Kabanets, and Wigderson \citep{impagliazzo2010uniform}.
    \item[--] We introduce a \emph{new computational model} between classical and quantum computation: \emph{inherently probabilistic~circuits}. It highlights an important difference between classical and quantum algorithms, and provides a way to mitigate the complex task of analyzing quantum computations. 
\end{itemize}
In the next sections, we provide more details about the proof of Theorem \ref{thm:main}, explain the role of the contributions mentioned above, and contrast our techniques to  prior work.

\subsubsection{The classical proof and our new perspective}\label{sec:classical_vs_quantum}
Before discussing the quantum perspective, it is instructive to review the approach for showing the classical connection between randomized learners and lower bounds.

\paragraph{Randomized learners.} 

Techniques from computational complexity theory and from the theory of pseudorandomness play a key role in the proof from \citep{DBLP:conf/coco/OliveiraS17} that non-trivial randomized learners yield complexity lower bounds.  In a bit more detail, their argument proceeds roughly as follows: 
\begin{itemize}
    \item[1.] First, they show that \emph{sub-exponential} time randomized learners for a class $\mathfrak{C}$ imply that $\mathsf{BPE} \nsubseteq \mathfrak{C}$. This part refines ideas from \citep{DBLP:journals/jcss/FortnowK09, DBLP:conf/coco/KlivansKO13} that rely on results from structural complexity theory (e.g.,~a special $\mathsf{PSPACE}$-complete language \citep{DBLP:journals/cc/TrevisanV07} and diagonalization).
    \item[2.] This is followed by a proof that the existence of ``\emph{non-trivial}'' randomized learners for a class $\mathfrak{C}$ of polynomial-size concepts implies the existence of \emph{sub-exponential} time randomized learners for $\mathfrak{C}$. This implication relies in part on connections between learning theory and the theory of pseudorandomness introduced by \citep{DBLP:conf/coco/CarmosinoIKK16}.
\end{itemize}
The proof of the learning-to-lower-bound connection  immediately follows from Items 1 and 2 above.

\paragraph{The quantum case: A new perspective.} In contrast to \citep{DBLP:conf/coco/OliveiraS17}, here we take a more direct path to show that ``non-trivial'' quantum learning algorithms for $\mathfrak{C}$ imply lower bounds against $\mathfrak{C}$. While several technical tools from the theory of pseudorandomness that we use still correspond to quantum extensions of core results behind \citep{DBLP:conf/coco/OliveiraS17}, our approach is conceptually very different. In more detail, we are able to show that $\PRG$s against \emph{uniform} (classical or quantum) computations can be used to establish a learning-to-lower-bound connection in a \emph{black-box} way. In particular, we do not follow the 2-step approach outlined above. 

The benefits of our new perspective are twofold: (a) on the one hand, stronger $\PRG$ statements immediately lead to stronger connections between learning algorithms and lower bounds, and (b) it allows for a more modular proof of the learning-to-lower-bound connection. In particular, with our perspective the argument becomes more manageable in the quantum setting, where new technical difficulties are present compared with the randomized case.

At a very high level, we use a $\PRG$ to generate a ``hard'' function that is not correctly learned by the quantum learning algorithm. Consequently, this function does not belong to the circuit class~$\mathfrak{C}$, and it can be used to define a language that cannot be computed by $\mathfrak{C}$-circuits of bounded size.  What makes the approach viable is that a $\PRG$ that fools \emph{uniform} computations is sufficient. We discuss our proof in more detail in \cref{sec:overview_proof}. Before that, at a more conceptual level, we explain some of the challenges associated with the transition from randomized  to  quantum~computations. 

\subsubsection{Challenges in the quantum setting}\label{sec:quantum_challenges}

Our goal in Theorem \ref{thm:main} is to show that the existence of a quantum learning algorithm for a class $\mathfrak{C}$ of polynomial-size circuits can be used to construct a function $h \in \mathsf{BQE}$ that cannot be computed by $\mathfrak{C}$-circuits of size $n^k$. As mentioned above, and explained in more detail in Section \ref{sec:overview_proof} below, this will be achieved through the design of a $\PRG$ against uniform quantum computations, along with other ideas. In more detail, we show that if a certain language $L$ is hard for sub-exponential time uniform quantum computations, i.e., $L \notin \mathsf{BQTIME}[2^{n^{\gamma}}]$ for some $\gamma > 0$, then there is a generator $G \colon \{0,1\}^{\ell} \to \{0,1\}^m$ computable in deterministic time~$2^{O(\ell)}$, where $\ell = \mathsf{poly}(n)$ and $m = 2^{n^{\Omega(1)}}$, that is able to ``fool'' uniform quantum circuits of size $\mathsf{poly}(m)$.

Generators of this form are know in the realm of classical computations (see \citep{DBLP:journals/jcss/ImpagliazzoW01, DBLP:journals/cc/TrevisanV07, CRTY20}), i.e., under an appropriate hardness assumption against $\mathsf{BPTME}[2^{n^\gamma}]$, it is possible to fool uniform probabilistic computations running in time $\mathsf{poly}(m)$. The proof of these results, essentially, proceeds as follows. If there is a sequence of (uniform) circuits $\{C_m\}$ of bounded complexity that distinguish the $m$-bit output of $G$ from a random $m$-bit string, then there is a faster uniform algorithm for the hard problem $L$, which leads to a contradiction. 
In order to quantize this argument, we aim to construct a $\PRG$ secure against uniform \emph{quantum} circuits. This leads to the natural question:
    \emph{What can go wrong in these classical proofs when each $C_m$ is a quantum circuit?} 

\paragraph{Randomized circuits versus quantum circuits.} It turns out that there is a general property of probabilistic computations that is not available to quantum computations. To be precise, consider the standard model of randomized Boolean circuits, i.e., a Boolean circuit $C(x,y)$, where $x$ is the input string and $y$ is the random string. We say that~$C$ computes a Boolean function $g \colon \{0,1\}^m \to \{0,1\}^k$ on an input $x \in \{0,1\}^m$ if $\Pr_y[C(x,y) = g(x)] \geq~2/3$. Similarly, we can also discuss the correlation between $C$ and $g$, captured by  $\eta = \Pr_{x,y}[C(x,y) = g(x)]$. An extremely useful property of this model is that, by an averaging argument, there exists a fixed string $y'$ such that $\Pr_x[C(x,y') = g(x)] \geq \eta$, i.e., there is a \emph{deterministic} circuit $C_{y'}(x) = C(x,y')$ that correctly computes $g$ on an $\eta$-fraction of inputs. This often allows one to reduce the analysis of randomized Boolean circuits to the deterministic case. Additionally, $C_{y'}$ ``forces'' $C(x,y)$ to commit to a single consistent output on every $x$, which can be relevant if $C$ is used as a subroutine in other computations.

As a concrete example, note that this idea is crucial in the proof of $\mathsf{BPP} \subseteq \mathsf{P}/\mathsf{poly}$, which is a simple combination of  reducing the failure probability of a randomized circuit and fixing its random input. (In contrast, it is open if $\mathsf{BQP}  \subseteq \mathsf{P}/\mathsf{poly}$.) Importantly, while quantum computations share superficial similarities with the model of randomized computations, 
there is no obvious way of reducing a quantum computation to a distribution of ``deterministic'' quantum computations. This has an important effect on the design of algorithms as well as on their analysis, creating several difficulties in our proof. We explain one of the most significant of these below. 

\paragraph{Uniform hardness amplification for quantum circuits.} A core component in several  $\PRG$~constructions is to amplify the hardness of a function $f \colon \{0,1\}^n \to \{0,1\}$ that is only assumed to be mildly hard. This is a standard element of $\PRG$s that follow the hardness versus randomness~framework of Nisan and Wigderson \citep{nisan1994hardness}, and we need it here as well. To be precise, suppose we start with  $f$  that is weakly hard for quantum circuits of bounded size, i.e., for every circuit $\mathcal{A}$ we have
$$
\E_{x \sim \{0,1\}^n}\Big [\Pr_\mathcal{A}\big[\mathcal{A}(x) = f(x)\big]\Big]  \eqdef \E_{x \sim \{0,1\}^n}\Big[ \norm{\Pi_{f(x)}\mathcal{A}\ket{x,0^{\ell}}}^2\Big ] \;\leq\; 1 - \delta,
$$
where $\delta = 1/n$ for instance.
In the notation above, $\norm{\Pi_{f(x)}\mathcal{A}\ket{x,0^{\ell}}}^2$ is the probability of obtaining~$f(x)$ when measuring the first qubit of the circuit $\mathcal{A}$ on input $\ket{x,0^{\ell}}$, and we average its success probability over a random input $x$. Our goal is to define a function $g \colon \{0,1\}^{n'} \to \{0,1\}^k$ from $f$ that is exponentially harder than $f$, i.e., for every quantum circuit $\mathcal{B}$ of bounded size, we have
$$
\E_{y \sim \{0,1\}^{n'}}\Big [\Pr_\mathcal{B}\big[\mathcal{B}(y) = g(y)\big]\Big]  \eqdef \E_{y \sim \{0,1\}^{n'}}\Big[ \norm{\Pi_{g(y)}\mathcal{B}\ket{y,0^{\ell'}}}^2\Big ] \;\leq\; \varepsilon,
$$
where $\varepsilon = 2^{-n^{\Omega(1)}}$ for instance.  We have two additional requirements beyond the generalization from classical to quantum computation: (1) we need a hardness amplification scheme with a near-optimal setting of parameters (so that our $\PRG$ achieves good parameters), and (2) since we are constructing a $\PRG$ under a \emph{uniform} hardness assumption, it is also crucial to control the amount of ``non-uniformity'' in the argument that establishes the correctness of the construction (i.e., how a circuit $\mathcal{B}$ violating the conclusion can be used to construct a circuit $\mathcal{A}$ that violates the hypothesis).

Our main technical contribution is to prove a quantum analogue of the near-optimal (uniform) hardness amplification result from Impagliazzo, Jaiswal, Kabanets, and Wigderson \citep{impagliazzo2010uniform} which, in the language of coding theory, can be interpreted as a local-list decoding algorithm for direct product codes. In other words, their work considers $g = f^k$, i.e., $g \colon \{0,1\}^{n'} \to \{0,1\}^k$ is defined as the computation of $k$ independent copies of $f$ (thus $n' = kn$, and one can think of $k = \mathsf{poly}(n)$). 

The extension of \citep{impagliazzo2010uniform} to quantum circuits is challenging for the following reason. In the classical setting, given (without loss of generality) a \emph{deterministic} circuit $B$ such that
\begin{equation}\label{eq:intro_IJKW}
    \Pr_{y \sim \{0,1\}^{n'}}[B(y) = g(y)] > \varepsilon,
\end{equation}
there is only one way to distribute the correlation between $B$ and $g$: there is a set $S \subseteq \{0,1\}^{n'}$ such that $B$ is correct on $S$ and wrong on $\overline{S}$, where $|S| > \varepsilon \cdot 2^{n'}$. On the other hand, for a quantum circuit $\mathcal{B}$, the correlation between $\mathcal{B}$ and $g$ can be distributed in an arbitrary way among the inputs. For instance, perhaps on every input string $y \in \{0,1\}^{n'}$, $\Pr_{\mathcal{B}}[\mathcal{B}(y) = g(y)] = 2 \varepsilon$. Going beyond that, we might also have quantum circuits that ``interpolate'' between these two extreme examples in an arbitrary way. Unfortunately, as opposed to randomized circuits, there is no way for us to ``fix the quantumness'' of $\mathcal{B}$ and reduce the analysis to the deterministic case, i.e., to the simpler setting of Equation (\ref{eq:intro_IJKW}). (For instance, one could query $\mathcal{B}$ on each input only once, memorizing its output after that. This, however, would not provide a \emph{fixed} hard language if we employ $\mathcal{B}$ as a sub-routine when extracting lower bounds from quantum learning algorithms. Hence this trick would not work.) This more general scenario in the quantum setting affects the original analysis from \citep{impagliazzo2010uniform} and introduces several fundamental ``asymmetries'' in the argument.

\paragraph{An inherently random ``channel'' and a quantization of \citep{impagliazzo2010uniform}.} From the perspective of coding theory, what we are trying to achieve is local list-decoding in a setting where querying a coordinate (an input for $\mathcal{B}$) does not produce a deterministic outcome.
It turns out that this problem can be investigated without reference to quantum circuits. To achieve that, we explore a computational model that captures the scenario described above: \emph{inherently probabilistic circuits} (see Section \ref{ss:inherently_probabilistic} for the definition and discussion of related work). Roughly speaking, this is a model for randomized computations where ``the random input cannot be fixed'', i.e., the random input remains inaccessible and the randomized circuit is only accessed as a black-box. It can be seen as an intermediate model between classical randomized circuits and quantum circuits. Most importantly for us, this model captures nearly all the difficulties when quantizing building blocks of the~$\PRG$.

Employing an intricate analysis that extends~\citep{impagliazzo2010uniform}, we are able to show that their local list decoding algorithm (with a natural modification) works in the setting of inherently probabilistic circuits, with a minor loss in the parameters. The result can then be translated to the setting of quantum computations without much effort, which provides the near-optimal uniform hardness amplification needed in our $\PRG$ construction.
We refer the reader to Section \ref{sec:IJKW} for the technical details. Having established the above, we are ready for a high-level proof overview of Theorem~\ref{thm:main}.

\subsubsection{Overview of the proof of Theorem \ref{thm:main}}\label{sec:overview_proof}

Let $\mathfrak{C}$ be a circuit class, and suppose that polynomial-size concepts from $\mathfrak{C}$ can be learned with error parameter $\varepsilon \leq 1/2 - \gamma$ in quantum time $o(\gamma^2 \cdot 2^n/n)$. For a given $k \geq 1$, we argue that there is a language  $L \in \mathsf{BQE} = \mathsf{BQTIME}[2^{O(n)}]$ such that $L \notin \mathfrak{C}[n^k]$. 

In contrast with known proofs of existing learning-to-lower-bound connections, our black-box ``$\PRG$-based'' approach relies on the following basic principle from the theory of pseudorandomness: If we have a property $\mathcal{P}$ of objects that is ``dense'' in the set of all possible objects, and a $\PRG$ $G$ that is able to ``fool'' a class of ``tests'' that contain $\mathcal{P}$ (e.g., when $\mathcal{P}$ is easy to compute and $G$ fools predicates of bounded complexity), then some output of $G$ must be an object that satisfies $\mathcal{P}$. For us, $\mathcal{P}$ is simply the class of functions that are not in $\mathfrak{C}[n^k]$. We would like to find a (fixed) function in $\mathcal{P}$ and compute according to it in $\mathsf{BQTIME}[2^{O(n)}]$.  We start off with the ``ideal'' plan for the proof, then discuss difficulties when implementing this strategy and how to overcome them.\\

\noindent\textbf{Ideal plan for the proof of Theorem \ref{thm:main}:}
\begin{itemize}
  \setlength{\itemsep}{-1pt}
    \item [1.] From a quantum learning algorithm $\mathcal{A}$ that \emph{finds} ``structure'', we get a (uniform) quantum algorithm $\mathcal{B}$ that \emph{decides} the ``absence of structure'' on an input string $f \in \{0,1\}^{2^n}$. (Since a typical \emph{random} string is ``unstructured'', a \emph{pseudorandom} string is likely to be accepted by~$\mathcal{B}$.)
    \item [2.] Assuming that there exists a $\PRG$ $G \colon \{0,1\}^{O(n)} \to \{0,1\}^{2^n}$ secure against uniform quantum computations, we use it with $\mathcal{B}$ as a ``test'' to find a function $h_n \colon \{0,1\}^n \to \{0,1\}$ that is not in $\mathfrak{C}[n^k]$, i.e., a function that ``lacks structure''. Given an $n$-bit input $x$, the hard language $L$ outputs $h_n(x)$.
    \item [3.] We show that a pseudorandom generator $G$ with the desired parameters and guarantees exists.
\end{itemize}

\noindent Intuitively, if we could implement these steps then $L$ would indeed be a hard function, and the time complexity of $L$ can be bounded using upper bounds on the complexities of computing $G$ and $\mathcal{B}$.

\paragraph{Implementing Step 1: (promise) quantum natural properties.} In order to formalize Step~1, we introduce the notion of (promise) \emph{quantum natural properties}, a generalization of the central concept from \citep{DBLP:journals/jcss/RazborovR97} to quantum computations. Informally, a natural property (useful)  against a class $\mathfrak{C}$ of functions is an algorithm $B$ that: (a) rejects every function $g \in \mathfrak{C}$, when $g$ is viewed as a $2^n$-bit string; (b) the set of strings accepted by $B$ is dense in $\{0,1\}^{2^n}$; and (c) $B$ runs in time $\mathsf{poly}(2^n) = 2^{O(n)}$ on an input string $f \in \{0,1\}^{2^n}$. Due to these properties, algorithm $B$ can be seen as a way to efficiently tell apart ``structured'' strings (those that encode functions from~$\mathfrak{C}$) from a dense set of ``unstructured'' strings. A \emph{quantum} natural property  against~$\mathfrak{C}$ is simply a natural property against~$\mathfrak{C}$ that is decided by a quantum algorithm. To our knowledge, this is the first time that quantum natural properties are considered in the literature. 

We show that quantum learners with parameters as in Theorem \ref{thm:main} imply the existence of quantum natural properties against $\mathfrak{C}[n^k]$. The argument follows an idea that has appeared in a few different works (e.g.,~\citep{DBLP:journals/eccc/Oliveira13, DBLP:conf/icalp/Volkovich14, DBLP:conf/coco/OliveiraS17,AGS:shallow}): to test if an input string $f \in \{0,1\}^{2^n}$ is not in  $\mathfrak{C}[n^k]$, one can simulate the learning algorithm $\mathcal{A}$ when its oracle computes according to $f$ (now viewed as function), and accepts if the hypothesis output by $\mathcal{A}$ and $f$ have agreement estimated to be less than $1/2 + \gamma$. For a function $f \in \mathfrak{C}[n^k]$, when~$\mathcal{A}$ learns $f$ its hypothesis will have a larger correlation with $f$, and so $f$ is rejected. On the other hand, it is possible to prove that if~$\mathcal{A}$ runs in quantum time $o(\gamma^2 \cdot 2^n/n)$, then a dense set of Boolean functions are not learned by $\mathcal{A}$ even with a large error $\varepsilon = 1/2 - \gamma$ (simply because such functions cannot be approximated by hypotheses of size $o(\gamma^2 \cdot 2^n/n)$). Consequently, any such function will pass the low-correlation test with high probability and be accepted, no matter how many times the learner is simulated. 

We directly extend this idea to the quantum setting in~\Cref{sec:qnp-from-ql} while also addressing a subtlety that arises in the case of randomized and quantum learners. The resulting quantum algorithm $\mathcal{B}$ computing a natural property against $\mathfrak{C}[n^k]$ might accept some input strings $f \in \{0,1\}^{2^n}$ with a probability that is not bounded away from $1/2$. This happens because we cannot control the behavior of the algorithm $\mathcal{A}$ on functions near the ``border'' of $\mathfrak{C}[n^k]$ (i.e., those inputs that are not as hard as a ``random'' function but are not in $\mathfrak{C}[n^k]$). Hence, we only get a \emph{promise} quantum natural property: we are guaranteed that strings from the dense set are accepted with high probability and strings corresponding to functions in $\mathfrak{C}[n^k]$ are rejected with high probability.

\paragraph{Challenges in implementing Step 2.} Unfortunately, the strategy described in Step~2 is problematic for a number of reasons:  
\begin{itemize}
  \setlength{\itemsep}{1pt}
    \item [(\emph{i})] The $\PRG$ $G$ constructed in Step 3 requires a hardness assumption, while Step 2 needs a language $L$ that is \emph{unconditionally} not in $\mathfrak{C}[n^k]$. This can only be achieved if $G$ provably works (i.e.~it does not depend on an unproven assumption).
    \item [(\emph{ii})] Since $\mathcal{B}$ only computes a \emph{promise} quantum natural property, it is not immediately clear how to use the output of $G$ and the test $\mathcal{B}$ to fix with high probability a \emph{unique} ``hard'' function $h_n$. This is needed so the language $L$ is well defined.
    \item [(\emph{iii})] The (conditional) $\PRG$ $G$ that we are able to construct in Step 3 is much weaker: it will stretch $\mathsf{poly}(n)$ bits to $m(n) = 2^{n^\lambda}$ bits for some $\lambda > 0$, and it is only guaranteed to fool uniform quantum computations of time $\mathsf{poly}(m)$ on infinitely many choices of $n$.
\end{itemize}
These issues create technical difficulties that are addressed in~\Cref{sec:qnp-lb} with some modifications to our original plan, as we explain next.\\ %

\noindent \textbf{Issue (\emph{i}).} To resolve this and relax the conditions on the $\PRG$, we consider two possible~scenarios:
\begin{itemize}
  \setlength{\itemsep}{1pt}
    \item [(\emph{a})] Classical computations performed in polynomial space can be simulated by quantum algorithms running in sub-exponential time $2^{n^{o(1)}}$.
    \item [(\emph{b})] Item (\emph{a}) does not hold, i.e., there is a language $L \in \mathsf{PSPACE}$ and $\alpha>0$ such that  $L \notin~\mathsf{BQTIME}[2^{n^{\alpha}}]$.  
\end{itemize}
While we cannot currently decide which of the two scenarios holds, we obtain lower bounds against $\mathfrak{C}[n^k]$ in \emph{each scenario}, which is sufficient for the purpose of proving Theorem~\ref{thm:main}.
(We note that such ``win-win'' arguments have appeared in many previous works, including \citep{DBLP:journals/jcss/FortnowK09, DBLP:conf/coco/KlivansKO13, DBLP:conf/coco/OliveiraS17}.) In more detail, if (\emph{a}) holds, we employ standard diagonalization techniques from complexity theory to get a language $L \in \mathsf{PSPACE} \setminus \mathfrak{C}[n^k]$, then show that $L \in \mathsf{BQE}$ via a sub-exponential time quantum simulation that is granted to exist by assumption (\emph{a}). We can therefore assume for the rest of the proof that (\emph{b}) holds. In other words, we have a hardness hypothesis against quantum computations that we can hope to use to construct a pseudorandom generator.\\

\noindent \textbf{Issue (\emph{ii}).} 
We fix this issue as follows. Suppose, for simplicity, that we did manage to construct a $\PRG$ $G \colon \{0,1\}^{O(n)} \to \{0,1\}^{2^n}$ of exponential stretch that is computable in time exponential in $n$. 
A $2^n$ bit string output by $G$ can be seen as the truth table for a function $h_n$. Since each seed for $G$ produces a corresponding function, $G$ encodes a collection of functions. What we know from the pseudorandomness of $G$ and the existence of the promise quantum natural property is that at least one of these functions must lie outside $\mathfrak{C}[n^k]$. 
Then, the hard language $L$ could be defined by, say, the ``first" of the $h_n$ that lies outside of $\mathfrak{C}[n^k]$ -- the set of ``easy functions". However, finding the $h_n$ that satisfies this condition in $\BQTIME[2^{O(n)}]$ remains problematic.\footnote{This is because $\mathcal{B}$ computes a \emph{promise} natural property, and we cannot easily estimate the exact probability that $\mathcal{B}$ accepts a string to consistently compute the ``first'' such string.}
Instead, we define a language $L$ over $O(n) + n$ input bits that \emph{simultaneously} computes according to all output functions of $G$. 
More precisely, the hard language $L$ is defined over a pair of input strings $w$ and $x$, where $w$ is a seed for $G$ of length $O(n)$, and $x$ is an $n$-bit string. We then set $L(w, x) = f_w(x)$, where $f_w$ is the function defined by $G(w) \in \{0,1\}^{2^n}$. 
As explained above, given that the $\PRG$ $G$ is secure against the uniform quantum computation performed by $\mathcal{B}$, it must produce at least one string $y = G(w)$ encoding a function $h_n$  which lies outside the set $\mathfrak{C}[n^k]$. Since $L$ computes $h_n$ when we fix its first input string to its corresponding string $w$, $L$ cannot be computed by $\mathfrak{C}$-circuits of size $n^k$.\\

\noindent \textbf{Issue (\emph{iii}).} Finally, this issue is handled using a more careful description of the hard language $L$. The proof makes use of the fact that we are able to get a quantum natural property against $\mathfrak{C}[n^d]$ for each fixed choice of $d$, since by assumption we have learning algorithms for arbitrarily large polynomial-size concepts from $\mathfrak{C}$.
Thus, a potential loss in the stretch of the generator (from $2^n$ to just $2^{n^{\lambda}}$ for some $\lambda > 0$) can be compensated by considering a natural property against harder functions, together with standard translation arguments from complexity theory. The weaker infinitely often guarantee of the generator can also be addressed with a careful implementation, and we refer to the formal proof for these technical details. We  highlight that the construction of better (conditional) pseudorandom generators immediately leads to a tighter connection between the circuit size in the lower bound and the circuit size in the learning assumption.\footnote{Indeed, this has happened in practice, where techniques employed to design a better $\PRG$ also led to a new learning-to-lower-bound connection \citep{CRTY20}. Our work shows that this is not a coincidence, i.e., better $\PRG$s lead to better connections in a black-box way.} 

\paragraph{Revised Step 3: (conditional) $\PRG$ against uniform quantum computations.} 
This leaves us with the last and most technical part of the proof: the construction of a $\PRG$ of sub-exponential stretch $2^{n^{\Omega(1)}}$ that fools uniform quantum computations infinitely often, assuming $\mathsf{PSPACE} \nsubseteq \mathsf{BQSUBEXP}$ (obtained from Item (\emph{b}) above). In Section \ref{sec:classical_vs_quantum}, we explained that establishing this result in the quantum setting is more delicate than in the classical scenario. Here we focus on the different components employed in the proof and how they fit together in the $\PRG$ construction.

First, we stress that our $\PRG$ $G_n \colon \{0,1\}^{\ell(n)} \to \{0,1\}^{m(n)}$ is defined \emph{classically}, i.e., it is computed by a \emph{deterministic} algorithm running in time exponential in $\ell$. The extension to 
quantum circuits lies in its security analysis, where we  argue that if there is a sequence $\{C_m\}$~of uniform quantum circuits of size $\mathsf{poly}(m)$ that distinguish the output of $G$ from random, then $\mathsf{PSPACE} \subseteq \mathsf{BQSUBEXP}$.

As alluded to above, we employ the hardness versus randomness paradigm. More precisely, our construction relies on the beautiful insight of Impagliazzo and Wigderson \citep{DBLP:journals/jcss/ImpagliazzoW01} on how to extend this paradigm to the \emph{uniform} case, where the non-uniform advice appearing in the security proof can be eliminated via a clever recursive approach. This allows one to prove security based on a \emph{uniform} hardness assumption. Following \citep{DBLP:journals/jcss/ImpagliazzoW01}, in order to achieve this we also base the generator  on a problem that is \emph{downward-self-reducible} and \emph{self-correctable} (such properties are useful to eliminate advice). However, we deviate from their work in terms of some other technical tools and parameters of our $\PRG$ construction, which is formally shown in~\Cref{sec:prg-construction}.

In more detail, we construct a family of generators $G_n^\lambda$, each parameterized by a fixed $\lambda > 0$. Each $G_n^\lambda$ employs a special $\mathsf{PSPACE}$-complete language $L^\star$ on $n$ input bits from Trevisan and Vadhan \citep{DBLP:journals/cc/TrevisanV07} (which has the self-reducibiliy properties cited above),  and applies the well-known Nisan-Wigderson generator \citep{nisan1994hardness} with sub-exponential stretch $m(n) = 2^{n^{\lambda}}$ to an ``amplified'' version  $\mathsf{Amp}(L^\star)$ of $L^\star$ defined over $\mathsf{poly}(n)$ input bits. $\mathsf{Amp}(L^\star)$ is obtained from $L^\star$ as follows. First, we consider its $k$-product $(L^\star)^k \colon \{0,1\}^{kn} \to \{0,1\}^k$, which is discussed in Section \ref{sec:quantum_challenges} in the context of hardness amplification and a result from \citep{impagliazzo2010uniform}. Then, we convert $(L^\star)^k$ into a Boolean function by a standard application of the Goldreich-Levin construction \citep{DBLP:conf/stoc/GoldreichL89}, which XORs the output of $(L^\star)^k$ with a new input string $r \in \{0,1\}^k$. This completes the sketch of the definition of $G_n^\lambda$. By a careful choice of parameters $k = \mathsf{poly}(n)$ and seed length $\ell(n) = \mathsf{poly}(n)$, it is not hard to prove that $G_n^\lambda \colon \{0,1\}^{\ell(n)} \to \{0,1\}^{m(n)}$ can be computed in time $2^{O(\ell(n))}$. 

The non-trivial aspect here is to prove that $G_n^\lambda$ is quantum secure for some choice of $\lambda$. It is enough to argue that, if for every $\lambda > 0$ the generator $G_n^\lambda$ can be broken, then $L^\star \in \mathsf{BQTIME}[2^{n^{o(1)}}]$. Since this language is $\mathsf{PSPACE}$-complete, this violates our  uniform hardness assumption. So we focus next on a fixed $\lambda > 0$ and its corresponding generator  $G_n^\lambda$, and assume that there is a uniform sequence of quantum circuits $\{C_m\}$ of size $\mathsf{poly}(m)$ that distinguishes its output from random. Our goal is to conclude that $L^\star \in \mathsf{BQTIME}[2^{n^{O(\lambda)}}]$. %
To achieve this, we need \emph{quantum} analogues~of the ``reconstruction'' procedures of~\citep{nisan1994hardness,DBLP:conf/stoc/GoldreichL89,impagliazzo2010uniform} implemented in a uniform way as in~\citep{DBLP:journals/jcss/ImpagliazzoW01}.\\

\noindent \emph{Quantum Nisan-Wigderson reconstruction.} 
We observe that the original analysis of~\cite{nisan1994hardness} can be adapted without difficulty in our setting. In~\Cref{sec:nw}, we first verify the result in the intermediate model of inherently probabilistic circuits and then extend it to the quantum case.\\

\noindent  \emph{Quantum Goldreich-Levin reconstruction.} 
We observe that a quantum analogue of~\cite{DBLP:conf/stoc/GoldreichL89}~was established by Adcock and Cleve \citep{adcock2002quantum}. In~\Cref{sec:gl}, we adapt their argument to match our~notation and setting of parameters. For this reason, we do not use inherently probabilistic computations here, which are convenient in the investigation of the other building blocks of our~$\PRG$~construction.\\

\noindent  \emph{Quantum Impagliazzo-Jaiswal-Kabanets-Wigderson reconstruction.} As explained above, this turns out to be the most complicated aspect of the security proof, since it is non-trivial to extend the original analysis from \citep{impagliazzo2010uniform} to the quantum setting. We refer to Section \ref{sec:quantum_challenges} above for an informal explanation, and to Section \ref{sec:IJKW} for more details. We remark that it is not hard to quantize the well-known XOR Lemma for hardness amplification (e.g., Impagliazzo's proof \citep{impagliazzo1995hard}), but in this work we need a \emph{uniform} hardness amplification result with near-optimal parameters.\\

It remains to put together these tools to conclude that the existence of a quantum distinguisher~$\{C_m\}$ implies that $L^\star \in \mathsf{BQTIME}[2^{n^{O(\lambda)}}]$. As explained above, this is done in a uniform way following the approach of \citep{DBLP:journals/jcss/ImpagliazzoW01}. However, controlling the \emph{uniformity} of the final sequence of quantum circuits that compute $L^\star$ is  delicate. Recall that to show that $L^\star \in \mathsf{BQTIME}[T]$ we need to produce in uniform \emph{deterministic} time $T$, the description of a quantum circuit $Q_n$ for $L^\star$ on inputs of length $n$. In the recursive construction of \citep{DBLP:journals/jcss/ImpagliazzoW01} that provides an algorithm to compute $L^\star$ on an input string $x \in \{0,1\}^n$, one ``learns'' how to compute $L^\star$ on every input length from $1$ to $n$, by producing a sequence of \emph{randomized} circuits $D_1, \ldots, D_n$ that compute $L \cap \{0,1\}^i$ for $i \in \{1, \ldots, n\}$. In order to compute $D_{i}$ from $D_{i - 1}$, it is necessary to simulate $D_{i - 1}$ on some inputs, which requires \emph{randomness}. Similarly, the natural way to proceed in the quantum case is by generating a uniform sequence of quantum circuits $Q_1, \ldots, Q_n$ as above. Note that simulating $Q_{i - 1}$ to produce $Q_i$ now requires a \emph{quantum} computation. However, we must be able to \emph{deterministically} produce the code of quantum circuits in order to show that $L^\star \in \mathsf{BQTIME}[T]$. We address this issue in~\Cref{sec:self_reduc} by going to another ``meta-level'' in this simulation, where the circuit $Q_i$ ``incorporates''  this recursive process from $i$ to $1$ by manipulating the \emph{codes} of quantum circuits $Q_{i-1}$ to $Q_1$. (Formally, this is not too different from \citep{DBLP:journals/jcss/ImpagliazzoW01}, which also manipulates descriptions of circuits, but in the quantum case one needs to be more explicit.) We observe that it is not hard to implement this idea if the uniform versions of the quantum reconstruction procedures described above are stated in a convenient way. 

This completes the sketch of Step 3, and our overview of the proof of Theorem \ref{thm:main}.

\subsection{Directions and open problems}

The most ambitious direction is to address the possibility that quantum computation might lead to faster learning algorithms for expressive classes of concepts.

\begin{question}
Is there a quantum learning algorithm for Boolean circuits of size $O(n)$ that runs in time $o(2^n/n)$ and with constant probability achieves error $\varepsilon \leq 0.49$ under the uniform distribution?
\end{question}
\noindent Addressing this and related problems might be out of reach given our current techniques. We focus below on directions that we find particularly interesting and possibly fruitful. 

It has been suggested to us that it might also be possible to extend existing algorithms for local-list decoding of Reed-Muller codes to the inherently probabilistic setting. (This would provide an alternative presentation of our PRG construction.) While we have not verified all the details, we believe that this is quite plausible. It would be interesting to understand if there is a more general phenomenon in place here, i.e., whether certain classes of ECCs and decoding algorithms can be extended to the inherently random setting in a generic fashion, and to investigate further applications of such codes.

Our results offer the exciting possibility that new circuit lower bounds might follow through the design of quantum algorithms. Recall that Williams \citep{DBLP:journals/jacm/Williams14} established via the satisfiability-to-lower-bound connection that $\mathsf{NE} \nsubseteq \mathsf{ACC}^0$. Similarly, can we use our new learning-to-lower-bound connection to show that $\mathsf{ACC}^0$ circuits cannot compute all functions in $\mathsf{BQE}$? Using our techniques (cf.~Corollary \ref{cor:lbs-from-qnp}), it would be sufficient to provide a positive answer to the following question.

\begin{question}
Is there a \emph{(}promise\emph{)} quantum natural property useful against $\mathsf{ACC}^0$ circuits?
\end{question}

Note that in order to achieve this it suffices to construct a quantum natural property for the larger class $\mathsf{SYM}^+$ of quasi-polynomial size depth-$2$ circuits consisting of an arbitrary symmetric gate at the top layer fed with $\mathsf{AND}$ gates of poly-logarithmic fan-in at the bottom layer \citep{DBLP:journals/cc/BeigelT94}. Similarly, it is enough to get a quantum natural property for the class of functions that can be approximated by a torus polynomial of bounded degree \citep{DBLP:conf/innovations/BhrushundiHLR19}, or for Boolean matrices of bounded ``symmetric rank'' \citep{DBLP:journals/toc/Williams18}. Perhaps quantum computations can be helpful in the design of natural properties for these or for other circuit classes (e.g.,~Boolean formulas of size $n^{3.01}$)?

We are also curious about the prospects of designing non-trivial quantum learners for restricted circuit classes. Servedio and Tan \citep{DBLP:conf/innovations/ServedioT17} explored this possibility in the context of classical computation, showing several examples where non-trivial savings can be achieved compared to the trivial ``brute-force'' learning algorithm that runs in time $O(2^n)$. As alluded to above, another regime of parameters in the quantum setting that might be interesting to explore is that of polynomial-time learners that achieve a non-trivial advantage $\gamma \gg 2^{-n/2}$. 

Finally, the investigation of the classical learning-to-lower-bound connection established in \citep{DBLP:conf/coco/OliveiraS17} led to many other results, such as the existence of a learning speedup phenomenon \citep[Lemma 1]{DBLP:conf/coco/OliveiraS17} and unconditional complexity lower bounds \citep{DBLP:conf/icalp/Oliveira19}. It would be interesting to see if new consequences in quantum complexity theory and quantum learning theory can be obtained using the techniques introduced in this work.

\subsection{Related work}
\label{sec:related_work}

There is a rich history of connections between circuit complexity theory and the investigation of learning algorithms. In many cases, \emph{specific} circuit lower bound techniques have been used to design new algorithms. For instance, Linial, Mansour, and Nisan \citep{DBLP:journals/jacm/LinialMN93} relied on the method of random restrictions to show that constant-depth polynomial-size circuits can be PAC-learned in quasi-polynomial time under the uniform distribution from random examples. In a more recent development, Carmosino, Impagliazzo, Kabanets, and Kolokolva \citep{DBLP:conf/coco/CarmosinoIKK16} showed that learning algorithms can be obtained from \emph{any} lower bound technique that is ``\emph{constructive}'' in the sense of the theory of natural proofs of Razborov and Rudich \citep{DBLP:journals/jcss/RazborovR97}. This allowed them to show that constant-depth polynomial-size circuits augmented with parity gates can be PAC learned in quasi-polynomial time under the uniform distribution from membership queries. These, and several other results (cf.~Servedio and Tan \citep{DBLP:conf/innovations/ServedioT17}), show that a circuit lower bound against a circuit class $\mathcal{C}$, in most cases, can be converted into a learning algorithm for $\mathcal{C}$. 

The first results in the opposite direction, i.e., showing that learning algorithms imply circuit lower bounds, were established by Fortnow and Klivans \citep{DBLP:journals/jcss/FortnowK09}. In their work, among other results, they proved that any \emph{sub-exponential time} (i.e.,~$2^{n^{o(1)}}$), \emph{deterministic} exact learning algorithm for $\mathcal{C}$ (using  membership and equivalence queries) implies the existence of a function in $\mathsf{E}^{\mathsf{NP}}$ that is not in $\mathcal{C}$. They also showed that if $\mathcal{C}$ is PAC learnable with membership queries under the uniform distribution or exact learnable in \emph{randomized polynomial-time}, there is a function in $\mathsf{BPE}$ (an exponential time analog of $\mathsf{BPP}$) that is not in $\mathcal{C}$.  Note that these results have the appealing feature that they make no assumption about the techniques employed in the design of the learning algorithm. Nevertheless, there is an important drawback in the initial results of \citep{DBLP:journals/jcss/FortnowK09}: they make strong assumptions about the \emph{resources} of the learning algorithm. For example, many learning algorithms are randomized and require at least quasi-polynomial time (e.g.,~\citep{DBLP:journals/jacm/LinialMN93, DBLP:conf/coco/CarmosinoIKK16}), and the results from \citep{DBLP:journals/jcss/FortnowK09} do not apply in this case. 

Over the last decade several works have addressed this and other aspects of the learning-to-lower-bound connection. Harkins and Hitchcock  \citep{DBLP:journals/toct/HarkinsH13} eliminated the $\mathsf{NP}$ oracle and showed that $\mathsf{E} \nsubseteq \mathcal{C}$ from the existence of exact deterministic sub-exponential time learning algorithms with membership and equivalence queries. Shortly after, Klivans, Kothari, and Oliveira \citep{DBLP:conf/coco/KlivansKO13}  simplified these proofs, strengthened a few existing learning-to-lower-bound connections, and obtained results for additional learning models (e.g.,~learning from statistical queries). In particular, \citep{DBLP:conf/coco/KlivansKO13} proved that $\mathsf{E} \nsubseteq \mathcal{C}$ if there are exact learners for $\mathcal{C}$ running in deterministic time $o(2^n)$, showing that ``non-trivial'' \emph{deterministic} learners yield lower bounds.

In a different direction, Volkovich \citep{DBLP:conf/icalp/Volkovich14} showed how to extract lower bounds in polynomial-time classes from randomized learners running in polynomial time. More precisely, \citep{DBLP:conf/icalp/Volkovich14} shows that $\mathsf{BPP}/1$ ($\mathsf{BPP}$ with $1$ bit of non-uniform advice per input length) is not contained in $\mathcal{C}[n^k]$ for all $k\geq 1$ if $\mathcal{C}$ can be PAC learned with membership queries under the uniform distribution by randomized learners running in polynomial time. In another work, Volkovich \citep{DBLP:conf/colt/Volkovich16} explores connections between algebraic complexity lower bounds and learning algorithms.

While extracting circuit lower bounds from non-trivial \emph{deterministic} learners can be done using elementary techniques \citep{DBLP:conf/coco/KlivansKO13}, extending the result to \emph{randomized} learners required advanced tools from complexity theory. This was first achieved by Oliveira and Santhanam \citep{DBLP:conf/coco/OliveiraS17}, who proved that if for every $k \geq 1$ the class $\mathcal{C}[n^k]$ can be PAC learned under the uniform distribution with membership queries by a randomized algorithm running in time $2^{n}/n^{\omega(1)}$, then for every $k \geq 1$ we have $\mathsf{BPE} \nsubseteq \mathcal{C}[n^k]$. Their result admits an extension to faster learners with a weaker advantage $\gamma$, with parameters similar to our \cref{thm:main}.  A simpler proof that non-trivial randomized learning yields lower bounds, with a stronger consequence, was obtained by Oliveira and Santhanam \citep{DBLP:conf/approx/OliveiraS18} under the additional assumption that the learner has ``pseudo-deterministic'' behaviour. We refer to their work for details.

In a more recent work, Oliveira \citep{DBLP:conf/icalp/Oliveira19} explored the learning-to-lower-bound connection in an \emph{indirect} way to show that a natural problem in time-bounded Kolmogorov complexity cannot be solved in probabilistic polynomial time. This can be achieved by proving a lower bound \emph{under the assumption} that learning algorithms do not exist (since otherwise a useful lower bound immediately follows from an existing learning-to-lower-bound connection).

Chen, Rothblum, Tell, and Yogev \citep{CRTY20} established a ``fine-grained'' learning-to-lower-bound connection with respect to circuit size in the setting of randomized learners. More precisely, they show that learning general Boolean circuits of size $n \cdot \mathsf{poly}(\log n)$ in randomized time $2^{n/\mathsf{poly}(\log n)}$ yields a function in $\mathsf{BPE}$ that cannot be computed by circuits of size $n \cdot \mathsf{poly}(\log n)$. While their running time assumption is much more stringent than $2^n/n^{\omega(1)}$, it is not necessary to learn circuits of large polynomial size to get interesting lower bound consequences.

In a parallel line of work, Williams \citep{DBLP:journals/siamcomp/Williams13} employed completely different methods to establish that ``non-trivial'' (running in time $2^n/n^{\omega(1)}$) deterministic \emph{satisfiability} algorithms for a class $\mathcal{C}$ of polynomial-size circuits imply that $\mathsf{NE} \nsubseteq \mathcal{C}$, where $\mathsf{NE}$ is an exponential time analogue of $\mathsf{NP}$. The satisfiability-to-lower-bound connection and its extensions have been highly successful in establishing new circuit lower bounds for a variety of circuit classes (see e.g.,~\citep{DBLP:journals/jacm/Williams14, DBLP:conf/stoc/MurrayW18, DBLP:conf/stoc/ChenR20, CLW20} and references therein). 
Chen, Oliveira, and Santhanam \citep{DBLP:conf/latin/ChenOS18} combined \emph{learning} and \emph{satisfiability} algorithms to strengthen lower bound consequences from learning algorithms in the particular case of $\mathcal{C} = \mathsf{ACC}^0$ (constant-depth polynomial-size circuits extended with modular gates). 

In contrast to all these works, here we establish the first general connection between the design of \emph{quantum} (learning) algorithms for an arbitrary class $\mathcal{C}$ of polynomial-size circuits and the existence of lower bounds against $\mathcal{C}$.

\paragraph{Organization.} In \cref{sec:prelim}, we fix notation and review a few basic definitions and results, as well as formalize our quantum learning model and the useful concept of inherently random circuits. In \cref{sec:lbs_from_learning}, we establish that non-trivial quantum learners imply lower bounds (\cref{thm:main}), under the assumption that a conditional $\PRG$ against quantum computations exists. \cref{sec:tech-tools} develops the tools that are needed to establish the correctness of our $\PRG$ construction. Finally, \cref{sec:prg-construction} defines and analyses a $\PRG$ with the desired properties, which completes the proof of \cref{thm:main}.

\paragraph{Acknowledgements.} We thank Lijie Chen (MIT) for several comments and feedback on a preliminary version of this paper, and the anonymous reviewers of QIP'2021 and FOCS'2021 for many useful remarks about the presentation.  S.A. was partially supported by the IBM Research Frontiers Institute  and acknowledges support from the Army Research Laboratory,  the Army Research Office under grant number W911NF-20-1-001. A.S. was partially supported by the Joint Centre for Quantum Information and Computer Science, University of Maryland, USA and acknowledges support from the Department of Defense. 
Part of this work was done while A.G. was affiliated to CWI and QuSoft. Part of this work was done while S.A, A.G.  and A.S were participating in the program ``The Quantum Wave in Computing" at Simons Institute for the Theory of Computing. 
This work received support from the Royal Society University Research Fellowship URF$\setminus$R1$\setminus$191059 and the UKRI Future Leaders Fellowship MR/S031545/1.
 
\section{Preliminaries}\label{sec:prelim}

\subsection{Basic definitions and notation}\label{sec:prelim_notation}

We begin with standard notation that will be employed throughout this work.\\ 

\noindent -- We use $\mathbb{N}$ to denote the set $\{1,2,3,\ldots\}$ of positive integers.\\
\noindent -- We denote by $\calU_n$ the uniform distribution over $n$-bit strings.\\
\noindent -- When sampling a uniformly random element $a$ from a set $A$, we might use $a \in A$ or $a \sim A$.\\
\noindent -- We say that two Boolean functions $f, g \colon \{0,1\}^n \to \{0,1\}$ are $\lambda$-close if $\Pr_{x \sim \mathcal{U}_n}[f(x) \neq g(x)] \leq \lambda$.\\
\noindent -- We say that $f$ computes $g$ with advantage $\gamma$ if $\Pr_{x \sim \mathcal{U}_n}[f(x) = g(x)] \geq 1/2 + \gamma$.\\
\noindent -- We use $\mathsf{negl}(n)$ to denote a function $g \colon \mathbb{N} \to \mathbb{N}$ such that, for every univariate polynomial $p(\cdot)$ with positive coefficients, there exists $n_0 \in \mathbb{N}$ such that $g(n) \leq 1/p(n)$ for every $n \geq n_0$.\\

We assume the familiarity of the reader with standard complexity classes (e.g.~$\mathsf{PSPACE}$) and basic notions from classical computational complexity theory and refer to a textbook such as \citep{book_complexity} for more information. Some notions from quantum computing and quantum complexity theory are reviewed in Section \ref{sec:prelim-quantum}.

\addtocontents{toc}{\setcounter{tocdepth}{-10}}

\subsubsection{Useful results} 

The following standard concentration bound will be used in some proofs.

\begin{lemma}[{Chernoff bound; see e.g.~\cite[Theorem~A.1.15]{alon2016probabilistic}}]
\label{lem:chernoff}
  Let $X_1,\ldots ,X_k$ be i.i.d.~$\01$-valued random variables each taking value $1$ with probability $p$. Let  $X = \sum_{i} X_i$ and $\mu=p k$. Then for any $\delta \in (0,1)$ it follows that
  \[
  \Pr\left[X \leq (1-\delta) \mu\right] \leq e^{-\delta^2\mu/2}.
  \]
\end{lemma}
An alternate version of this bound will also be useful.
\begin{lemma}[{Chernoff bound; see e.g.~\cite[Theorem~2.1]{janson2011random}}]
\label{lem:chernoff2}
  Let $X_1,\ldots ,X_k$ be i.i.d.~$\01$-valued random variables each taking value $1$ with probability $p$. Let  $X = \sum_{i} X_i$ and $\mu=p k$. Then for any $t\geq 0$ it follows that
  \[
  \Pr\left[|X -\E[X]|\geq t\right] \leq \exp\Big(-\frac{t^2}{2(\mu+t/3)}\Big).
  \]
\end{lemma}

We will also need the following form of the Hoeffding bound \citep{Hoeffding63}.

\begin{lemma}[Hoeffding bound~\citep{Hoeffding63}]
\label{lem:hoeffding}
Let $F \colon \mathcal{V} \to [0,1]$ be a function over a finite set $\mathcal{V}$ with expectation $\E_{x \sim \mathcal{V}}[F(x)] = \alpha \in [0,1]$. Let $R$ be a random subset of $\mathcal{V}$ of size $t$, and consider the random variable $X = \sum_{x \in R} F(x)$. Then the expectation of $X$ is $\mu = \alpha t$, and for any $0 \leq \gamma \leq 1$, $$\Pr\big [|X - \mu| \geq \gamma \mu \big ] \leq 2 \cdot e^{-\gamma^2 \mu/3}.$$ 
\end{lemma}

\subsubsection{Circuit and concept classes}
\label{sec:prelim-circuits}

For a typical Boolean circuit class $\mathcal{C}$ such as $\mathcal{AC}^0$, $\mathcal{TC}^0$, $\mathsf{Formulas}$, etc., and for a size function $s \colon \mathbb{N} \to \mathbb{N}$ measuring number of gates, we use $\mathcal{C}[s]$ to denote the set of languages $L \subseteq \{0,1\}^*$ that can be computed by a sequence of $\mathcal{C}$-circuits of size $s(n)$. For circuit classes of bounded depth, we might use $\mathcal{C}_d$ when referring to circuits of depth at most $d$. We stress that our notation refers to \emph{non-uniform} circuit classes, and consequently the complexity lower bound appearing in Theorem \ref{thm:main} is a non-uniform circuit complexity lower bound.

We will also employ circuit classes in the context of learning algorithms, where they are often referred to as concept classes. In this case, we will abuse notation and view $\mathcal{C}[s(n)]$ as the class of Boolean functions $f \colon \{0,1\}^n \to \{0,1\}$ that can be computed by $\mathcal{C}$ circuits over $n$ input variables of size at most $s(n)$. For convenience, we use $\mathsf{size}_\mathcal{C}(f)$ to denote the minimum number of gates in a $\mathcal{C}$-circuit computing $f$. We omit the subscript $\mathcal{C}$ in $\mathsf{size}_\mathcal{C}(f)$ when we refer to general Boolean circuits. For concreteness, in this case circuit size refers to number of gates in the  model consisting of fan-in two circuits over $\mathsf{AND}$, $\mathsf{OR}$, and $\mathsf{NOT}$ gates.

Theorem \ref{thm:main} applies to 
a broad family of circuit classes investigated in computational complexity theory, including fixed-depth classes such as depth-$2$ circuits consisting of majority gates. The only assumptions needed on the circuit class $\mathcal{C}$ are that:
\begin{itemize}
    \item[(\emph{i})] Any function $f_n \in \mathcal{C}[s(n)]$ can be computed by a general Boolean circuit of size $\mathsf{poly}(s(n))$. This is the case, for instance, if every gate allowed in $\mathcal{C}$ can be simulated by a Boolean circuit of size polynomial in the number of input wires of the gate.
    \item[(\emph{ii})] The class $\mathcal{C}$ is closed under restrictions of input variables to constants $0$ and $1$. In other words, if $f_n \in \mathcal{C}[s(n)]$ and we fix some input variables of $f_n$ to obtain a function $f'_n \colon \{0,1\}^{n'} \to \{0,1\}$, then $f'_n$ is also computed by a $\mathcal{C}$-circuit of size at most $s(n)$.
\end{itemize}

These assumptions are needed only in Section \ref{sec:qnp-lb}, and we rely on each of them as follows. We use Item $(\emph{i}\,)$ to show that if a language $L$ is not computable by (unrestricted) Boolean circuits of size $n^\alpha$ for some $\alpha \geq 1$, then it cannot be computed by $\mathcal{C}$-circuits of size $n^\beta$, where $\beta = \alpha/C'$ for some universal positive constant $C'$ that is independent of $\alpha$. On the other hand, we will rely on Item (\emph{ii}) to say that if a function $f \colon \{0,1\}^n \to \{0,1\}$ can be computed by $\mathcal{C}$-circuits of size polynomial in $n$, then any sub-function of $f$ defined over $n' = n^{\Omega(1)}$ input variables via a restriction also admits $\mathcal{C}$-circuits of size polynomial in $n'$. 

Note that Theorem \ref{thm:main} indeed applies to virtually any  class of polynomial size circuits.

We refer to a standard reference such as \citep{book_circuit} for more background on Boolean functions and circuit complexity theory.

\subsubsection{Classical learning algorithms}
\label{sec:classicallearning}
For the convenience of the reader, we review in this section (classical) learnability under the uniform distribution in the Probably Approximately Correct (PAC) model extended with membership queries. The quantum learning model that we consider in this work and its extensions are discussed in Section \ref{sec:prelimquantumlearn}.

\begin{definition}
Let $\mathcal{C}_n$ be a family of Boolean functions on $n$ input variables, $\mathcal{C}=\cup_{n\geq 1} \mathcal{C}_n$, $T \colon \nat\rightarrow \nat$, and $\varepsilon, \delta \colon \mathbb{N} \to [0,1]$. We say that $\mathcal{C}$ can be $(\varepsilon,\delta)$-learned in time $T$ under the uniform distribution with membership queries if there exists a randomized algorithm $\mathcal{A}$ that satisfies the~following guarantees: 
\begin{quote}
     For all $n\in \nat$, for every $f\in \mathcal{C}_n$, given $n$ and  query access to $f$, $\mathcal{A}$ runs in time $T(n)$ and with probability at least $1-\delta(n)$ over its internal randomness outputs a hypothesis~$h$ \emph{(}encoded as a general Boolean circuit\emph{)} that satisfies
    $$
    \Pr_{x\sim \calU_n}[h(x)=f(x)]\geq 1-\varepsilon(n).
    $$
\end{quote}
\end{definition}

It is also possible to focus on a fixed circuit class $\mathcal{C}$, and to consider learnability of functions computed by $\mathcal{C}$-circuits of an arbitrary size $s(n)$, for a fixed learning algorithm $\mathcal{A}$ that is independent of $s$ but that is given the value $s_f = \textsf{size}_\mathcal{C}(f)$ (or an upper bound on $s_f$) as part of its input. In this case, we allow the running time of $\mathcal{A}$ to depend on $s_f$. Similarly, we can provide the learner with $\delta$ and $\varepsilon$, and allow its running time to depend on these parameters.

We stress that the output hypothesis $h$ produced by $\mathcal{A}$ is not required to be a ``circuit'' from~$\mathcal{C}$, when $\mathcal{C}$ explicitly refers to a circuit class and a size measure.

We refer to a standard reference such as \citep{book_learning} for more background in computational learning theory.

\addtocontents{toc}{\setcounter{tocdepth}{3}} 
\subsection{Quantum computation}\label{sec:prelim-quantum}

We assume the reader is familiar with the quantum computing framework and notation, such as Dirac's bra-kets. We refer to a standard text such as~\cite{NC2010} for a detailed introduction to quantum computation. Here, we discuss concepts and notation of specific relevance to this paper.

A quantum circuit $U$ on $n$-qubits is a sequence of quantum gates (i.e.,~unitary matrices). Throughout this paper we will assume our quantum gates are restricted to the one-qubit Hadamard gate $\textsf{H}$ and $3$-qubit Toffoli gate $\textsf{Toff}$ defined as follows: 
$$
\textsf{H}:\ket{b_1} \rightarrow \frac{1}{\sqrt{2}}\big(\ket{0}+(-1)^{b_1}\ket{1}\big), \qquad \textsf{Toff}:\ket{b_1,b_2,b_3}\rightarrow \ket{b_1,b_2,b_3\oplus b_1\cdot b_2}
$$
for $b_1,b_2,b_3\in \01$. The Toffoli gate is particularly useful for us as any classical circuit can be implemented as a quantum circuit using only Toffoli gates. Additionally, we choose this set of gates because $\{\textsf{H},\textsf{Toff}\}$ is universal for approximating unitaries with only real entries,  i.e., every unitary with real entries can be approximated arbitrarily well with only $\textsf{H}$ and $\textsf{Toff}$ gates~\cite{Aharonov03}. In fact, our results still go through as long as we have a gate set whose size is constant.
In our case, the size of a quantum circuit is the number of Hadamard and Toffoli gates in the circuit. 
The classical description of a quantum circuit $U$ on $n$-qubits, denoted as $\code(U)$, is a canonical encoding of the sequence of gates (and the qubits they act on) in the circuit, represented as a binary string. We will repeatedly use the result that, given the description $\code(U)$ of a quantum circuit $U$ of a predetermined size, there exists an \emph{efficient} universal quantum circuit $\mathscr{U}$ that can simulate the action of $U$ on any $n$-qubit input $\ket{y}$. Formally, we have the following definition and result.

\begin{definition}[Universal quantum circuit]
Fix $n \in \mathbb{N}$ and let $\calC$ be the collection of quantum circuits on $n$-qubits of size $s(n)$. An $(n+m)$-qubit quantum circuit $\mathscr{U}$ is universal for $\calC$ if for every circuit $U \in \calC$ and associated $\code(U) \in \01^*$,
$$
\mathscr{U} (\ket{y} \otimes \ket{\code(U)}) = (U \ket{y}) \otimes \ket{\code(U)}, \quad \text{for every } y \in \01^n.
$$
\end{definition}
\noindent
There exist efficient constructions of $\mathscr{U}$ whose size has only a log-factor blow-up in the parameters $n$ and $s(n)$ (see~\citep{BFGH10universal}).  (We note that a polynomial blow-up would still be sufficient in our constructions.)

 We say that a quantum circuit $U$ computes a function $f:\01^n\rightarrow \01$ if the following holds: there exists $m\geq 0$ such that, for every $x\in \01^n$, when measuring the first qubit of $U$ on input $\ket{x,0^m}$, we get $f(x)$ with probability at least $2/3$, i.e.,
$$
\big\|\Pi_{f(x)} U\ket{x,0^m}\big\|^2\geq 2/3,
$$
where $\Pi_{f(x)} \eqdef \kb{f(x)}\otimes \id$. The extra $m$ qubits are called auxiliary qubits. As in probabilistic computing, the constant $2/3$ is not important here as one can use standard amplification techniques, such as taking the majority vote over many repetitions, to amplify this probability to $1-\delta$ for some $\delta > 0$. 

Often in this paper, we will abuse notation and write $\Pr[U(x)=f(x)]$ when referring to $\|\Pi_{f(x)} U\ket{x,0^m}\|^2$ -- the probability that the quantum circuit $U$ on input $x$ outputs $f(x)$ after measurement. Similarly, we might write
$$
\Pr_{x \sim \{0,1\}^n,\,U}[U(x) = f(x)] \eqdef 
\E_{x \sim \{0,1\}^n}\left [\big\|\Pi_{f(x)} U\ket{x,0^m}\big\|^2 \right ]
$$
to refer to the \emph{expected agreement} between the output of $U$ and $f$ over a random input string.

These definitions extend naturally to the case of functions with a non-Boolean output. For instance, if $g \colon \{0,1\}^n \to \{0,1\}^\ell$, then we use
$$
\Pr_{x \sim \{0,1\}^n,\,U}[U(x) = g(x)] \eqdef 
\E_{x \sim \{0,1\}^n}\left [\big\|\Pi_{g(x)} U\ket{x,0^m}\big\|^2 \right ],
$$
where $\Pi_{w}$ for a string $w \in \{0,1\}^\ell$ is defined as $=\kb{w}\otimes \id$.

Unlike classical probabilistic computation, quantum interference could result in the output qubit being entangled with junk values left in the auxiliary qubits when they are used in intermediate computations. This could impact the probability that $[U(x) = f(x)]$ on measurement -- especially when $U$ is used as a sub-routine in larger computations. However, by relying on the reversible nature of unitary computations, we can ``\emph{remove the garbage}'' via standard techniques. In particular, given $U$ on $n+m$ qubits that computes $f$,\footnote{We remark that in Section~\ref{sec:self_reduc} when we need to remove garbage in unitary computations, we will always amplify the success probability of $U$, so as to ensure that $U$ computes $f$ with overwhelming probability before constructing $\tilde{U}$.}  construct the garbage free version $\tilde{U}$ on $n+m+1$ qubits as follows: 
\begin{itemize}
    \item[(\emph{i})] Compute $U$ on input $\ket{x, 0^m}$;
    \item[(\emph{ii})] Copy the first qubit of $U$ into the $(n+m+1)$-th qubit of $\tilde{U}$ with a quantum gate;
    \item[(\emph{iii})] Un-compute $U$ by applying $U^{\dag}$ on the first $n+m$ qubits -- i.e., apply each gate of $U$ in reverse order. This will return $\ket{x, \textsf{junk}}$ back to $\ket{x, 0^m}$; 
    \item[(\emph{iv})] Measure the $(n+m+1)$-th qubit to check if $\tilde{U}(x) = f(x)$ with $\widetilde{\Pi}_{f(x)} = \kb{f(x)} \otimes \kb{x} \otimes \kb{0^m}$. 
\end{itemize}

The lower bounds in this paper involve quantum complexity classes, which are defined next. Owing to the probabilistic nature of quantum circuits, all classes of interest are in the bounded error paradigm. We first consider the uniform class $\BQTIME$ (bounded-error quantum time). 
\begin{definition}[$\BQTIME$]
Let $f : \01^* \rightarrow \01$ and $t: \nat \rightarrow \nat$. We say that $f$ is computable in bounded-error quantum $t(n)$-time if there exists a deterministic algorithm $A$ which on input $1^n$ runs in time $\poly(n, t(n))$ and outputs the description of a quantum circuit $U$, represented as a string $\code(U)$, with at most $t(n)$ quantum gates such that $U$ computes $f$. In this case, we write $f \in \BQTIME[t(n)]$.
\end{definition}

We also fix notation for the complexity classes that arise when $t(n)$ scales polynomially or exponentially in $n$. 

\begin{definition}[Bounded-error quantum complexity classes]
Let $n \in \nat$ and $\nu, c > 0$ be constants. Then, quantum polynomial time refers to the class $\BQP := \bigcup_{c > 0} \BQTIME[n^c]$. Similarly, $\BQE := \bigcup_{c > 0} \BQTIME[2^{c \cdot n}]$ refers to the class of languages computable in quantum $2^{O(n)}$-time. Quantum sub-exponential time is denoted as $\mathsf{BQSUBEXP} := \bigcap_{0 < \nu < 1} \BQTIME\big[ 2^{n^\nu} \big]$.
\end{definition}

\subsection{Quantum learning algorithms and extensions} \label{sec:prelimquantumlearn}

 In this section we define the quantum learning model. In contrast to the classical model, a quantum learner is given quantum oracle access to a concept $f$. This means the learner is allowed to perform a \emph{quantum membership query}, which is defined by a unitary map $O_{f}$ acting on $n+1$ qubits whose action on basis states is defined as
$$
O_{f}:\ket{x,b}\rightarrow \ket{x,b\oplus f(x)},
$$
for every $x\in \01^n$ and $b\in \01$. Naturally, a quantum learner can perform a quantum computation in between quantum queries, and its goal is the same as for a classical learner (defined in Section~\ref{sec:classicallearning}), i.e., to output a hypothesis $h$ that approximates the target concept $f\in \mathcal{C}$. 

We can formalise this model as follows. A quantum learning algorithm $\mathcal{L}$ running in time $G$ is described by a uniform sequence of quantum circuits $Q_n$, where each $Q_n$ contains at most $G(n)$ gates. The quantum circuit $Q_n$ expects as input $\ket{0^m}$ for some $m = m(n)$, and consists of a sequence of gates from $\{\mathsf{H}, \mathsf{Toff}, \mathcal{O}\}$, where $\mathcal{O}$ is defined over $n + 1$ qubits and is considered as applying a single gate. In other words, when learning an unknown Boolean function $f \colon \{0,1\}^n \to \{0,1\}$, $\mathcal{O}$ computes as the unitary map $O_f$ described above.  Finally, the output hypothesis of $Q_n$ when computing with quantum membership query access to $f$ (referred to as $Q_n^f$) is described by the string corresponding to the output measurement of $Q_n^f$. We note that intermediate measurements are also allowed in $Q_n$.

\begin{definition}[Standard quantum learners]
Let $\mathcal{C}_n$ be a family of Boolean functions on $n$ input variables and $\mathcal{C}=\cup_{n\geq 1} \mathcal{C}_n$. Let $G:\nat\rightarrow \nat$. We say $\mathcal{C}$ can be  $(\varepsilon(n),\delta(n))$-learned in time $G(n)$ if there exists a quantum learning algorithm $\mathcal{L}$ that satisfies the following: 
\begin{quote}
     For all $n\in \nat$, for every $f\in \mathcal{C}_n$, given quantum oracle access to $f$, $\mathcal{L}$ uses at most $G(n)$ gates and with probability at least $1-\delta(n)$ \emph{(}over the output measurement of $\mathcal{L}$\emph{)}  outputs a \emph{(}classical\emph{)} hypothesis $h$ that satisfies
    $$
    \Pr_{x\sim \calU_n}[h(x)=f(x)]\geq 1-\varepsilon(n).
    $$
\end{quote}
\end{definition}

We remark that one distinct advantage of having quantum oracle access to $f$ is that one can efficiently \emph{Fourier sample} for the squared Fourier distribution $\{\widehat{f}^2(S)\}_S$,  which might be computationally hard for classical learners (we discuss this in more detail in Appendix~\ref{app:Fsampling}).
For more on this  quantum learning model, we refer the interested reader  to~\cite{servedio2004equivalences,arunachalam2017guest}. %

In this paper, we consider a natural generalization of this model: we allow the quantum learner to output a {\em quantum} hypothesis, meaning that $\mathcal{L}$ is allowed to output a classical description of a quantum circuit $\widetilde{U_f}$ that \emph{approximates} the function $f$. Recall that \emph{computing} $f$ on a given input $x$ means that measuring the first qubit of $U$ acting on an input $\ket{x}$ and auxiliary qubits set to $\ket{0}$, in the computational basis, produces $f(x)$ with probability $\geq 2/3$. On the other hand, our notion of a quantum circuit $\widetilde{U_f}$ that approximates $f$ refers to the \emph{expected agreement} between the output of $\widetilde{U_f}$ and $f$ on a random input $x$, as defined in Section~\ref{sec:prelim-quantum}. 

We formally define this model below.\footnote{We only define the uniform distribution learning model below. Note that one can naturally define a quantum learner with quantum hypothesis under an \emph{arbitrary} distribution by changing $\calU_n$ to an arbitrary distribution $\De$.}

\begin{definition}[Quantum learning with quantum hypothesis]
Let $\mathcal{C}_n$ be a family of Boolean functions on $n$ input variables and $\mathcal{C}=\cup_{n\geq 1} \mathcal{C}_n$. Let $G:\nat\rightarrow \nat$. We say $\mathcal{C}$ can be  $(\varepsilon(n),\delta(n))$-learned in time $G(n)$ if there exists a quantum learning algorithm $\mathcal{L}$ that satisfies the following: 
\begin{quote}
    For all $n\in \nat$, for every $f\in \mathcal{C}_n$, given quantum oracle access to $f$, $\mathcal{L}$ uses at most $G(n)$ gates and with probability at least $1-\delta(n)$ \emph{(}over the output measurement of $\mathcal{L}$\emph{)}  outputs the description of a quantum circuit $\widetilde{U_f}$ that satisfies
      \[
 \E_{x\sim \calU_n} \big[ \norm{\Pi_{f(x)} \widetilde{U_f}\ket{x}\ket{0}}^2 \big] \geq     1-\eps(n) \;,
  \]
  where $\Pi_{f(x)} = \kb{f(x)} \otimes \id$. 
\end{quote}
\end{definition}

\subsection{Quantum natural properties}\label{sec:nat_prop}

Let $\mathcal{F}_n$ be the family of all Boolean functions on $n$ input bits. We say that $\Gamma = \{\Gamma_n\}_{n \geq 1}$ is a combinatorial property of Boolean functions if $\Gamma_n \subseteq \mathcal{F}_n$ for every $n \geq 1$. For every function $f \in \mathcal{F}_n$ and $N=2^n$, let $\mathsf{tt}(f)\in \{0,1\}^N$ be the truth table representing $f$. 
 We associate with $\Gamma$ a language $L_\Gamma \subseteq \{0,1\}^*$ defined as follows. A string $x$ is in $L_\Gamma$ if and only if $x = \mathsf{tt}(f)$ for some $n$ and $f \in \Gamma_n$. Conversely, given a string $w \in \01^N$, let $\mathsf{fnc}^w$ denote the boolean function $f : \01^n \rightarrow \01$ such that $\mathsf{tt}(f) = w$.

\begin{definition}[Natural property~\cite{DBLP:journals/jcss/RazborovR97}] \label{def:nat_prop}
Let $\Gamma$ be a combinatorial property, $\mathcal{C}$ be a circuit class, $\mathfrak{D}$ be a \emph{(}uniform or non-uniform\emph{)} complexity class, and $s \colon \mathbb{N} \to \mathbb{N}$. We say that $\Gamma$ is a $\mathfrak{D}$-natural property useful against $\mathcal{C}[s]$ if there exists $n_0 \in \mathbb{N}$ for which the following holds:
\begin{itemize}
\item \textsf{Constructivity}: $L_\Gamma \in \mathfrak{D}$, i.e., for every $n$, $f_n \stackrel{?}{\in}\Gamma_n$ is decidable in $\mathfrak{D}$.
\item \textsf{Largeness}: For every $n \geq n_0$, $\Pr_{f \sim \mathcal{F}_n}[f \in \Gamma_n] \geq 1/2$.
\item \textsf{Usefulness}: For every $n \geq n_0$, $\mathcal{C}_n[s(n)] \cap \Gamma_n = \emptyset$.
\end{itemize}
If $\mathfrak{D} = \mathsf{BQP}$, we say that $\Gamma$ is a \emph{quantum} natural property.
\end{definition}

Note for instance that for $\mathsf{BPP}$ and $\mathsf{BQP}$-natural properties the corresponding algorithm (that decides $L_\Gamma$)  is allowed to run in time polynomial in the input length $N = 2^n$.  We will need to slightly relax the guarantees offered by a natural property in order to establish a connection to~learning.

\begin{definition}[Informal]
We say that a circuit class $\mathcal{C}$ admits a \emph{promise} natural property if the underlying algorithm \emph{(}that decides $L_\Gamma$\emph{)} in Definition \ref{def:nat_prop} accepts every string in $L_\Gamma$ with probability $\geq 2/3$ and rejects every string encoding a function in $\mathcal{C}_n[s(n)]$ with probability $\geq 2/3$. \emph{(}For instance, the algorithm might accept some strings with probability close to $1/2$.\emph{)}
\end{definition} 

\begin{definition}[Promise $\mathsf{BQP}$-natural property]
\label{defn:qnaturalprop}
We say that there is a \emph{(}promise\emph{)} $\mathsf{BQP}$-quantum natural property against $\mathcal{C}$-circuits of size $s$ if there is a quantum algorithm $\calD$ operating over inputs of the form $N = 2^n$ for which the following holds for every large enough parameter $n$:
\begin{itemize}
\item \textsf{Constructiveness:} $\calD$ runs in time polynomial on its input length $N$.
\item \textsf{$s(n)$-hardness:} For every 
$f \in \mathcal{C}_n[s(n)]$, $\calD$ accepts $\mathsf{tt}(f)$ with probability at most $1/3$.
\item \textsf{Density:} There is a set $\Gamma_n \subseteq \{0,1\}^N$ of density $|\Gamma_n|/2^N \geq 1/2$ such that, for every $w \in \Gamma_n$, $\mathsf{fnc}^w \notin \mathcal{C}_n[s(n)]$ and $\calD$ accepts $w$ with probability at least $2/3$.
\end{itemize}
\end{definition}

\begin{remark}
We stress that by the definition of $\mathsf{BQP}$, a promise $\mathsf{BQP}$-natural property implies a {\em classical} algorithm $A$ such that $A(1^N)$ computes the quantum circuit $D_N$ that deals with inputs of size $N$.
\end{remark}

\subsection{Self-reducibility}

\begin{definition} [Downward self-reducibility]\label{def:downward-selfreducible}
A function $f \colon \{0,1\}^* \to \{0,1\}$ is said to be downward self-reducible if there is a deterministic polynomial-time oracle procedure $A^f$ such that\emph{:}
\begin{enumerate}
 \item[\emph{1.}] On any input $x$ of length $n$, $A^f(x)$ only makes queries of length $< n$.
 \item[\emph{2.}] For every input $x$, $A^f(x) = f(x)$.
\end{enumerate}
\end{definition}

\begin{definition} [Random self-reducibility]\label{def:random-selfreducible}
A function $f\colon \{0,1\}^* \to \{0,1\}$ is said to be random self-reducible if there are constants $a, b, c \geq 1$ and polynomial-time computable functions $g \colon \{0,1\}^* \to \{0,1\}^*$ and $h \colon \{0,1\}^* \to \{0,1\}$ satisfying the following conditions\emph{:}

\begin{enumerate}
\item[\emph{1.}] For large enough $n$, for every $x \in \{0,1\}^n$ and for each $i \in \mathbb{N}$ such that $i\in [n^a]$, $g(i,x,r) \sim~\mathcal{U}_n$ when $r \sim \mathcal{U}_{n^{c}}$.
\item[\emph{2.}] For large enough $n$ and for every function $\tilde{f}_n \colon \{0,1\}^n \to \{0,1\}$ that is $(1/n^b)$-close to $f$ on $n$-bit strings, for every $x \in \{0,1\}^n$\emph{:}
$$f(x) = h\big(x, r, \tilde{f}_n(g(1,x,r)), \tilde{f}_n(g(2,x,r)), \ldots , \tilde{f}_n(g(n^a, x, r))\big)$$
with probability $\geq 1-2^{-2n}$ when $r \sim \mathcal{U}_{n^{c}}$.
\end{enumerate}
\end{definition}

Trevisan and Vadhan \citep{DBLP:journals/cc/TrevisanV07} construted a language in $\mathsf{PSPACE}$ that simultaneously satisfies the following conditions.\footnote{The proof of this result in \citep{DBLP:journals/cc/TrevisanV07} refers to ``self-correction'' instead of ``random-self-reducibility'', but they ultimately rely on a decoding algorithm for polynomials making queries that are uniformly distributed on each coordinate.}

\begin{theorem}[Trevisan and Vadhan \citep{DBLP:journals/cc/TrevisanV07}] \label{t:TV_language} There is a language $L^\star \subseteq \{0,1\}^*$ satisfying the following properties:
\begin{itemize}
    \item $L^\star \in \mathsf{PSPACE}$. In particular, there is a positive constant $d_\star$ such that $L$ on inputs of length $n$ can be computed in time $O(2^{n^{d_\star}})$. 
    \item $L^\star$ is complete for $\mathsf{PSPACE}$ with respect to deterministic polynomial-time reductions.
    \item $L^\star$ is downward-self-reducible and random-self-reducible, with parameters $a_\star$, $b_\star$, and $c_\star$ in Definition \ref{def:random-selfreducible}.
\end{itemize}
\end{theorem}
\subsection{Pseudorandomness}
\label{sec:def-prg}

For $s \in \mathbb{N}$ and $\varepsilon \in [0,1]$, we say that a distribution $\mathcal{D}$ supported over $\{0,1\}^m$ is $(s, \varepsilon)$-\emph{pseudorandom against quantum circuits} if for every quantum circuit $C$ of size $s$ defined over $m$ input bits, we have 
\[ 
\left |  \Pr_{x \sim \{0,1\}^m,\,C}[C(x) = 1] - \Pr_{y \sim \mathcal{D},\,C}[C(y) = 1]  \right | \leq \eps.
\]
Note that each probability above refers to an appropriate random input and the randomness present in the output of $C$. 

We will consider weaker forms of pseudorandomness that hold against \emph{uniformly constructed} families of quantum circuits. For functions $s \colon \mathbb{N} \to \mathbb{N}$ and $\varepsilon \colon \mathbb{N} \to [0,1]$, we say that an ensemble $\{\mathcal{D}_m\}_{m \geq 1}$ of distributions $\mathcal{D}_m$ supported over $\{0,1\}^m$ is $(s,\varepsilon)$-\emph{pseudorandom against uniform quantum circuits} if for every deterministic algorithm $A(1^m)$ that runs in time $s(m)$ and outputs a quantum circuit $C_m$ over $m$ input variables and of size at most $s(m)$,
\[ 
\left |  \Pr_{x \sim \{0,1\}^m,\,C_m}[C_m(x) = 1] - \Pr_{y \sim \mathcal{D}_m,\,C_m}[C_m(y) = 1]  \right | \leq \eps(m).
\]

For technical reasons, the pseudorandom distributions we construct are supported over strings of length $m(n) = \lfloor 2^{n^\lambda} \rfloor$, for $n \in \mathbb{N}$ and a fixed $\lambda > 0$. (The index parameter $n$ will play an important role in our constructions.) Moreover, we will only be able to show that the distributions  $\mathcal{D}_{m(n)}$ are pseudorandom for infinitely many values of $n$. For these reasons, we refine the previous definition in the following way.

Let $m \colon \mathbb{N} \to \mathbb{N}$, $s \colon \mathbb{N} \to \mathbb{N}$ and $\varepsilon \colon \mathbb{N} \to [0,1]$.
 We say that an ensemble $\{\mathcal{D}_{m(n)}\}_{n \geq 1}$ of distributions $\mathcal{D}_m$ supported over $\{0,1\}^{m(n)}$ is \emph{infinitely often} $(s, \varepsilon)$-\emph{pseudorandom against uniform quantum circuits} if for every deterministic  algorithm $A(1^m)$ that runs in time $s(m)$ and outputs a quantum circuit $C_m$ over $m$ input variables and of size at most $s(m)$, for infinitely many values of $n$ and $m \eqdef m(n)$,
 \[ 
\left |  \Pr_{x \sim \{0,1\}^m,\,C_m}[C_m(x) = 1] - \Pr_{y \sim \mathcal{D}_m,\,C_m}[C_m(y) = 1]  \right | \leq \eps(m).
 \]

Finally, our distributions are obtained from \emph{pseudorandom generators}. Let $\ell \colon \mathbb{N} \to \mathbb{N}$. We say that a family of functions $\{G_n\}_{n \geq 1}$ is a \emph{quick infinitely often} $(s, \varepsilon)$-\emph{pseudorandom generator against uniform quantum circuits} of \emph{seed length} $\ell(n)$ and \emph{output length} $m(n)$ if the following conditions hold:
\begin{itemize}
    \item[(\emph{i})] \textsf{Stretch}: Each function \emph{stretches} an $\ell(n)$-bit input to $m(n)$-bits i.e., $G_n \colon \{0,1\}^{\ell(n)} \to \{0,1\}^{m(n)}$.
    \item[(\emph{ii})]  \textsf{Uniformity and Running Time}: There is a deterministic algorithm $A$ that when given $1^n$ and $x \in \{0,1\}^{\ell(n)}$ runs in time $O(2^{\ell(n)})$ and outputs $G_n(x)$.
    \item[(\emph{iii})]  \textsf{Pseudorandomness}:  For each $n \geq 1$, let $\mathcal{D}_{m(n)} = G_n(\mathcal{U}_{\ell(n)})$ be the distribution over $m(n)$-bit strings induced by evaluating $G_n$ on a random $\ell(n)$-bit string. Then the corresponding ensemble $\{\mathcal{D}_{m(n)}\}_{n \geq 1}$ is infinitely often $(s, \varepsilon)$-pseudorandom against uniform quantum circuits.
\end{itemize}
For convenience, we will also say in this case that 
$$
G = \{G_n\}_{n \geq 1}~\text{is an infinitely often}~(\ell, m, s,\varepsilon)\text{-generator}.
$$ 
Under our notation, observe that $\ell(n)$ and $m(n)$ are functions indexed by $n$ that control the ``stretch'' of $G_n$ (i.e.~$G_n$ maps $\ell$ bits to $m$ bits), and $s(m)$ and $\varepsilon(m)$ are pseudorandomness parameters of the induced distribution $\mathcal{D}_m = G_n(\mathcal{U}_\ell)$.

\subsection{Inherently probabilistic computations}\label{ss:inherently_probabilistic}

As an intermediate model between classical and quantum computations, it will be useful to consider \emph{inherently probabilistic circuits}. Roughly speaking, these are computational devices whose randomness remains ``inacessible'', in the sense that there is no simple way of decomposing an inherently probabilistic circuit as a distribution of deterministic circuits.

To make this point more precise, we review the discussion from Section \ref{sec:quantum_challenges}. Consider
the standard model of randomized Boolean circuits. In other words, we consider a Boolean circuit $C(x,y)$, where $x$ is the input string and $y$ is the random string, and  say that $C$ computes a Boolean function $g \colon \{0,1\}^m \to \{0,1\}^k$ on an input $x \in \{0,1\}^m$ if $\Pr_y[C(x,y) = g(x)] \geq 2/3$. Similarly, we can also discuss the correlation between $C$ and $g$, captured by $\gamma \eqdef \Pr_{x,y}[C(x,y) = g(x)]$. An extremely useful trick in this model is that, by an averaging argument, there must exist a fixed string $y'$ such that $\Pr_x[C(x,y') = g(x)] \geq \gamma$. In other words, there is a \emph{deterministic} circuit $C_{y'}(x) \eqdef C(x,y')$ that correctly computes $g$ on a $\gamma$ fraction of inputs. This often allows one to reduce the analysis of randomized Boolean circuits to the deterministic case. Additionally, $C_{y'}$ ``forces'' $C(x,y)$ to commit to a single consistent output on every $x$, which can be relevant if $C$ is needed as a subroutine in other computations.

As a concrete example, note that this idea is crucial in the proof of $\mathsf{BPP} \subseteq \mathsf{P}/\mathsf{poly}$, which is a simple combination of  reducing the failure probability of a randomized circuit and fixing its random~input. 

While quantum computations share superficial similarities with the model of randomized computations and their corresponding randomized circuits, there is no obvious way of reducing a quantum computation to a distribution of ``deterministic'' quantum computations. In some cases, this is part of the difficulty when extending classical results to the quantum setting. With this in mind, below, we introduce inherently probabilistic circuits, which serve as a useful intermediate model in our analysis of quantum computations in the context of error correction and hardness amplification.\\

\noindent \textbf{Inherently probabilistic circuits.} We adopt a definition of inherently probabilistic computations that is sufficient for our purposes. In order to be as general as possible and because of how we access these computations, we model an \emph{inherently probabilistic circuit} $\mathcal{A}$ over $m$ input bits and that produces $\ell$ output bits as a function $\mathcal{A} \colon \{0,1\}^m \to \mathfrak{F}$, where $\mathfrak{F}$ is the set of probability distributions supported over a fixed domain $\{0,1\}^\ell$. In other words, $\mathcal{A}$ assigns to each input $z \in \{0,1\}^m$ a distribution $\mathcal{A}(z)$ supported over $\{0,1\}^\ell$. We say that $\mathcal{A}$ computes a function $g \colon \{0,1\}^m \to \{0,1\}^\ell$ with probability at least $\varepsilon$ if 
$$
\Pr_{\substack{z \sim \{0,1\}^m,\\ v \sim \mathcal{A}(z)}}[v = g(z)] \geq \varepsilon.
$$
Note that nothing is assumed about the relation between distributions $\mathcal{A}(z_1)$ and $\mathcal{A}(z_2)$ for $z_1 \neq z_2$, and that there is no way of globally ``fixing'' the randomness of $\mathcal{A}$ across different input strings.

We will also need to employ inherently probabilistic computations as subroutines inside standard deterministic and randomized computations.  We model this situation using standard Boolean circuits and ``inherently probabilistic oracles gates'', described next.

Deterministic Boolean circuits with oracle gates are defined in the standard way. More precisely, we use $C^\mathcal{O}$ to represent a deterministic circuit that in addition to its original gates can label certain gates of fan-in $m$ and of fan-out $\ell$ by $\mathcal{O}$. The computation of $C^\mathcal{A}(x)$ on an input $x$ is defined once we set $\mathcal{O}$ to an inherently probabilistic circuit $\mathcal{A}$. Formally, each gate of $C^\mathcal{A}(x)$ is now a \emph{random variable} whose value depends on the input $x$ and on the outcomes of the calls to $\mathcal{A}$, each distributed according to $\mathcal{A}(z)$ when the input to $\mathcal{A}$ is the string $z$ (two calls of $\mathcal{A}$ over the same string $z$ are distributed independently). For convenience, if $C^\mathcal{O}$ is an oracle circuit with $k$ designated output gates, we use $C^\mathcal{A}(x)$ to denote the random variable supported over $\{0,1\}^k$ and distributed according to the bits produced by its output gates (under a fixed order of the output gates of $C^\mathcal{O}$).

If $g \colon \{0,1\}^n \to \{0,1\}^k$, $C^\mathcal{O}$ is a Boolean oracle circuit over $n$ input bits with oracle gates of fan-in $m$ and fan-out $\ell$, and $\mathcal{A}$ is inherently probabilistic, we say that $C^\mathcal{A}$ computes $g$ with probability at least $\varepsilon$ if $\Pr_{x \sim \{0,1\}^n, \mathcal{A}}[C^\mathcal{A}(x) = g(x)] \geq \varepsilon$.

It will also be useful to consider \emph{randomized} oracle circuits with access to an \emph{inherently probabilistic oracle} $\mathcal{A}$. In this setting, $C^\mathcal{O}$ has, in addition to its input string $x$, an extra input $y$ for random bits. The computation of $C^{\mathcal{A}}$ on an input pair $(x,y)$ is defined as above, i.e., the inputs $x$ and $y$ are considered as a single input string of $C^{\mathcal{A}}$. We can similarly consider the probability $\Pr_{x, y, \mathcal{A}}[C^{\mathcal{A}}(x,y) = g(x)]$ of computing $g$ using $C^{\mathcal{A}}$. Note in this case that $y$ can be fixed to a given string $y'$, but the randomness associated with each call to $\mathcal{A}$ cannot be fixed. %

When it is clear from the context, we might omit superscripts $\mathcal{O}$ and $\mathcal{A}$ when referring to such oracle circuit computations. 

A somewhat similar model is explored by Kearns and Schapire \cite{DBLP:journals/jcss/KearnsS94} (see also \citep{DBLP:journals/ml/Yamanishi92}) from a different perspective.\footnote{This reference has been brought to our attention by an anonymous reviewer.} Their paper focuses on the learnability of functions of the form $c \colon X \to [0,1]$, where $c(x)$ can be interpreted as a probability distribution over $\{0,1\}$. In contrast, in this work we assign to each input $x$ a \emph{probability distribution over strings}. Most importantly, we are concerned with the subtleties of \emph{computing} with such devices, which will be used as internal building blocks in other computations, while the focus of \cite{DBLP:journals/jcss/KearnsS94} and \citep{DBLP:journals/ml/Yamanishi92} is on the \emph{learnability} of probabilistic/stochastic functions.

We notice that any quantum circuit that is accessed through {\em classical queries} can be viewed as an inherently probabilistic circuit.

\begin{lemma}\label{lem:quantum_as_inherently_random}
Let $\mathcal{R}$ be an inherently probabilistic circuit and $\mathcal{Q}$ a quantum circuit such that on every input $x$, the output distributions of $\mathcal{R}(x)$ and $\mathcal{Q}(x)$ are exactly the same.\footnote{Notice that here $\mathcal{Q}(x)$ could also have a garbage register that is traced out.} Then for every probabilistic algorithm
 $\mathcal{A}$, which might have some classical input $w$, that is allowed to make \emph{(}classical\emph{)} queries to $\mathcal{R}$ or $\mathcal{Q}$, we have that for every $y$,
\[
\Pr_{\calA,\mathcal{R}}[\calA^{\mathcal{R}}(w) = y] = 
\Pr_{\calA,\mathcal{Q}}[\calA^{\mathcal{Q}}(w) = y].
\]  
\end{lemma}
\begin{proof}
Notice that since $\mathcal{A}$ is a classical algorithm and is only allowed to make classical queries to $\mathcal{Q}$, $\mathcal{A}$ only has access to samples of the output distribution of $\mathcal{Q}$. Since the output distribution of $\mathcal{Q}$ and $\mathcal{R}$ is exactly the same for every input, the result follows.
\end{proof}

\section{Lower bounds from learning algorithms: a modular approach via PRGs}\label{sec:lbs_from_learning}

In this section, we show how to derive circuit lower bounds from quantum learning algorithms via two steps: first, in \cref{sec:qnp-from-ql} we show that quantum learners imply quantum natural properties; then, in \cref{sec:qnp-lb} we show that given a pseudorandom generator against uniform quantum computation, quantum natural properties imply circuit lower bounds. 

\subsection{Quantum natural properties from quantum learning algorithms}
\label{sec:qnp-from-ql}

\begin{theorem}\label{lem:learning_to_natural} There is a universal constant $\lambda \geq 1$ for which the following holds.
 Let $\gamma \colon \nat \rightarrow [0,1/2]$ be a constructive function such that $\gamma(n) \geq \lambda \cdot  2^{-n/2}$, and let $s(n) = \gamma(n)^2\cdot 2^{n}/(\lambda\cdot n)$. If $\mathcal{C}[n^k]$ can be learned under the uniform distribution in quantum time at most $s(n)$ by a $(\delta, \varepsilon)$-learner with $\delta(n) \leq 0.99$ and $\varepsilon(n) \leq 1/2 - \gamma(n)$, then there exists a promise quantum natural property against~$\mathcal{C}[n^k]$. 
\end{theorem}
\begin{proof} First, we prove the result under the assumption that $\delta \leq 1/10$. Then we show that this can be easily relaxed to learners that only succeed with probability $1/100$, as in the statement of the~result.

  By our assumption, there exists a quantum learning algorithm $\calA^{\ket{f}}$ that runs in quantum time $s(n)$ for a function $f : \01^n \rightarrow \01$ computed by a circuit $f \in \calC[n^k]$, and with probability at least $9/10$, the algorithm $\calA^{\ket{f}}$ outputs a quantum hypothesis $U$ such that
  \begin{align}
  \label{eq:conditionofgoodleanringNP}
      \E_{x} \big[ \norm{\Pi_{f(x)} U\ket{x,0}}^2 \big]\geq \frac{1}{2} + \gamma(n) \;.
    \end{align}
  We show now how to construct a quantum natural property from $\calA^{\ket{f}}$. Let $\calD$ be the quantum algorithm that receives as input $\mathsf{tt}(f)$, the truth-table of $f$, and performs the following steps.
  \begin{enumerate}
    \item  Simulate $\calA^{\ket{f}}$ by
  answering all the queries $\sum_x \alpha_x \ket{x,b} \to \sum_x \alpha_x \ket{x,b \oplus f(x)}$ for arbitrary amplitudes $\{\alpha_x\}_x$ (this can be performed efficiently since $\mathcal{D}$ has access to the truth table of $f$).

    \item If the simulation of $\calA^{\ket{f}}$ does not output a hypothesis $U$, output $1$.
    \item Otherwise, set $T=\gamma(n)^{-3}+100n$ 
    and choose $x_1, \ldots, x_T \in \01^n$ uniformly at random, invoke $U\ket{x_i,0}$ and measure the first bit for every $i\in [T]$; denote by $b_i\in \01$ the output of the $i$th invocation.
    \item Output $0$ \emph{if and only if}
\begin{equation}
\label{eq:satisfyingconditionforD}
    \frac{1}{T}\sum_{i \in [T]} [b_i=f(x_i)] \geq \frac{1}{2} +
  \frac{1}{4}\cdot \gamma(n)\;.
\end{equation}
  \end{enumerate}

We prove that $\calD$ satisfies all three conditions of \cref{defn:qnaturalprop}. Constructiveness follows by noting that simulating $\calA^{\ket{f}}$ can be done in $\poly(2^n)$ time,
and $\calD$ performs $T=\gamma(n)^{-3}+100n$ invocations of the hypothesis $U$ and a check that requires summing over $T$ terms.

We proceed to show the hardness condition. Fix $f \in \mathcal{C}[n^k]$. We show that the algorithm $\calD$ accepts $\mathsf{tt}(f)$ with probability at most $1/3$. By our assumption, with probability at least $9/10$, the learning algorithm $\calA$ outputs $U$ that satisfies \cref{eq:conditionofgoodleanringNP}. Suppose this is the case. Then, for every $i \in [T]$, let $q_i = \Pr[b_i=f(x_i)]$ (where the probability is taken over the randomness of $x_i$ and algorithm $\calA$)  and note that $q_i\geq 1/2+\gamma(n)$. By the Chernoff bound (in Theorem~\ref{lem:chernoff}), we have
 \begin{align}
       \Pr\Big[\frac{1}{T}\sum_i [b_i=f(x_i)]< \frac{1}{2}+\frac{1}{4}\cdot \gamma(n)\Big]\leq
       \exp\left(-\frac{\gamma(n)^2T}{T(\frac{1}{2} + \gamma(n)) + \frac{T\gamma(n)}{12}}\right) \leq
       \exp\big(-O(1/\gamma(n))\big) \;,
   \end{align}
   where in the last inequality we use the fact that $T = \gamma(n)^{-3}+100n$.  Thus, with probability at least $(9/10) \cdot \big(1-\exp\big(-1/\gamma(n)\big)\big)\geq 2/3$, the algorithms $\calD$ rejects.
  
We proceed to show density. For that, let us fix an arbitrary quantum circuit $U$ of size $s(n)$. Also, for a fixed $f, x$, let $W_{f, x} = \Pr_{U} [U(x) = f(x)]]$. Then, we have that 
\[
\E_{f, x}[W_{f, x}] = \frac{1}{2},
\]
since for any fixed $x$, we have that 
\begin{align*}
    \E_{f}\Big[\E_{U} [U(x) = f(x)]\Big] &=
\frac{1}{2}\E_{f}
\left[\E_{U} \big[U(x) = 0|f(x) = 0\big]+ 
\E_{U} \big[U(x) = 1|f(x) = 1\big]\right] \\
&= 
\frac{1}{2}
\E_{f}
\left[\E_{U} \big[U(x) = 0\big]+ 
\E_{U} \big[U(x) = 1\big]\right] \\
&= \frac{1}{2},
\end{align*} 
where in the second equality we have that the fact 
$\Pr_{U} [U(x) = b]$ is independent of $f$. 
We now consider drawing an independent random function $f$ and use the Chernoff Bound (see Lemma~\ref{lem:chernoff2}) to bound the random variable $\sum_{x} W_{f, x}$. 
\begin{align}
    \Pr_f\Big[\Big|\sum_{x\in \01^n} W_{f,x}- \E_{f}\big[\sum_{x\in \01^n}W_{f,x}\big]\Big|\geq t\Big]\leq \exp\Big(-\frac{t^2}{2(t/3+2^n\cdot \frac{1}{2})}\Big).
\end{align}
Setting $t = \gamma(n)\cdot 2^n/8$, we have,
\begin{align*}
    & \Pr_f\Big[\Big|\sum_x W_{f,x}- 2^n/2  \big]\Big|\geq \gamma(n) \cdot 2^n/8\Big] \leq \exp\Big(-\frac{\gamma(n)^2 2^{2n}/64}{2(\gamma(n)\cdot 2^n/24+2^{n}/2)}\Big),
\end{align*}
which implies
$$
 \Pr_f\Big[\Big| \E_x [W_{f, x}]- 1/2 \big]\Big| \geq \gamma(n)/8 \Big] \leq \exp\Big(-\frac{\gamma(n)^2 2^{n}}{64(\gamma(n)/12+ 1)}\Big).
$$
This means that for every quantum circuit $U$ (regardless its size), the fraction of functions that can be computed by it with advantage $\gamma(n)/8$ on a random $x$
is at most  $\exp\Big(-\frac{\gamma(n)^2 2^{n}}{64(\gamma(n)/12+ 1)}\Big)$.

We now count the number of unitaries of size $s(n)$. Observe that there exists a constant $\alpha < 50$ such that
there exist at most $2^{\alpha s(n)\log(s(n))}$ circuits of size $s(n)$ if we fix, for example, the universal gateset $\{\textsf{H}, \textsf{Toff}\}$ as described in \Cref{sec:prelim-quantum}: this is true because one can define,  for the $i$-th gate, if the gate is a Hadamard or a Toffoli, which requires just a single bit, and it takes $O(\log{s(n)})$ bits to denote the qubits on which the gate acts.  Therefore, using a union bound, we have that the fraction of functions that can be $\gamma(n)/8$ approximated by $s(n)$ circuits is at~most
$$
\exp\Big(
-\frac{\gamma(n)^2 2^{n}}{64(\gamma(n)/12+ 1)}
+ \alpha s(n)\log s(n)\Big) \leq
\exp\Big(\gamma(n)^2 2^{n}\left(-\frac{1}{64(\gamma(n)/2+1)} + \frac{O(n)}{\lambda\cdot n}\right)\Big) < 1/2,
$$
where the final inequality used that $\lambda \geq 1$ is a sufficiently large constant as in the theorem statement. 
To conclude the proof, finally observe that for a function $f$ that cannot be approximated by an $s(n)$-sized circuit $U$, i.e., for all $s(n)$-sized circuits $U$
\[
\Pr_{x,U}[U(x) = f(x)] < \frac{1}{2} + \frac{\gamma(n)}{8},
\]
  by the Chernoff bound (in Theorem~\ref{lem:chernoff}), we have that by picking $x_1,\ldots,x_T$ uniformly at random and $b_i \sim U(x)$, we obtain
 \begin{align}
       \Pr\Big[\frac{1}{T}\sum_i [b_i=f(x_i)] \geq  \frac{1}{2}+\frac{1}{4}\cdot \gamma(n)\Big]\leq  \exp\big(-O(1/\gamma(n))\big) \;. 
   \end{align}
Hence the distinguisher $\calD$ accepts it with overwhelming probability. This completes the proof under our initial assumption that $\delta \leq 1/10$. %

In order to obtain a promise quantum natural property from learners that  succeed with probability $1/100$, we can modify the construction above as follows. The  quantum algorithm $\De$ simply runs steps (1)-(3) above $\ell = O(1)$ times, obtaining hypotheses $U_1, \ldots, U_\ell$, and outputs 0 if and only if at least one of them satisfies the condition in step (4). It is not hard to see that the analysis presented above still holds under this modification, provided that $\ell$ is a large enough constant.
\end{proof}

\subsection{Circuit lower bounds from quantum natural properties}
\label{sec:qnp-lb}

In this section, we show that promise quantum natural properties imply circuit lower bounds. This connection between quantum natural properties and circuit lower bounds is demonstrated using pseudorandom generators against uniform circuits.

The key technical component we shall need is the following theorem, which we prove in \Cref{sec:prg-construction}, that shows a quantum-secure $\PRG$ conditioned on the assumption that $\mathsf{PSPACE} \nsubseteq \mathsf{BQSUBEXP}$.

\begin{theorem}[Conditional PRG against uniform quantum computations]\label{t:main_PRG}
Suppose that 
there is a language $L \in \mathsf{PSPACE}$ and $\gamma > 0$ such that $L \notin \mathsf{BQTIME}[2^{n^\gamma}]$. Then, for some choice of constants $\alpha \geq 1$ and $\lambda \in (0,1/5)$, there is an infinitely often $(\ell, m, s, \varepsilon)$-generator $G = \{G_n\}_{n \geq 1}$, where 
$\ell(n) \leq n^\alpha$, $m(n) = \lfloor 2^{n^\lambda} \rfloor$,  $s(m) = 2^{n^{2\lambda}} \geq \mathsf{poly}(m)$ \emph{(}for any polynomial\emph{)}, and $\varepsilon(m) = 1/m$.
\end{theorem}

We shall also need the following diagonalization lemma.
\begin{lemma}[[Diagonalization Lemma, see e.g.~{\cite[Corollary 2]{DBLP:conf/coco/OliveiraS17}}]
\label{lem:diag}
For every $k \in \mathbb{N}$, there exists a language $L_k \in \mathsf{PSPACE}$ such that $L_k$ cannot be computed by Boolean circuits of size $n^k$ whenever $n$ is sufficiently large.
\end{lemma}

Recall that $\mathsf{E} \eqdef \mathsf{DTIME}[2^{O(n)}]$ and $\mathsf{BQE} \eqdef \mathsf{BQTIME}[2^{O(n)}]$. Our goal for the rest of this section is to establish the following  %
connection between quantum natural properties and circuit lower bounds.

\begin{theorem}[Connection between natural properties and lower bounds]
\label{thm:quantum_tight_connection}
For every circuit class $\mathcal{C}$ that is closed under restrictions of input variables to constants $0$ and $1$, at least one of the following lower bounds holds:
\begin{enumerate}
\item For every $k \in \mathbb{N}$ there exists a language $L \in \mathsf{BQE}$ that cannot be computed by general Boolean circuits of size $n^k$, for sufficiently large input lengths $n$.
\item There exists a universal constant $\upsilon \geq 1$ for which the following holds. If there is a promise quantum natural property against $\mathcal{C}[n^a]$, for $a > 1$, then there exists a language $L \in \mathsf{E}$ that cannot be computed by $\mathcal{C}$-circuits of size $n^{a/\upsilon}$, for infinitely many input lengths $n$.
\end{enumerate}
\end{theorem}

\begin{proof}
The proof is via a win-win argument. Starting with the easy case, suppose that for every $\gamma > 0$ it holds that $\mathsf{PSPACE} \subseteq \mathsf{BQTIME}[2^{n^\gamma}]$. Fix $k \in \mathbb{N}$. By~\Cref{lem:diag}, there exists a language $L_k$ that cannot be computed by Boolean circuits of size $n^k$ (for any sufficiently large input length $n$). By our assumption, the language $L_k$ can be computed by a quantum algorithm of time complexity $2^{n^\gamma}$, and so $L_k \in \mathsf{BQE}$, and the lower bound in Item~(1) of \cref{thm:quantum_tight_connection} holds.

Otherwise, suppose that there exists a language $L \in \mathsf{PSPACE}$ and $\gamma > 0$ such that $L \notin \mathsf{BQTIME}[2^{n^\gamma}]$. We show that this implies the lower bound in Item~(2). Indeed, under this assumption, by \cref{t:main_PRG}, there exist constants $\alpha \geq 1$ and $\lambda \in (0,1/5)$ such that there is an infinitely often $(\ell, m, s, \varepsilon)$-generator $\mathcal{G} = \{G_n\}_{n \geq 1}$, where 
$\ell(n) \leq n^\alpha$, $m(n) = \lfloor 2^{n^\lambda} \rfloor$,  $s(m) = 2^{n^{2\lambda}} \geq \mathsf{poly}(m)$, and $\varepsilon(m) = 1/m$; that is,
\begin{enumerate}[label=(\roman*)]
    \item Each $G_n$ is a function mapping $\{0,1\}^{\ell(n)}$ to $\{0,1\}^{m(n)}$.
    \item There is a deterministic algorithm $A$ that when given $1^n$ and $x \in \{0,1\}^{\ell(n)}$ runs in time $O(2^{\ell(n)})$ and outputs $G_n(x)$.
    \item For each $n \geq 1$, let $\mathcal{D}_{m(n)} = G_n(\mathcal{U}_{\ell(n)})$ be the distribution over $m(n)$-bit strings induced by evaluating $G_n$ on a random $\ell(n)$-bit string. Then the corresponding ensemble $\{\mathcal{D}_{m(n)}\}_{n \geq 1}$ is infinitely often $(s, \varepsilon)$-pseudorandom against uniform quantum circuits.
\end{enumerate}

We define a language $L^\mathcal{G}$ using the generator $\mathcal{G} = \{G_n\}_{n \geq 1}$ and argue that it is hard to decide this language due to the \emph{existence} of a quantum natural property, as assumed in Item (2) of the statement.
Our language $L^\mathcal{G}$ is defined by the following deterministic algorithm $D$:

\begin{enumerate}
    \item \textsf{Encoding check:} On an input $u \in \{0,1\}^n$, reject if $u$ is not of the form
    \begin{equation*}
        u = (1^t, w, x)
    \end{equation*}
    for some $t \in \mathbb{N}$ such that $|w| = \ell(t)$ and $|x| = t'$, where we denote $t' = \lfloor t^\lambda \rfloor$. (Note that $2^{t'}$ is the maximal power of $2$ that is not greater than $m(t)$.)\footnote{Notice that there exists at most one value of $t$ for which these constraints hold.}
    \item \textsf{PRG invocation:} Compute $ G_t(w) \in \{0,1\}^{m(t)}$ and let $g_w$ denote its leftmost $2^{t'}$ bits.
    \item \textsf{Emulate computation:} Let $h_w \colon \{0,1\}^{t'} \to \{0,1\}$ be the function determined by $g_w$ (i.e., $\mathsf{fnc}^{g_w}$), and output $h_w(x)$.
\end{enumerate}
To prove the lower bound in Item~(2) of \cref{thm:quantum_tight_connection}, we first observe that $D$ runs in deterministic time complexity $2^{O(n)}$, and hence $L^\calG \in \mathsf{E}$. To see that, observe that by construction, the runtime of $D$ is dominated by Step (2), which computes $G_t(y)$ given a correctly encoded input. Note that for valid inputs we have $n \geq t + \ell(t) + t'$. As $\ell(t) < n$, the time complexity of $D$ is at most $O(2^{n})$.

To complete the proof, we need to prove the following lemma, which shows that the language $L^\mathcal{G}$ is hard for $\mathcal{C}$-circuits. Recall that we are under the assumption that there is a promise quantum natural property against $\mathcal{C}[n^a]$, where $a > 1$ is arbitrary  but can be assumed to be sufficiently large  without loss of generality.

\begin{lemma}
There exists a constant $\upsilon = \upsilon(\alpha,\lambda) \geq 1$, where $(\alpha,\lambda)$ are the parameters of the $\PRG$ $\mathcal{G}$, such that $L^\mathcal{G} \notin \mathcal{C}\!\left[n^{\frac{a}{\upsilon}}\right]$ for infinitely many input lengths $n$.
\end{lemma}

\begin{proof}
Set $\upsilon = 2 \cdot  \frac{\alpha}{\lambda}$, and let $a>1$. We prove the lemma by contradiction: suppose that there exists a finite set $\mathcal{N}$ such that for every $n \in \mathbb{N} \setminus \mathcal{N}$ it holds that $L^\mathcal{G} \in \mathcal{C}\!\left[n^{\frac{a}{\upsilon}}\right] $ over length-$n$ inputs. We argue next that this circuit upper bound and the existence of a quantum natural property as in the statement of the theorem lead to a violation of property (\emph{iii}) of $\mathcal{G}$ stated above.

Let $F$ be a quantum algorithm that gets as input a string of length $N = 2^{n}$ and computes a (promise) quantum natural property in the sense of Definition \ref{defn:qnaturalprop} in $\poly(N)$ time. Recall that~$F$ accepts a dense set $\Gamma_{n}$ of inputs and rejects functions computable by $\calC$-circuits of size $n^a$, both with probability at least $2/3$. Performing error reduction on $F$ using amplitude amplification, we can assume without loss of generality that on each string $z \in \Gamma_{n} \subseteq \{0,1\}^N$, we have $\norm{\Pi_{1} F\ket{z,0}}^2 \geq 1 - 2^{-2N}$, and on each string $z \in \{0,1\}^N$ such that $\mathsf{fnc}^z \in \mathcal{C}_{n}[n^a]$, we have $\norm{\Pi_{1} F\ket{z,0}}^2 \leq 2^{-2N}$. 

We construct a deterministic algorithm $A(1^{m})$ to generate quantum circuits that  violate property (\emph{iii}) of $\mathcal{G}$ (parameterized here by a parameter $t$). 
On an input $1^m$, where $m = m(t) = \lfloor 2^{t^\lambda}\rfloor$, 
$A$ writes $m = 2^{t'} + r$, where $r \in [0,2^{t'} - 1]$. It outputs a quantum circuit $C_F$  on $m$ input bits %
that computes according to $F$ applied to the $2^{t'}$-bit prefix of the input and ignores the remaining $r$ input bits. To complete the proof of the lemma, we show that $A(1^m)$ outputs a quantum circuit $C_F$ of size at most $s(m) = 2^{t^{2\lambda}}$ that $\Omega(1)$-distinguishes $G_t(\mathcal{U}_{\ell(t)})$ from $\mathcal{U}_{m(t)}$ for every large enough $t$, thereby violating property $(iii)$.  %
In other words,
\begin{align}
\left| \E_{y \sim \mathcal{U}_{m(t)}} [C_F(y) = 1] - \E_{\beta \sim \mathcal{U}_{\ell(t)}} [C_F(G_t(\beta) ) = 1] \right| \geq \Omega(1). \label{eq:natpropdistinguisher}
\end{align}

For each $m = 2^{t'} + r$ of the form defined before, %
\begin{itemize}
\item[(a)] $C_F$ rejects with overwhelming probability every $m$-bit string $z$ whose $2^{t'}$-bit prefix encodes a function $f \in \mathcal{C}[(t')^a]$.
\item[(b)] $C_F$ accepts with overwhelming probability every $m$-bit string $z$ whose $2^{t'}$-bit prefix is in $\Gamma_{t'}$.
\item[(c)] $|C_F| = O(\poly(2^{t'})) < s(m)$ as its run time is dominated by the size of $F$ which runs in $\poly(2^{t'})$ time on inputs of size $2^{t'}$.
\end{itemize} 

Note that, as a consequence of the density of $\Gamma_{t'}$ and (b), we have $$
\E_{z \sim \mathcal{U}_{m(t)}}[C_F(z) = 1] = \E_{y \sim \mathcal{U}_{2^{(t')}}}[F(y) = 1] \geq \frac{1}{2} \cdot \E_{y \sim \mathcal{U}_{2^{(t')}}} \big[ F(y) = 1 \mid y \in \Gamma_{t'} \big] \geq \frac{1}{2} \cdot ( 1 - 2^{-2N}).
$$

To satisfy~\Cref{eq:natpropdistinguisher}, it is enough to argue that for each seed $w \in \{0,1\}^{\ell(t)}$ and the string $\tau \eqdef G_t(w)$, we have $\Pr_{\tau} [C_F(\tau) = 1] \leq 2^{-2N}$. According to Item~(a) above, it suffices to show that $g_{\tau}$, the $2^{t'}$-prefix of $\tau$, is the truth-table of a function $h_w \in \mathcal{C}[(t')^a]$. For this we rely on the assumption that $L^\mathcal{G} \in \mathcal{C}[n^{\frac{a}{\upsilon}}]$ for every input length $n \in \mathbb{N} \setminus \mathcal{N}$. Details follow. %

Let $n = t + \ell(t) + t' + O(1) \leq 2 t^\alpha < 2 (t')^{\alpha/\lambda}$ such that $n \in \mathbb{N} \setminus \mathcal{N}$, where $n$ is an input length of $L^\mathcal{G}$ that contains inputs of the following form: $(1^t, w, x)$, where $x$ is an arbitrary $t'$-bit string, and $w$ is the seed that generates $\tau$. Let $g_{\tau}$ be the $2^{t'}$-bit prefix of $\tau$ and $h_{\tau} \eqdef \mathsf{fnc}^{g_{\tau}}$ be a Boolean function on $t' < n$ inputs bits. Then, by the assumption that $L^\mathcal{G} \in \mathcal{C}[n^{\frac{a}{\upsilon}}]$ where $\calC$ is a circuit class closed under restriction of input variables to constants $0$ and $1$ (see Item (ii) in~\Cref{sec:prelim-circuits}), we obtain a $\mathcal{C}$-circuit that computes $h_{\tau}(x)$.  Furthermore, this circuit is of size at most 
\[
n^{a/\upsilon} \; \leq \left(2 t^\alpha \right)^{a/\upsilon} \; < \left( 2 \left(t'^{\frac{\alpha}{\lambda}} \right) \right)^{a/\upsilon} \leq 2^{a/\upsilon} \cdot (t')^{\frac{\alpha}{\lambda} \cdot \frac{a}{\upsilon}} < (t')^{2 \cdot \frac{\alpha}{\lambda} \cdot \frac{a}{\upsilon}} = (t')^a,
\]
where we have used that $2 z \leq z^2$ for $z \geq 2$ in the penultimate inequality and $2 \cdot \frac{\alpha}{\lambda} \cdot \frac{a}{\upsilon} \leq a$ as $\upsilon = 2 \cdot \frac{\alpha}{\lambda}$ in the last equality. %
\end{proof}
This concludes the proof of \cref{thm:quantum_tight_connection}.
\end{proof}

As an immediate corollary, we obtain the following theorem.

\begin{corollary}[Lower bounds from quantum natural properties]
\label{cor:lbs-from-qnp}
Let $\mathcal{C}$ be a circuit class. %
Suppose that for every $a \geq 1$ there is a promise quantum natural property against $\mathcal{C}[n^a]$. Then for every constant $b \geq 1$ we have $\mathsf{BQE} \nsubseteq \mathcal{C}[n^b]$.
\end{corollary}
\begin{proof}
From~\Cref{thm:quantum_tight_connection}, note that there are two possibilities and at least one of them holds. 

In the first case, for every $k \in \mathbb{N}$, there exists a language $L \in \BQE$ that cannot be computed by general Boolean circuits of size $n^k$. Since by assumption $\mathcal{C}$-circuits can be simulated by general Boolean circuits with only a polynomial blowup on circuit size (Section \ref{sec:prelim-circuits}), it follows in this case that for every $b \geq 1$ we have $\BQE \not\subseteq \calC[n^b]$.

In the second case, given $b \geq 1$, we let $a = b \upsilon$. Since by assumption there is a natural property against $\mathcal{C}[n^a]$, the second case of Theorem \ref{thm:quantum_tight_connection} gives a language in $\mathsf{E}$ that is not in $\mathcal{C}[n^{a/\upsilon}]$. Since $\mathsf{E} \subseteq \mathsf{BQE}$, the result follows.
\end{proof}

\subsection{Non-trivial quantum learning yields non-uniform  circuit lower bounds}
\label{sec:learning-to-lbs}

\begin{theorem}[\cref{thm:main}, formally restated]
There is a universal constant $\lambda \geq 1$ for which the following holds. Let $\mathfrak{C}$ be a concept class. Let $\gamma \colon \mathbb{N} \to [0,1/2]$ and $T:\mathbb{N}\rightarrow \mathbb{N}$ be arbitrary constructive functions with $\gamma(n) \geq \lambda\cdot 2^{-n/2}$ and $T(n)\leq \gamma(n)^2\cdot 2^n/\lambda n$ for large enough $n$. Suppose that, for every $k \geq 1$, the class $\mathfrak{C}[n^k]$ can be learned in quantum time $T(n)$ with probability~$\geq~1/100$ and advantage $\gamma(n)$. Then, for every $k \geq 1$, we have $\mathsf{BQE} \nsubseteq \mathfrak{C}[n^k]$.
\end{theorem}

\begin{proof}
This proof combines results from the previous two sections to relate quantum learning algorithms to circuit lower bounds through the existence of quantum natural properties.
 The assumptions for the algorithm that learns $\mathfrak{C}[n^k]$ imply that it runs in time at most $T(n) \leq \gamma(n)^2 \cdot 2^n/\lambda n$, has a confidence probability $1 - \delta(n) \geq 1/100$ (i.e., $\delta \leq 0.99$) and error probability $\varepsilon(n) \leq 1/2 - \gamma(n)$. Then, from~\Cref{lem:learning_to_natural}, for every constant $k \geq 1$, there is a promise quantum natural property against $\mathfrak{C}[n^k].$ Following our notation from~\Cref{sec:prelim-circuits}, the concept class $\mathfrak{C}[n^k]$ also denotes a corresponding circuit class. %
Using~\Cref{cor:lbs-from-qnp}, the existence of a quantum natural property against $\mathfrak{C}[n^k]$ for every $k \geq 1$ implies that for every constant $b \geq 1$, $\BQE \nsubseteq \mathfrak{C}[n^b]$. This completes the proof.
\end{proof}

Similarly to previous work \citep{DBLP:conf/coco/OliveiraS17}, our arguments can be adapted to show a relation between the non-trivial learnability of a concept class $\mathcal{C}[s]$ of size-$s$ circuits, where $s(n) = n^{\omega(1)}$, and lower bounds against $\mathcal{C}[s']$, where $s, s' \colon \mathbb{N} \to \mathbb{N}$ and $s'$ is a function that depends on $s$. We have decided not to pursue the most general form of the result in this paper, as our proofs are already significantly involved and the relation between $s$ and $s'$ deteriorates as $s$ becomes super-quasi-polynomial, due to the use of win-win arguments.

\section{Technical tools}
\label{sec:tech-tools}
In this section, we discuss extensions to quantum computing of several fundamental results from complexity theory. This is needed to establish the correctness of our $\PRG$ construction.

In \Cref{sec:nw}, we provide a (fairly straightforward) quantization of \cite{nisan1994hardness}, where we show that a quantum distinguisher for the (candidate) $\PRG$ $\mathsf{NW}^f$ implies that $f$ can be non-trivially approximated by quantum algorithms.
In \Cref{sec:gl}, we give a self-contained exposition of a  quantum analogue of the Goldreich-Levin algorithm~\cite{DBLP:conf/stoc/GoldreichL89} discovered by \cite{adcock2002quantum}.
In \Cref{sec:IJKW}, we quantise the near-optimal uniform hardness amplification result of \cite{impagliazzo2010uniform}. We remark that this is the most technically involved result in this section. Finally, in \Cref{sec:self_reduc}, we show how to use downward and random self-reducibility of languages to devise quantum algorithms to compute them, adapting ideas from \citep{DBLP:journals/jcss/ImpagliazzoW01}.

\subsection{Nisan-Wigderson generator against quantum adversaries}

\label{sec:nw}

An ordered family $\mathcal{S} = (S_1, \ldots, S_t)$ of sets is a $(t, m, n, \alpha)$-\emph{design} if the following conditions are satisfied:
\begin{itemize}
    \item $|\mathcal{S}| = t$, and for each $i \in [t]$, $S_i \subseteq [m]$ and $|S_i| = n$.
    \item For every distinct pair $i, j \in [t]$, $|S_i \cap S_j| \leq \alpha$.
\end{itemize}

\begin{lemma}[{\citep[Lemma 2.5]{nisan1994hardness}}]\label{l:NW_design}
There is an absolute constant $c \geq 1$ for which the following holds. For any positive integers $n$ and $t$ such that $n \leq t \leq 2^n$, there is an ordered family $\mathcal{S}$ that is a $(t, cn^2, n, \log t)$-design. Moreover, given $n$, $t$, and an index $i \in [t]$, the set $S_i$ can be output in time $\mathsf{poly}(t)$.
\end{lemma}

\begin{definition}[Nisan-Wigderson generator~\cite{nisan1994hardness}]
\label{def:nwprg}
Let $f \colon \01^n\rightarrow \01$, $m=cn^2$, and $n \leq t \leq 2^n$. In addition, let $\mathcal{S}$ be a {$(t, m, n, \log t)$-design}. Then the \emph{Nisan-Wigderson generator} $\NW^f_\mathcal{S} \colon \01^m\rightarrow \01^t$ is the function defined as
$$
\NW^f_\mathcal{S}(z) \eqdef f(z\vert_{S_1})  f(z\vert_{S_2})\cdots f(z\vert_{S_t}),
$$
where $z\vert_{S_i}$ is the $n$-bit string formed by restricting $z \in \{0,1\}^m$ onto the coordinates given by the $i$-th set $S_i \in \mathcal{S}$.
\end{definition}

From now on we will only consider the designs given by Lemma \ref{l:NW_design}, so in order to simplify notation, we omit the underlying collection $\mathcal{S}$ and simply write $\mathsf{NW}^f$.

The next lemma verifies that the usual analysis of the Nisan-Wigderson generator extends to inherently probabilistic circuits.  

\begin{lemma}[Uniformity and correctness of the NW generator for inherently probabilistic circuits]\label{lem:nw}
Let $n$ and $t$ be positive integers such that $1 \leq n  \leq t \leq 2^n$, $f \colon \01^n\rightarrow \01$, and consider the corresponding function $\NW^f \colon \01^m\rightarrow \01^t$, where $m  = c n^2$. Let $\mathcal{D}$ be an inherently probabilistic circuit defined over $t$ input bits that satisfies
\begin{align}\label{eq:distinguisher-NW}
\left| \Pr_{s\in \01^m,\,\mathcal{D}} \big [\mathcal{D}(\NW^f(s)) = 1 \big ]
- \Pr_{y\in \01^t,\,\mathcal{D}} \big [\mathcal{D}(y) = 1\big ]
\right | \;\geq\; \gamma.
\end{align}
There exists a probabilistic algorithm $\calB$ with oracle access to $f$ that, given an input $1^t$, runs in time $\mathsf{poly}(t)$ and outputs a deterministic oracle circuit $\mathcal{A}^{\mathcal{O}}$ of size $O(t^2)$ for which the following holds. For any choice of $\mathcal{D}$ as above, with probability $\Omega(\gamma/t^2)$ over the internal randomness of $\mathcal{B}^f$, the generated circuit $\mathcal{A}^\mathcal{O}$ satisfies 
\begin{align}\label{eq:conclusion-nw}
\Pr_{x \in \{0,1\}^n,\,\mathcal{D}}\left[\calA^{\mathcal{D}}(x) = f(x)\right] 
 \geq \frac{1}{2} + \frac{\gamma}{2t}.
\end{align}
\end{lemma}

\begin{proof}
We show that there exists a collection of deterministic oracle circuits $\calA^{\mathcal{O}}$ such that, for any choice of $\mathcal{D}$, a noticeable fraction of such oracle circuits provide the desired advantage. The argument also establishes the existence of a uniform algorithm $\calB^f$ with the required properties. Note that $\mathcal{B}^f$ does not have access to $\mathcal{D}$ and that the computation of each $\mathcal{A}^{\mathcal{O}}$ does not specify its oracle. However, in order to make the argument more concrete, in our exposition below we refer to the oracle $\mathcal{O}$ as $\mathcal{D}$, and consider a fixed $\mathcal{D}$ satisfying the assumption of the lemma. After that, we explain how the conclusion of the result follows from this argument.

We prove first that there exists a \emph{fixed} $j \in [t]$ (depending on $\mathcal{D}$) and a sub-collection of deterministic oracle ``next-bit predictor'' circuits $\mathcal{P}_{r,d}^{\calD}$ parameterized by $r \in \01^{t-j}$ and a \emph{fixed} $d \in \{0,1\}$ (depending on $\mathcal{D}$), each of size $O(t)$, such that
\begin{align}
\label{eq:constofP}
\Pr_{s\in \01^m, r, \calD}\left[\calP_{r,d}^\calD \big( \NW^f(s)_1,\ldots,\NW^f(s)_{j-1} \big) = \NW^f(s)_{j}\right] \geq 1/2 + \gamma/t.
\end{align}
For notational simplicity, from here onwards we abuse notation and write $z\sim \NW^f(s)$, meaning that $s\sim  \01^m$ is uniformly random and $z=\NW^f(s)$. Additionally, for $i\in [t]$ and $z\in \01^t$, we let ${z}^{(i)}$ be a \emph{random variable} with the first $i$ bits being equal to $z$ and the last $t - i$ bits being a uniformly random string $r \in \{0,1\}^{t - i}$. Also, let $$p_i = \Pr_{z \sim \NW^f(s),\calD}\big[\calD({z}^{(i)}) = 1\big],$$
and observe that 
$$
p_0 = \Pr_{y\in \01^t,\calD}[\mathcal{D}(y) = 1] \quad \text{and} \quad  p_t = \Pr_{z \sim \NW^f(s),\calD }[\mathcal{D}(z) = 1].
$$
Therefore, by Eq.~\eqref{eq:distinguisher-NW}, we have $|p_t - p_0| \geq \gamma$. By the triangle inequality, there exists $j \in \{0,1, \ldots, t-1\}$ such that $|p_{j+1} - p_j| \geq \gamma/t$. 
Setting $\hat{z} ={z}^{(j)}$ for convenience, notice that, %
\begin{align*}
\gamma/t ~&\leq
|p_{j+1} - p_j| \\
&=
\Bigg|\Bigg(\E_{z \sim \NW^f(s),\calD, z^{(j+1)} }\big[\mathcal{D}(z^{(j+1)}) =1\big] - \E_{z \sim \NW^f(s),\calD, \hat{z} }\big[\mathcal{D}(\hat{z}) =1\big]\Bigg)\Bigg| \\
&=
\Bigg|\Bigg(\E_{z \sim \NW^f(s),\calD, \hat{z} }\big[\mathcal{D}(\hat z) =1| \hat{z}_{j+1} = z_{j+1}\big] \\
&\quad - 
\frac{1}{2}\left(
\E_{z\sim \NW^f(s),\calD,\hat{z} }\big[\mathcal{D}(\hat z)  = 1| \hat{z}_{j+1} = z_{j+1}\big]
+\E_{z \sim \NW^f(s),\calD,\hat{z} }\big[\mathcal{D}(\hat z) = 1 | \hat{z}_{j+1} \ne z_{j+1}\big]\right)\Bigg)\Bigg|
\\  
&
=  
\left|\frac{1}{2}\left(
\E_{z\sim \NW^f(s),\calD,\hat{z} }\big[\mathcal{D}(\hat z) = 1 | \hat{z}_{j+1} = z_{j+1}\big]
-\E_{z\sim \NW^f(s),\calD,\hat{z} }\big[\mathcal{D}(\hat z) = 1 | \hat{z}_{j+1} \ne z_{j+1}\big]\right)\right|
\\
&=  \left| \frac{1}{2}\left(\E_{z \sim \NW^f(s),\calD,\hat{z}}\big[\mathcal{D}(\hat z) = 1 | \hat{z}_{j+1} = z_{j+1}\big] +
\E_{z\sim \NW^f(s), \calD, \hat{z} }\big[\mathcal{D}(\hat z) \ne 1 | \hat{z}_{j+1} \ne z_{j+1}\big] - 1 \right) \right|,%
\end{align*}
where in the second equality we use three simple facts:
\[\E_{z \sim \NW^f(s),\calD, z^{(j+1)} }\big[\mathcal{D}(z^{(j+1)}) =1\big] = \E_{z \sim \NW^f(s),\calD, \hat{z} }\big[\mathcal{D}(\hat z) =1| \hat{z}_{j+1} = z_{j+1}\big]\]
and 
\begin{align*}
&\E_{z \sim \NW^f(s),\calD, \hat{z} }\big[\mathcal{D}(\hat{z}) =1\big]\\
&= 
\Pr_{z \sim \NW^f(s),\calD, \hat{z}}[\hat{z}_{j+1} = z_{j+1}]
\E_{z \sim \NW^f(s),\calD, \hat{z} }\big[\mathcal{D}(\hat{z}) =1\big | \hat{z}_{j+1} = z_{j+1}] \\
&\quad+ 
\Pr_{z \sim \NW^f(s),\calD, \hat{z}}[\hat{z}_{j+1} \ne z_{j+1}]
\E_{z \sim \NW^f(s),\calD, \hat{z} }\big[\mathcal{D}(\hat{z}) =1\big | \hat{z}_{j+1} \ne z_{j+1}]
\end{align*}
and $\Pr_{z \sim \NW^f(s),\calD, \hat{z}}[\hat{z}_{j+1} = z_{j+1}]
 = \Pr_{z \sim \NW^f(s),\calD, \hat{z}}[\hat{z}_{j+1} \ne z_{j+1}]
 = \frac{1}{2}$, since $\hat{z}_{j+1} = r_1$ is a uniformly random bit.

The inequality above motivates the following definition. Let $d \in \01$ be such that 
$(-1)^d(p_{j+1} - p_j) > 0$.
We define the ``next-bit predictor'' $\mathcal{P}^\calD_{r,d}$, for $r \in \{0,1\}^{t-j}$, as follows. On input $z_1,\ldots,z_{j}$ (where $z\sim \NW^f(s)$), $\mathcal{P}_{r,d}^\calD$ makes an oracle call to~$\mathcal{D}$ on input $z_1,...,,z_{j},r_1,...,r_{t-j}$ and let $b$ be its output. If the output $b = 1$, then $\mathcal{P}$ outputs $r_1 \oplus d$, otherwise $\mathcal P$ outputs $r_{1}\oplus d \oplus 1$. It follows from the inequality above, the definition of $\mathcal{P}_{r,d}^\calD$, and a simple manipulation (see e.g.,~\citep[Proof of Theorem 9.11]{book_complexity}) that 
\begin{align*}
&\E_{z \sim \NW^f(s), r, \calD}\big[\calP_{r,d}^\calD(z_1,\ldots,z_{j}) = z_{j+1}\big] - \frac{1}{2}\\
&\;=\;
\frac{(-1)^d}{2}\left(\E_{z \sim \NW^f(s),\calD,\hat{z}}\big[\mathcal{D}(\hat z) = 1 | \hat{z}_{j+1} = z_{j+1}\big] +
\E_{z\sim \NW^f(s), \calD, \hat{z} }\big[\mathcal{D}(\hat z) \ne 1 | \hat{z}_{j+1} \ne z_{j+1}\big] -1 \right)\\
&\;\geq\; \gamma/t.
\end{align*}
This concludes the construction of a collection of deterministic oracle circuits $\calP_{r,d}^\calD$ satisfying   Eq.~\eqref{eq:constofP}.

We now construct an algorithm $\calA_{j,r,r',d}^\calD(x)$ that has hardwired on its code the $(t-j)$-bit (random) string $r$ from above and another  $(m - n)$-bit (random) string $r'$ introduced next. We will show that, on average over the choices of $r$ and $r'$, the algorithm computes $f(x)$ on a random $x$ with noticeable probability. For each $i \neq j+ 1$, we now assume (and subsequently prove) the existence of a deterministic circuit $C_{i,r'}(x)$ of size $O(t)$ that computes $\NW^f(s)_{i}$, where $s \in \01^m$ is defined as  $s|_{S_{j+1}} = x$ and $s|_{\overline{S_{j+1}}} = r'$.

The circuit $\calA_{j,r,r',d}^\calD$ on an input $x$ invokes $\calP^{\mathcal{D}}_{r,d}$ as follows. $\calA_{j,r,r',d}^\calD(x)$ fixes seed $s$ consisting of $s|_{S_{j+1}} = x$ and $s|_{\overline{S_{j+1}}} = r'$, computes
$\NW^f(s)_1,\ldots,\NW^f(s)_{j}$  using the corresponding circuits $C_{i,r'}(x)$, and outputs $\calP^{\calD}_{r,d}\big(\NW^f(s)_1,\ldots,\NW^f(s)_{j}\big)$.

Notice that, averaging over the random choices of $r$ and $r'$, and writing $s$ as the random string obtained from random choices of $x$ and $r'$, we have:
\begin{align}
\E_{r,r'}\left[\Pr_{x,\calD}[\calA_{j,r,r',d}^\calD(x) = f(x)]\right]&= 
\E_{r,r'}\left[\Pr_{\substack{x\in \01^n}, \calD}\big[\calP_{r,d}^\calD\big(\NW^f(s)_1,\ldots,\NW^f(s)_{j}\big)=f(x)\big]\right]
\nonumber
\\
&=
\Pr_{z \sim \NW^f(s),r,\calD}\big[\calP^\calD_{r,d}(z_1,\ldots,z_{j})= z_{j+1}\big]
 \geq \frac{1}{2} + \gamma/t. \label{eq:nw-average}
\end{align}
Consequently, we get that
\begin{align}
\Pr_{r,r'}\left[\Pr_{x,\calD}[\calA_{j,r,r',d}^\calD(x) \ne f(x)] \geq  \frac{1}{2} - \gamma/2t\right] 
& \;\leq\; \frac{
\E_{r,r'}\left[\Pr_{x,\calD}[\calA^\calD_{j,r,r',d}(x) \ne f(x)]\right]}{\frac{1}{2} - \gamma/2t}  \nonumber \\
& \;\leq\; \frac{
\frac{1}{2} - \gamma/t}{\frac{1}{2} - \gamma/2t}   \;=\; \frac{
1 - 2\gamma/t}{1 - \gamma/t}   \;\leq\; 1 - \gamma/t,  \label{eq:markov-nw}
\end{align}
where the first inequality comes from Markov's inequality, the second inequality comes from \Cref{eq:nw-average}, and the last inequality comes from the fact that $1-2x \leq (1-x)^2$ for every $x \in \mathbb{R}$.

Using the assumption about the size of $C_{i,r'}$, we have that the size of $\calA^{\mathcal{O}}$ is trivially $O(t^2)$. We prove next the existence of circuits $C_{i,r}(x)$ of the claimed size. First notice that $\NW^f(s)_{i}$ depends on at most $\log{t}$ bits of $s|_{S_{j+1}} = x$, by \Cref{l:NW_design}. Therefore, for fixed $i$ and $r'$,  we can define $C_{i, r'}(x)$ as a lookup-table of size $O(t)$ that stores the corresponding values of $f$ for all possible choices of the $\leq \log t$ relevant bits of $x$.

Finally, we define the uniform algorithm $\calB^f(1^t)$ as follows. First, it picks $j^* \in [t-1]$, $d^* \in \01$, $r \in \01^{t-j}$ and $r' \in \01^{m-n}$ uniformly at random. Then $\calB^f$ outputs the oracle circuit $\calA_{j^*,r,r',d^*}^{\mathcal{O}}$,  where the lookup table of each relevant circuit $C_{i, r'}$ can be computed by $\calB$ using membership queries to its oracle $f$. Note that the computation of $\mathcal{B}^f$ and the description of $\mathcal{A}^{\mathcal{O}}$ are indeed oblivious to the particular choice of $\mathcal{D}$.

Since for any fixed $\mathcal{D}$ as in the hypothesis of the result good choices of $j^*$ and $d^*$ are made with probability at least $\frac{1}{2t}$, and in this case $r$ and $r'$ yield a circuit $\calA^\calD_{j^*,r,r',d^*}$ with the desired properties with probability at least $\gamma/t$ by \Cref{eq:markov-nw}, we have that for any admissible $\mathcal{D}$, with probability at least $\Omega(\gamma/t^2)$ over its internal randomness $\mathcal{B}^f$ outputs a deterministic oracle circuit $\calA^{\mathcal{O}}_{j^*,r,r',d^*}$ that satisfies \Cref{eq:conclusion-nw}. (Note that the analysis shows that the construction succeeds with the desired probability for any fixed choice of the inherently probabilistic circuit $\mathcal{D}$, though different random choices of the parameters $d$, $r$, and $r'$ might be needed as a function of $\mathcal{D}$.)
\end{proof}

We now establish a quantum analogue of this result, stated in a way that is convenient for our application.

\begin{lemma}[Uniform Nisan-Wigderson reconstruction for quantum circuits]\label{lem:nw_quantum} Let $s_F, s_D, t \colon \mathbb{N} \to \mathbb{N}$ and $\gamma \colon \mathbb{N} \to [0,1]$ be constructive functions, where $n \leq t(n) \leq 2^n$ for every $n$. There exists a sequence $\{\mathcal{C}_n^{\mathsf{NW}}\}_{n \geq 1}$ of quantum circuits $\mathcal{C}_n^{\mathsf{NW}}$ such that:
\begin{itemize}
    \item[\emph{(}i\emph{)}] \emph{Input:} Each circuit $\mathcal{C}_n^{\mathsf{NW}}$ gets as input strings $1^n$, $1^{t(n)}$, $\mathsf{code}(D)$, and $\mathsf{code}(F)$, where $\mathsf{code}(D)$ encodes a quantum circuit $D$ of size $\leq s_D(n)$ over $t(n)$ input bits, and $\mathsf{code}(F)$ encodes a quantum circuit $F$ of size $\leq s_F(n)$ over $n$ input bits.
    \item[\emph{(}ii\emph{)}] \emph{Uniformity and Size:} Each circuit $\mathcal{C}_n^{\mathsf{NW}}$ is of size $S(n) = \mathsf{poly}(t(n), s_F(n), s_D(n))$, and there is a uniform deterministic algorithm that given $1^n$ runs in time $\mathsf{poly}(S(n))$ and prints the string $\mathsf{code}(\mathcal{C}_n^{\mathsf{NW}})$.
    \item[\emph{(}iii\emph{)}] \emph{Output and Correctness:} Suppose that the input circuit $F$ computes a Boolean function $f \colon \{0,1\}^n \to \{0,1\}$, and consider the associated function $\mathsf{NW}^f \colon \{0,1\}^{m(n)} \to \{0,1\}^{t(n)}$, where $m(n) = cn^2$. In addition, assume that the input circuit $D$ satisfies
    \begin{align}\label{eq:distinguisher-NW_quantum}
\left| \Pr_{s\in \01^{m(n)},\,D} \big [D(\NW^f(s)) = 1 \big ]
- \Pr_{y\in \01^{t(n)},\,D} \big [D(y) = 1\big ]
\right | \;\geq\; \gamma(n).
\end{align}
Then with probability $\Omega(\gamma(n)/t(n)^2)$ over its output measurement, $\mathcal{C}_n^{\mathsf{NW}}$ produces a string $\mathsf{code}(A)$ encoding a quantum circuit $A$ of size $O(t(n)^2 \cdot s_D(n))$ such that   
\begin{align}\label{eq:conclusion-nw-quantum}
\Pr_{x \in \{0,1\}^n,\,A}\left[A(x) = f(x)\right] \eqdef \E_{x \sim \{0,1\}^n,\; A}\Big[ \norm{\Pi_{f(x)}A\ket{x}\ket{0^{\ell}}}\Big ] 
 \;\geq\; \frac{1}{2} + \frac{\gamma(n)}{2t(n)}.
\end{align}
\end{itemize}
\end{lemma}

\noindent \textbf{Note.}~We stress that the result is non-trivial because the size of $A$ is independent of $s_F(n)$ (otherwise it would be sufficient for $\mathcal{C}_n^{\mathsf{NW}}$ to output the string $\mathsf{code}(F)$, which encodes a quantum circuit that computes $f(x)$ on every input $x$ with probability $\geq 2/3$).

\begin{proof}[Proof of Lemma \ref{lem:nw_quantum}]
We follow the argument in the proof of Lemma \ref{lem:nw}. Here the circuit $\mathcal{C}_n^{\mathsf{NW}}$ takes the role of $\mathcal{B}^f$, the oracle $f$ is replaced by the input circuit $F$, and the input circuit $D$ behaves as the inherently probabilistic circuit $\mathcal{D}$.

Under the assumption that the quantum circuit $D$ distinguishes the output of $\mathsf{NW}^f(s)$ on a random seed $s$ from a random string $y$ (Equation \ref{eq:distinguisher-NW_quantum}), and proceeding as in the proof of Lemma \ref{lem:nw_quantum} while replacing $\mathcal{D}$ with $D$, it follows from Lemma \ref{lem:quantum_as_inherently_random} that the same probability analysis holds. Consequently, there is a quantum circuit $A$ that employs $D$ as a sub-routine and computes $f$ with advantage $1/2 + \gamma/2t$. The only relevant difference here is that in order to generate with the desired probability a correct description of $A$, it is necessary to evaluate the function $f$ on different input strings (recall that we hardwire the corresponding answers in the circuits $C_{i,r'}$, which appear as sub-circuits of $A$). To achieve that, $\mathcal{C}_n^{\mathsf{NW}}$ simulates the input quantum circuit $F(x)$ (using a universal quantum circuit), amplifying its success probability of computing $f(x)$ via standard techniques in a way that reduces the probability that $F$ errs on \emph{any} string employed during these simulations to less than $1/2$. Whenever the correct values of $f$ needed in the circuits $C_{i,r'}$ are produced, we get that with probability $\Omega(\gamma(n)/t(n)^2)$ over its output measurement,  $\mathcal{C}_n^{\mathsf{NW}}$ successfully generates a ``good'' quantum circuit $A$ from circuits $F$ and $D$ and its internal random choices, i.e.,  Equation \ref{eq:conclusion-nw-quantum} holds for $A$.

While in the proof of Lemma \ref{lem:nw} the size of each deterministic oracle circuit $\mathcal{A}^\mathcal{O}$ is $O(t^2)$, here the number of gates in the corresponding quantum circuit $A$ is $O(t(n)^2 \cdot s_D(n))$, since $A$ explicitly incorporates the computation of circuit $D$, which by assumption has size at most $s_D(n)$.

We now discuss the uniformity and size of each quantum circuit $\mathcal{C}_n^{\mathsf{NW}}$. Recall that  $\mathcal{C}_n^{\mathsf{NW}}$ operates in a way that is analogous to algorithm $B^f$ from Lemma \ref{lem:nw}. A bit more formally, $\mathcal{C}_n^{\mathsf{NW}}$ must output a string $\mathsf{code}(A)$ from strings $\mathsf{code}(F)$ and $\mathsf{code}(D)$ (and the auxiliary parameters $1^n$ and $1^{t(n)}$). The computation of $\mathcal{C}_n^{\mathsf{NW}}$ involves manipulating the codes of $F$ and $D$ to produce the code of $A$, and includes the simulation of the quantum circuit $F$ on at most $t(n) \cdot t(n)$ distinct input strings (we have at most $t(n)$ circuits $C_{i, r'}$, and each of them stores the value of $f$ on at most $t$ inputs), with an overhead for the amplification of the success probability. Since the descriptions of $F$ and $D$ have length $\mathsf{poly}(s_F(n))$ and $\mathsf{poly}(s_D(n))$, respectively, and each quantum simulation can be done using $\mathsf{poly}(s_F(n))$ gates, it follows that each quantum circuit $\mathcal{C}_n^{\mathsf{NW}}$ can be implemented with $S(n) = \mathsf{poly}(t(n),s_F(n), s_D(n))$ gates. Finally, it is not hard to see that the code of $\mathcal{C}_n^{\mathsf{NW}}$ is explicit and can be deterministically generated from $1^n$ in time $\mathsf{poly}(S(n))$, since the description of $B^f$ in the proof of Lemma \ref{lem:nw} is also explicit. 
\end{proof}

\subsection{Goldreich-Levin lemma in the quantum setting}
\label{sec:gl}
\begin{lemma}\label{lem:gl}
Let $f\colon \01^{kn} \to \01^k$. Suppose there is a quantum circuit  $\mathcal{A}$ \emph{(}that uses $m\geq 1$ workspace qubits\emph{)} satisfying 
\[
\E_{x\in \01^{kn}}\E_{r\in\01^k}\left[\norm{\Pi_{f(x) \cdot r}\calA\ket{x,r,0^m}}^2\right] 
 \geq \frac{1}{2} + \gamma.
\]
Then there is a quantum oracle circuit $\calB^{\calA}$ of size $O(kn)$ that has oracle access to $\calA$ \emph{(}and to its inverse $\calA^{\dagger}$\emph{)} such that 
\[
\E_{x,\calB^{\calA}}\left[\norm{\Pi_{f(x)}\calB^{\calA}\ket{x,0^{k+m+1}}}^2\right] 
 \geq \frac{\gamma^3}{2}.
\]
Moreover, a quantum circuit $\mathcal{B}$ of this form can be constructed from a quantum circuit $\mathcal{A}$ as above by a uniform sequence of circuits, which we formalise as follows. Let $k, s_{\mathcal{A}} \colon \mathbb{N} \to \mathbb{N}$ and $\gamma \colon \mathbb{N} \to [0,1/2]$ be constructive functions. In addition, let $f_n \colon \{0,1\}^{kn} \to \{0,1\}^k$ be a sequence of functions, where $k = k(n)$. Then there is a sequence $\{\mathcal{C}^{\mathsf{GL}}_n\}_{n \geq 1}$ of deterministic circuits $\mathcal{C}^{\mathsf{GL}}_n$ of size $\mathsf{poly}(n, k(n), s_{\mathcal{A}}(n))$ for which the following holds. If $\mathcal{C}^{\mathsf{GL}}_n$ is given as input strings $1^n$ and a description $\mathsf{code}(\mathcal{A})$ of a quantum circuit $\mathcal{A}$ of size $\leq s_{\mathcal{A}}(n)$ with the property above, then it outputs a string $\mathsf{code}(\mathcal{B})$ describing a quantum circuit $\mathcal{B}$ of size $O(n \cdot k(n) \cdot s_{\mathcal{A}}(n))$ such that
\[
\E_{x,\calB}\left[\norm{\Pi_{f_n(x)}\calB\ket{x,0^{k+m+1}}}^2\right] 
 \geq \frac{\gamma(n)^3}{2}.
\]
\end{lemma}

\begin{proof}
Without loss of generality, let us assume that $\calA$ measures the first qubit of $\calA\ket{x,r,0^m}$ and produces it as an output.
We define the algorithm $\calB^{\calA}$ on input $x$ as follows:
\begin{enumerate}
    \item Start with $\frac{1}{\sqrt{2^k}}\sum_{r}\ket{x,r,0^m,1}$    \label{item:first_step}
    \item Apply $\mathcal{A}$ on the first $kn+k+m$ qubits. 
    \item Apply a control-Z operation with the ``output of $\calA$'' (i.e., first qubit) as the control qubit and the last qubit as the target qubit.
    \item Apply $\mathcal{A}^\dagger$ on the first $kn+k+m$ qubits. 
    \item Apply Hadamard transform on the second register.
    \item Measure all qubits in the computational basis.
    \item If the outcome is of the form $\ket{x,a,0^{k+m},1}$, output $a$.  \label{item:last_step}
\end{enumerate}
To analyze its correctness, we first define $G = \{x : \E_{r\in\01^k}\left[\norm{\Pi_{f(x) \cdot r}\calA\ket{x,r,0^m}}^2\right] 
 \geq \frac{1}{2} + \frac{\gamma}{2} \}$. We will show that 
  \begin{align}\label{eq:size-g}
  |G| \geq \frac{\gamma}{2} \cdot 2^n,
 \end{align}
 and that for every $x \in G$, we have
\begin{align}\label{eq:implication-g}
\E_{\calB^{\calA}}
\left[\norm{\Pi_{f(x)}\calB^{\calA}\ket{x,0^{k+m+1}}}^2\right] 
 \geq \gamma^2.
 \end{align}
Before proving \Cref{eq:size-g,eq:implication-g}, notice that they directly imply our statement, since
\begin{align*}
\E_{x,\calB^{\calA}}
\left[\norm{\Pi_{f(x)}\calB^{\calA}\ket{x,0^{k+m+1}}}^2\right] 
&\geq
\frac{|G|}{2^n}\E_{x\in G,\calB^\calA}
\left[\norm{\Pi_{f(x)}\calB^\calA\ket{x,0^{k+m+1}}}^2\right] \\ &\geq   \frac{\gamma}{2} \E_{x\in G,\calB^\calA}
\left[\norm{\Pi_{f(x)}\calB^\calA\ket{x,0^{k+m+1}}}^2\right] 
\geq \frac{\gamma^3}{2},
\end{align*}
where in the first inequality we removed some non-negative values, in the second inequality we use \Cref{eq:size-g} and in the third inequality we use \Cref{eq:implication-g}.

Let us now show \Cref{eq:size-g}. Suppose toward a contradiction that 
 $|G| < \frac{\gamma}{2} \cdot 2^n$. Then we have
  \begin{align*}
\frac{1}{2} + \gamma &\leq \E_{x\in \01^{kn}}\E_{r\in\01^k}\left[\norm{\Pi_{f(x) \cdot r}\calA\ket{x,r,0^m}}^2\right]  \\
&= \frac{|G|}{2^n} \E_{x\in G}\E_{r\in\01^k}\left[\norm{\Pi_{f(x) \cdot r}\calA\ket{x,r,0^m}}^2\right] + 
\left(1-\frac{|G|}{2^n}\right)\E_{x\not\in G}\E_{r\in\01^k}\left[\norm{\Pi_{f(x) \cdot r}\calA\ket{x,r,0^m}}^2\right] \\
&< \frac{|G|}{2^n} + \left(1-\frac{|G|}{2^n}\right)\left(\frac{1}{2} + \frac{\gamma}{2}\right)\\
&< \frac{1}{2} + \gamma/2,
\end{align*}
which is a contradiction. Therefore, \Cref{eq:size-g} must be true.
 
Now we prove \Cref{eq:implication-g} and for that let us fix an arbitrary $x \in G$.
This part of the proof closely follows the proof by Cleve and Adcock~\cite{adcock2002quantum}. Without loss of generality, we can assume that $\calA$ on the input state acts as
$$
\mathcal{A}\ket{x,r,0^m}=
\ket{x,r} \otimes\left(\alpha_{x,r,0}\ket{\psi_{x,r,0}}\ket{f(x)\cdot r} + \alpha_{x,r,1}\ket{\psi_{x,r,1}}\ket{1\oplus f(x)\cdot r}\right).
$$
From the assumption that $x \in G$, we have that 
\[
\E_{r}\big[|\alpha_{x,r,0}|^2\big] \geq \frac{1}{2} + \frac{\gamma}{2} \quad \text{ and } \quad 
\E_{r}\big[|\alpha_{x,r,1}|^2\big] \leq \frac{1}{2} - \frac{\gamma}{2}.
\]
We define two quantum states. First, we start with $\frac{1}{\sqrt{2^k}}\sum_{r}\ket{x,r,0^m,1}$ and apply steps $2$ and $3$ of~$\calB^\calA$ to obtain
\begin{align*}
    \ket{\phi_x} = 
   \ket{x} \otimes  \frac{1}{\sqrt{2^k}}\sum_{r}
(-1)^{f(x)\cdot r}\big(\alpha_{x,r,0}\ket{r}\ket{\psi_{x,r,0}}\ket{f(x)\cdot r,1} - \alpha_{x,r,1}\ket{r}\ket{\psi_{x,r,1}}\ket{1\oplus f(x)\cdot r,1}\big).
\end{align*}
In order to analyze the probability of the measurement outcome $a$ (the output of $\calB^\calA$ in step (7)) equalling $f(x)$, consider the state where we start with  $\ket{x,f(x),0^{k+m},1}$, apply the Hadamard transform on the second register and then apply $\calA$ on the first $kn+k+m$ qubits to obtain
\begin{align}
    \ket{\sigma_x} = 
    \frac{1}{\sqrt{2^k}} \ket{x}\sum_{r}
(-1)^{f(x)\cdot r}\big(\alpha_{x,r,0}\ket{r}\ket{\psi_{x,r,0}}\ket{f(x)\cdot r,1} + \alpha_{x,r,1}\ket{r}\ket{\psi_{x,r,1}}\ket{1\oplus f(x)\cdot r,1}\big).
\end{align}
The probability that $\calB^{\calA}(x)$ outputs $f(x)$ is then given by
\begin{align}
    |\braket{\sigma_x \mid \phi_x}|^2 = \left|\E_{r}\left[|\alpha_{x,r,0}|^2 - |\alpha_{x,r,1}|^2\right]\right|^2 \geq \gamma^2,
\end{align}
where we use the fact that for $x \in G$, we have $\E_{r}\left[|\alpha_{x,r,0}|^2 - |\alpha_{x,r,1}|^2\right] \geq \gamma$. This concludes the proof of \Cref{eq:implication-g}.
Notice that the size of $\calB^\calA$ can be simply checked by inspection from its~description.

We discuss now the moreover part of our lemma. The circuit $\mathcal{C}_n^{\mathsf{GL}}$ receives the input $\mathsf{code}(\calA)$ (and the auxiliary parameter $1^n$) and outputs the circuit $\calB$ that executes the steps 
$\eqref{item:first_step}-\eqref{item:last_step}$ described above, with an oracle call to $\calA$ replaced by the execution of $\textsf{code}(\calA)$ (or its inverse).
It is straightforward from the previous calculations that $\calB$ has the desired properties and that it has size at most $O(k(n) \cdot n\cdot s_{\calA}(n))$. Notice that the circuit $\mathcal{C}_n^{\mathsf{GL}}$  that prints $\calB$ is deterministic and it can be implemented using $\mathsf{poly}(n, k(n), s_A(n))$ gates. We also notice that code of $\mathcal{C}_n^{\mathsf{GL}}$ is explicit and it can be deterministically generated from $1^n$ in time $\mathsf{poly}(n, k(n), s_{\calA}(n))$. 
\end{proof}

\subsection{Local list decoding and uniform hardness amplification for quantum circuits}\label{sec:IJKW}

\newcommand{\algoIJKW}{\calG}
\newcommand{\univ}{\mathfrak{U}}

In this section, we start with some arbitrary function $g : \01^n \to \01$, which we assume to be mildly hard, and our goal is to amplify its hardness using techniques from local list decoding of (classical) error-correcting codes. Specifically, we prove a quantum analogue of the direct product theorem of Impagliazzo, Jaiswal, Kabanets, and Wigderson \cite{impagliazzo2010uniform}. For simplicity of notation, we will fix $n \in \mathbb{N}$ and denote $\univ = \01^n$. Roughly speaking, we will prove that computation of the concatenation of $k$ independent copies of $g$ amplifies its hardness exponentially as a function of $k$. For that, we first define the domain of such concatenation.

\begin{definition}[$k$-sets]\label{def:ksets}
Let $\mathcal{S}_{n,k} = \{S \subseteq \univ : |S| = k\}$ be the set of all subsets of $\01^n$ of size $k$. When $n$ is fixed, we denote $\mathcal{S}_{k} = \mathcal{S}_{n,k}$.
\end{definition}

\begin{remark}\label{r:remark_input}
We consider the $n$-bit strings in a set $B \in \mathcal{S}_{n,k}$ in \emph{lexicographic order} when $B$ is given as input to a Boolean circuit.
\end{remark}

We consider then the $k$-\emph{direct product} of $g$, denoted $g^k : \mathcal{S}_k \to \{0,1\}^k$, where $g^k(B)$ is the concatenation of $g(x)$ for every $x \in B$ in a canonical order, i.e,
$$
g^k(x_1, \ldots, x_k) = \big(g(x_1), \ldots, g(x_k)\big).
$$ 
Whenever it is clear from the context, we abuse the notation with $g(B) = g^k(B)$, for $B \in \calS_k$. There are well known connections between hardness amplification and classical error-correcting (for example see~\cite [Chapter 19]{book_complexity}) and part of our terminology reflects them. For instance, the code can be viewed as corresponding to $g^k$ and decoding will correspond to computing $g(y)$ for some input $y \in \univ$ when given appropriate access to $g^k$.

As a preliminary step, we consider hardness amplification for inherently probabilistic computations, which we introduced in Section \ref{ss:inherently_probabilistic} as an intermediate model between classical and quantum circuits. 

\subsubsection{Inherently probabilistic circuits}

Recall that an \emph{inherently probabilistic circuit} $\algoIJKW$ for computing a function $g \colon \{0,1\}^m \to \{0,1\}^\ell$ with probability at least $\eps$ is a circuit that assigns to each input $z \in \{0,1\}^m$ a distribution $\algoIJKW(z)$ supported over $\{0,1\}^\ell$ such that
\begin{equation*}
    \Pr_{\substack{z \sim \{0,1\}^m,\\ v \sim \algoIJKW(z)}}[v = g(z)] \geq \varepsilon \:.
\end{equation*}

Our goal in this section is to show that an inherently probabilistic circuit $\algoIJKW$ that computes $g^k$ with small probability $\varepsilon>0$ can be used as a subroutine in a randomized circuit with oracle access to $\algoIJKW$ to compute $g$ with high probability $1 - \delta$. Moreover, we aim to design a uniform randomized algorithm that is able to generate, with non-trivial probability $\zeta$ as a function of $\varepsilon$, a description of this oracle circuit from a description of $\algoIJKW$.

We state the main theorem of this section. In the following, it might be instructive to think of the following setting of parameters as a function of the input size: $\eps = 2^{-\sqrt{n}}$, $\delta = 1/\poly(n)$, and $k = \poly(n)$.

\begin{theorem}[Local list decoding for inherently probabilistic circuits]\label{thm:ijkw}
There exists a universal constant $C \geq 1$ for which the following holds. Let $n \geq 1$ be a positive integer,  $k$ be an even integer, and let $\varepsilon, \delta > 0$ satisfy 
\begin{equation}\label{eq:IJKW_relation_eps_delta}
k \geq C \cdot \frac{1}{\delta} \cdot \left [ \log \! \left( \frac{1}{\delta} \right )  + \log \! \left  ( \frac{1}{\varepsilon} \right ) \right ].
\end{equation}
If $\algoIJKW$ is an inherently probabilistic circuit defined over $\mathcal{S}_{n,k}$ with $k$ output bits such that
\begin{align}\label{eq:assumption-ijkw}
\E_{B \sim \calS_{n,k},\; \algoIJKW}\left[\algoIJKW(B) = g^k(B)\right] 
 \geq \eps,
\end{align}
then there is a randomized oracle circuit $\mathcal{B}^\mathcal{O}$ of size $\mathsf{poly}(n,k,\log(1/\delta), 1/\varepsilon)$ such that
\begin{equation}\label{eq:IJKW_conclusion}
\E_{\substack{x \sim \univ,\\y \sim \{0,1\}^*,\;\algoIJKW}}\Big [\mathcal{B}^\algoIJKW(x,y) = g(x)\Big ] \geq 1 - \delta.
\end{equation}
Moreover, there is a uniform randomized algorithm $\mathcal{D}$ that, given $n$, $k$, $\varepsilon$, and $\delta$ satisfying the conditions of the theorem, access to a description of $\algoIJKW$, and the ability to run $\algoIJKW$ on a given input $B \in \mathcal{S}_{n,k}$, computes in time $\mathsf{poly}(n,k,\log(1/\delta), 1/\varepsilon)$  and outputs with probability $\zeta = \Omega(\varepsilon^2)$ over its internal randomness and the randomness of $\algoIJKW$ a description of a circuit $\mathcal{B}^{\mathcal{O}}$ with  this property.
\end{theorem}

\begin{remark}[Amplification] We  observe that the generating probability $\zeta$ can be amplified using standard techniques (repetition and hypothesis testing) if the uniform randomized algorithm $\mathcal{D}$ is also given oracle access to  the function $g$.
\end{remark}
 
 \begin{remark}[On the correlation between $\algoIJKW$ and $g^k$] Before we proceed with the proof of Theorem \ref{thm:ijkw}, it is worth pointing out different ways in which the inherently probabilistic circuit $\algoIJKW$ can be correlated with $g^k$. First, $\algoIJKW$ might behave as a \emph{deterministic} circuit, being correct on an arbitrary set $V \subseteq \mathcal{S}_k$ of relative size $\varepsilon$ (i.e., on each $B \in V$ we have $\algoIJKW(B) = g^k(B)$ with probability $1$), and being incorrect elsewhere. At the other extreme, it is also possible for $\algoIJKW$ to agree with $g^k(B)$ on each $B \in \mathcal{S}_k$ with probability about $\varepsilon$ over its internal randomness, spreading out its correlation with $g^k$ across all inputs. Since $\algoIJKW$ is inherently probabilistic, there is no simple way of reducing this case to the preceding case (e.g.,~by fixing $\algoIJKW$'s internal randomness). Finally, it is possible that $\algoIJKW$'s behavior combines in an arbitrary way the two aforementioned cases, while maintaining an overall advantage $\varepsilon$ with respect to $g^k$. A proof of Theorem \ref{thm:ijkw} needs to address all possible scenarios. 
 \end{remark}
 
 To present the local decoding algorithm with which we will prove \cref{thm:ijkw}, it will be convenient to refer to the bipartite graph (the incidence graph of the Johnson scheme) that contains all $k/2$-sets on the left and all $k$-sets on the right, where the edges correspond to incidence of the sets.

\begin{definition}[Edges and Neighbors]
We define the set of \emph{edges} $\Is = \{(A,B) \in 
\mathcal{S}_{k/2}  \times \mathcal{S}_{k}: A \subseteq~B\}$. We define the \emph{neighbors} of a set $A$ by $N_\Is(A) = \{B \in \mathcal{S}_k : (A,B) \in \Is\}$ and the neighbors of $A \in \mathcal{S}_{k/2}$ and $x \in \univ \setminus A$ as $N_\Is(A,x) = \{B \in \mathcal{S}_k : A \cup \{x\} \subseteq B\}$. 
\end{definition}

We proceed to present the key sub-procedure for the list-decoding algorithm, which is a circuit that, given a $k/2$-set $A$ and an assignment $w$ to $A$, attempts to compute $g(x)$ for a given $x$  by choosing a random neighbour $B'$ of $A$ and $x$, running $\algoIJKW$ on $B'$ and outputting the corresponding value for $x$ if the output of $\algoIJKW(B')$ is consistent with $w$ on $A$.

\begin{construction}[The decoding circuit $C_{A,w}$]
\label{const:CAw}
For a fixed $A \in \mathcal{S}_{k/2}$, $w \in \01^{k/2}$, and a parameter $T$, we define the randomized circuit $C_{A,w}$ with oracle access to $\algoIJKW$ that, on input $x \in \univ$, works as follows:
\begin{enumerate}
    \item If $x \in A$, output $w|_x$.\footnote{By $w|_x$, we mean the following: suppose $g(y_1,\ldots,y_{k/2})=(w_1,\ldots,w_{k/2})$ for $y_i\in \01^n$ and $w_i\in \01$, and $y_\ell=x$ for some $\ell\in [k/2]$, then $w|_x=w_\ell$.}
    \item Repeat for $T$ steps:
    \begin{enumerate}[label*=\arabic*.]
        \item Pick a uniformly random $B' \in N_\Is(A,x)$.
        \item Sample $v' \sim \algoIJKW(B')$.
        \item If $v'|_{A} = w$, output $v'|_x$.
    \end{enumerate}
    \item Output $\bot$.
\end{enumerate}
\end{construction}

Using \cref{const:CAw}, the list-decoding algorithm will decode by evaluating the input circuit~$\algoIJKW$ on a random $k$-set $B$ and outputting the circuit $C_{A,w}$ with respect to a random $k/2$-set $A \subset B$.

\begin{construction}[The list-decoding algorithm $\mathcal{D}$]
\label{const:decoder}
The uniform randomized algorithm $\mathcal{D}$ is given $n$, $k$, $\varepsilon$, and $\delta$ satisfying the conditions of \cref{thm:ijkw}, access to a description of $\algoIJKW$, and a parameter $T$. The algorithm $\mathcal{D}$ operates as follows:
\begin{enumerate}
    \item Pick a random $k$-set $B \in \mathcal{S}_{n,k}$ and a random $k/2$-subset $A \subseteq B$. 
    \item Sample $\algoIJKW(B)$ and use it to set $w = \algoIJKW(B)|_A$.
    \item Output the description of the circuit $C_{A,w}$ (defined in \cref{const:CAw}) with respect to the given parameter $T$.
\end{enumerate}
\end{construction}

Similar to~\cite{impagliazzo2010uniform}, we prove \cref{thm:ijkw} in two steps. First,  we show that if $A$ and $w$, which were randomly chosen by the list-decoder, have certain desired properties and the repetition parameter $T$ is sufficiently large, the circuit $C_{A,w}$ computes $g$ with high probability over a random input $x \in \univ$, its internal randomness, and the inherent randomness of $\algoIJKW$. (Therefore, the oracle circuit $\mathcal{B}^\mathcal{O}$ in the statement of the lemma will be $C_{A,w}$ for a good choice of $A$ and $w$, with its oracle $\mathcal{O}$ computing as $\algoIJKW$, and the input string $y$ used as a source of random bits.) Then, we argue that with probability at least $\zeta = \Omega(\varepsilon^2)$  the uniform randomized algorithm $\mathcal{D}$ is able to generate good parameters $A$ and $w$. 

Throughout this section, we define $\correctCircuitA = C_{A,g^{k/2}(A)}$, for simplicity.

\begin{remark}[Well-defined conditional probabilities]
\label{rmk:well-defined-cond}
In order to simplify our exposition, we will assume without loss of generality throughout this section that for every $B \in \mathcal{S}_k$, we have $\Pr_{\algoIJKW}[\algoIJKW(B) = g^k(B)] > 0$. Indeed, this can be easily obtained by defining a modified circuit $\algoIJKW'$ from $\algoIJKW$ that, say, with probability $2^{-n}$ outputs a uniformly random value in $\{0,1\}^k$, and otherwise runs $\algoIJKW$ on $B$. We have this assumption so that definitions and proofs become simpler when considering certain conditional probabilities. We stress that our argument would still work without this trick, with a slightly more complicated exposition.
\end{remark} 

We will need a few definitions to prove the theorem. Observe that some of them implicitly refer to $\algoIJKW$ and $g$, which are fixed for the remainder of this section. We start with the definition of correctness of the algorithm $\algoIJKW$ on a $k$-set. 
\begin{definition}
For $B \in \mathcal{S}_{k}$, we define 
$$
\mathsf{Corr}_\algoIJKW(B) = \E_{\algoIJKW}[\algoIJKW(B) = g^k(B)].
$$ 
We say that $B $ is $\eta$-correct if 
$\mathsf{Corr}_\algoIJKW(B) \geq \eta$.
\end{definition}
Notice that, using this notation, we can rewrite the assumption of \Cref{eq:assumption-ijkw} in \cref{thm:ijkw}~as
\[
\E_{B\in  \mathcal{S}_k}\left[ \mathsf{Corr}_\algoIJKW(B) \right] \geq \eps.
\]
Since $\algoIJKW$ is fixed throughout this section, we might write $\mathsf{Corr}(B)$ instead of $\mathsf{Corr}_\algoIJKW(B)$.

Next, we define good edges as (mostly) correct sets for which many of their neighbours are also (mostly) correct.

\begin{definition}[Good edges]
We say an edge $(A,B) \in \Is$ is $(\gamma, \eta)$-{\em good} if 
\begin{itemize}
    \item[\emph{(\emph{i})}]  $B$ is $\eta$-correct; and
    \item[\emph{(\emph{ii})}] A $\gamma$ fraction of the neighbors of $A$ are $\eta$-correct, i.e.,
$$
\E_{B' \in N_\Is(A)}[\mathsf{Corr}(B') \geq \eta]\geq \gamma.
$$
\end{itemize}
\end{definition}

We remark that with the above definition we already start deviating from \cite{impagliazzo2010uniform}. Since their version of $\algoIJKW$ is deterministic, its answers are either correct or wrong and therefore they can afford to have $\eta = 1$. In our case, we need to take into account the \emph{randomized} aspect of $\algoIJKW$ and therefore we need to be more flexible with the definition of goodness.  

\newcommand{\Corr}{\ensuremath{\mathsf{Corr}}}
\newcommand{\ErrCons}{\ensuremath{\mathsf{ErrCons}}}

As in \cite{impagliazzo2010uniform}, in order to prove the correctness of $C_{A,w}$, we need edges that satisfy a stronger property than the above, referred to as ``excellence''. For that, we first define the function that computes the probability that an edge $(A,B')$ leads to a wrong answer on a random $x \in B'\setminus A$ conditioned on a correct answer on $A$.
\begin{definition}
For an edge $(A,B') \in \Is$, we define 
\begin{align}
\ErrCons(A,B') = \E_{\substack{x \in \NB}}
\big[p_{\mathsf{err}}(x,B') \big],
\end{align}
where $p_{\mathsf{err}}(x,B') = \Pr_{y \sim \algoIJKW(B')}[y|_x \neq g(x) \mid y|_A = g^{k/2}(A)]$.
\end{definition}

Using the foregoing definition, we define an excellent edge $(A,B)$ as a good edge for which the expected fraction of errors in the neighbors of $A$ is small, where this expectation gives weight to the edges based on their probability of being correct on $A$.

\begin{definition}[Excellent edges]
\label{def:excellent}
We call an edge $(A,B)$ $(\eta , \gamma, \alpha)$-{\em excellent} if it is $(\eta, \gamma)$-good~and
\begin{align}\label{eq:def-excellent}
\E_{\substack{B' \sim \mathcal{W}_\Is(A)}} \left[ \ErrCons(A,B')
    \right] \leq \alpha,
\end{align}
where the distribution $\mathcal{W}_\Is(A)$, supported over $N_\Is(A)$, is defined as follows: for each $B' \in N_\Is(A)$, let $p_\mathsf{cons}(B') = \Pr_{\algoIJKW}\big[\algoIJKW(B')|_A = g^{k/2}(A)\big]$. Moreover, let $p_{\mathsf{tot}}(A) = \sum_{B' \in N_\Is(A)} p_{\mathsf{cons}}(B')$. Then each~$B'$ is assigned probability $p_\mathsf{cons}(B')/p_{\mathsf{tot}}(A)$.
\end{definition}

For the reader familiar with \cite{impagliazzo2010uniform}, we mention that the definition above is really crucial. It refers to a more general class of distributions $\mathcal{W}_\mathcal{I}(A)$ when contrasted with the analogous definition from their paper. (This distribution naturally appears in our error analysis when we condition on $C_{A,w}$ not outputting the error symbol $\perp$.) Jumping ahead, it ties the two main lemmas stated below and proved in the subsequent sections, and introduces significant difficulties when establishing each of them.

Given these definitions, we can now state the main two lemmas needed for the proof of \Cref{thm:ijkw}. The first lemma shows that if the decoding algorithm hits an excellent edge, then it will decode correctly with high probability. The second lemma shows that there are many excellent edges, and so the decoding algorithm will hit one with non-trivial probability (this is close to 0). 

\begin{lemma}[Excellence implies correctness]\label{lem:excellent-implies-correctness}
Fix some $0 \leq \beta \leq 1$ and let $\lambda = 2 e^{-\beta k/24}$. Moreover, assume that
$$
\gamma, \eta < 1/10 \quad \text{and} \quad 4 e^{-k\alpha/12} \;\leq\; (\gamma \cdot \eta)^5.
$$
If $(A,B)$ is a $(\gamma,\eta,\alpha)$-excellent edge, then 
\[\E_{\substack{x \sim \univ \setminus A,\\ \correctCircuitA ,\;\algoIJKW}}[\correctCircuitA(x) = g(x)] \geq  1 - \beta - (1-\eta(\gamma-\lambda)/2)^T - 16 \alpha, \]
where
$T$ is the number of iterations in $\correctCircuitA$.
\end{lemma}

\begin{lemma}[Excellent edges are abundant]\label{lem:excellent-is-abundant}
For any $\alpha < \frac{1}{2}$,
if $\E_{B\in  \mathcal{S}_k}\left[ \mathsf{Corr}(B) \right] \geq \eps$ then at least an $(\varepsilon/3 - \frac{62208}{\alpha^3 \cdot \varepsilon^5} \cdot e^{-\frac{\alpha}{96}\cdot k})$-fraction of the edges $(A,B) \in \Is$ are $(\eps/3,\eps/3,\alpha)$-excellent.
\end{lemma}

We prove \Cref{lem:excellent-implies-correctness,lem:excellent-is-abundant} in \Cref{sec:excellent-implies-correctness,sec:excellent-is-abundant}, respectively. Assuming these results, we are ready to prove  \Cref{thm:ijkw}.

\begin{proof}[Proof of \Cref{thm:ijkw}]
We begin with the first part of the result; namely, showing that the algorithm $C_{A,w}$ defined in \cref{const:CAw}, parameterized as we specify below, is of size $\mathsf{poly}(n,k,\log(1/\delta), 1/\varepsilon)$ and satisfies
\begin{equation*}
\E_{\substack{x \sim \univ,\\y \sim \{0,1\}^*,\;\algoIJKW}}\Big [C_{A,w}^\algoIJKW(x,y) = g(x)\Big ] \geq 1 - \delta.
\end{equation*}
First, note that we can assume that $\varepsilon < 1/10$ without loss of generality. For this it is sufficient to redefine $\algoIJKW$ to output a random value with probability $1 - 10^{-3}$, and to compute as before otherwise. Its overall success probability drops at most by a constant factor, and now $\varepsilon < 1/100$. We also assume that $\delta \leq 1/2$,  by taking a smaller $\delta$ if necessary. 

Let $C$ be a large positive constant independent of the remaining parameters. Let
$$
\gamma, \eta \eqdef \varepsilon/3 < 1/10 \quad \text{and} \quad \beta \eqdef \delta/3 \quad \text{and} \quad \alpha \eqdef \delta/48 \quad \text{and} \quad T = C \cdot \log(1/\delta) \cdot (1/\varepsilon^2),
$$ 
where $T$ is the repetition parameter in \cref{const:decoder}. Note that
$$
k \;\geq\; C \cdot \frac{1}{\delta} \cdot \left [ \log \! \left( \frac{1}{\delta} \right )  + \log \! \left  ( \frac{1}{\varepsilon} \right ) \right ],
$$
and so $4 e^{-k\alpha/12} \;\leq\; (\gamma \cdot \eta)^5$, for a sufficiently large choice of $C$.

By applying \cref{lem:excellent-is-abundant} with respect to these $\alpha$, $\eps$, and $k$ we have that
 at least an $(\varepsilon/3 - \frac{62208}{\alpha^3 \cdot \varepsilon^5} \cdot e^{-\frac{\alpha}{96}\cdot k})$-fraction of the edges $(A,B) \in \Is$ are $(\eps/3,\eps/3,\alpha)$-excellent. Let $\lambda = 2 e^{-\beta k/24} < \varepsilon/6$. Note that our choice of parameters satisfies all the conditions of \cref{lem:excellent-implies-correctness}, and thus by invoking it we obtain that if $(A,B)$ is a $(\gamma,\eta,\alpha)$-excellent edge, then 
\[\E_{\substack{x \sim \univ \setminus A,\\ \correctCircuitA,\;\algoIJKW}}[\correctCircuitA = g(x)] \geq  1 - \beta - (1-\eta(\gamma-\lambda)/2)^T - 16 \alpha, \]
where $\correctCircuitA = C_{A,g^{k/2}}(x)$ as previously defined.

Now, observe that
$$
(1 - \eta (\gamma - \lambda)/2)^T \leq \left (1 - \frac{\varepsilon}{3} \left ( \frac{\varepsilon}{3} - \frac{\varepsilon}{6} \right ) \cdot \frac{1}{2} \right )^T = \left ( 1 - \frac{\varepsilon^2}{36} \right )^T \leq e^{-(\varepsilon^2/36) \cdot T} \leq \frac{\delta}{3},
$$
where we have used that $\varepsilon < 1/100$ and that the constant $C$ appearing in $T$ is large. As a consequence, if the randomized algorithm $\mathcal{D}$ succeeds in sampling a pair $(A,B)$ that is $(\varepsilon/3, \varepsilon/3, \delta/48)$-excellent and it also samples a value $w = \algoIJKW(B)|_A$ which agrees with $g^{k/2}(A)$, we get the randomised circuit $\correctCircuitA$ with oracle access to $\algoIJKW$ such that  %
\[
\E_{\substack{x \sim \univ \setminus A,\\\correctCircuitA ,\;\algoIJKW}}[\correctCircuitA(x) = g(x)] \;\geq\;  1  - \delta/3 - \delta/3 - \delta/3 = 1 - \delta. \]
It is not hard to see that the inequality above also holds when $x$ is sampled from the (marginally) larger set $\univ$, since $\correctCircuitA$ is always correct for $x \in A$. 

On the other hand, note that, due to our choice of parameters in Lemma \ref{lem:excellent-is-abundant} we have
$$
\frac{62208}{\alpha^3 \cdot \varepsilon^5} \cdot e^{-\frac{\alpha}{96}\cdot k} \;\leq\; \varepsilon/6.
$$
It follows from Lemma \ref{lem:excellent-is-abundant} that at least an $(\varepsilon/6)$-fraction of the pairs $(A,B) \in \Is$ are $(\varepsilon/3, \varepsilon/3, \delta/48)$-excellent. In particular, excellent edges of this form exist.

To sum up, the first part of the lemma follows by picking $\calB^{\mathcal{O}} = \correctCircuitA$ with $\mathcal{O} = \algoIJKW$ for any choice of $(A,B)$ as above. Note that, by our choice of $T$, the number of gates in the circuit $\mathcal{B}^{\mathcal{O}}$ satisfies the parameters of the lemma.  

We now prove that the uniform randomized algorithm $\mathcal{D}$ outputs a circuit with the desired properties with probability at least $\Omega(\eps^2)$ over its internal randomness and the randomness of $\algoIJKW$. It follows from our discussion in the paragraph above that, in order for the circuit output by $\mathcal{D}$ to have the desired accuracy, we need two properties from the  values produced by $\mathcal{D}$:
\begin{itemize}
    \item[(\emph{i})] the randomly selected edge $(A,B)$ is $(\eps/3,\eps/3,\delta/48)$-excellent; and
    \item[(\emph{ii})] the sampled value $w = \algoIJKW(B)|_A$ coincides with $g^{k/2}(A)$.
\end{itemize}
As we have already established, an $\Omega(\eps)$ fraction of edges $(A,B)$ are $(\eps/3,\eps/3,\delta/48)$-excellent, so by picking a random $k$-set $B$ and a random $k/2$-subset $A$ of $B$, which produces a uniformly random edge $(A,B)$, we have that $(A,B)$ is an excellent edge with probability $\Omega(\eps)$. Then, assuming that $(A,B)$ is $(\eps/3,\eps/3,\delta/48)$-excellent, it follows that $\algoIJKW(B) = g^k(B)$ with probability at least $\eps/3$, which implies that the probability that $\algoIJKW(B)|_A = g^{k/2}(A)$ is also at least $\eps/3$. Therefore, the probability that $\calD$ outputs a circuit with the desired properties is at least $\Omega(\eps^2)$.
\end{proof}

\subsubsection{Excellence implies correctness}
\label{sec:excellent-implies-correctness}

 In this section, we prove \Cref{lem:excellent-implies-correctness}, which shows that if an excellent edge $(A,B)$ is picked, then $\tilde{C}_A$ computes $g(x)$ with high probability, on average over the $x$'s. Since this proof is fairly technical, we first give an overview of its structure, before diving into the details.  Recall that the goal is to upper bound the quantity
\begin{align}
\label{eq:sketch0}    
\E_{\substack{x \in \univ \setminus A\\\correctCircuitA,\;\algoIJKW}}[\correctCircuitA(x) \ne g(x)].
\end{align}
It is easy to see that the event $\correctCircuitA(x) \ne g(x)$ occurs if either (1) the decoder outputs $\bot$ and aborts or (2)  the decoder does not output $g(x)$ conditioned on not outputting $\bot$.

We first show that the probability of (1) happening is small. More precisely, we show that given a fixed excellent edge $(A,B)$ and a sufficiently large number of iterations $T$, for most $x \in \univ \setminus A$ the probability that $\correctCircuitA$ outputs $\bot$ on $x$ is negligibly small. 

Then, we follow up showing that (2) happens with low probability, which turns out to be much more cumbersome than \cite{impagliazzo2010uniform}.
Here, we need to upper bound 
\begin{align}
    \label{eq:sketch2}
\mu \eqdef \E_{x \in \univ \setminus A }\hspace{1mm}\underbrace{\E_{y \sim \correctCircuitA(x),\;\algoIJKW} [y \ne g(x) \mid y \ne \bot]}_{:=h(x)}.
\end{align}

If we dive into the definition of $\tilde{C}_A(x)$, \Cref{eq:sketch2} computes, for a uniformly random $x \in \univ \setminus A$, the probability that at some iteration $\tilde{C}_A$ picks some $B' \sim N_{\Is}(A,x)$ and then $\algoIJKW(B')$ outputs some value $y$ that is incorrect on $x$ given that it is  correct on $A$. Alternatively (as a thought experiment), we could consider a different  sampling procedure that picks $B' \sim N_{\Is}(A)$ and then picks a uniformly random $x \in B' \setminus A$ and analyze the probability of being incorrect to $x$ conditioned on the fact of being correct to $A$.  
This latter procedure would give an expression similar to that for determining the excellence of the edge $(A, B)$ and could then be related to $\alpha$ as required. In fact, this is the approach taken by~\cite{impagliazzo2010uniform} where the sampling behaviour satisfied by $A, B$ and~$x$ (in their case), allows them to switch between both scenarios in a fairly direct way. This already becomes more complicated in our case. 

 For one thing, in our definition of excellence (see~\Cref{def:excellent}), we sample $B' \sim \mathcal{W}_\Is(A)$ and not uniformly at random. However, we succeed in indirectly relating the $\alpha$-excellence of $(A, B)$ to Eq.~\eqref{eq:sketch2} in two steps: (i) show that $\alpha$ upper bounds the expression 
 \begin{align}
 \label{eq:sketch3}
     \E_{B' \sim \mathcal{W}_I(A)}  \E_{x \sim \NB} [h(x)];
 \end{align}
 (ii) lower bound Eq.~\eqref{eq:sketch3} in terms of $\mu$. %
 In order to relate these two expressions, we still need to contend with the fact that $B' \in N_\Is(A)$ is sampled with probability proportional to $p_{\mathsf{cons}}(B)$ in Eq.~\eqref{eq:sketch3}. To this end, we first partition the set $N_\Is(A)$ as follows: let $\Gamma_i$ be the set of all $B'\in N_\Is(A)$ such that $p_{\textsf{cons}}(B')$ lies in the interval $(2^{-i},2^{-i+1}]$. Observe that in the scenario of~\cite{impagliazzo2010uniform}, if we were to do a similar partitioning, there would only be \emph{two} sets, $B'$s that are consistent with $A$ and the $B'$s that aren't consistent. For us, since each $B$ is associated with a different weight, each $B$ has a varying ``degree of consistency". We now decompose  Eq.~\eqref{eq:sketch3} into all the distinct buckets and write it as
 \begin{align}
 \label{eq:sketch4}
  \sum_{i} \E_{B' \sim \mathcal{W}_I(A)} \E_{x \sim \NB} \big [h(x) \mid B' \in \Gamma_i \big] \cdot \Pr_{B' \sim \mathcal{W}_I(A)} \big [B' \in  \Gamma_i \big].
 \end{align}
 At this point, we show with some straightforward calculations that most of the contribution in Eq.~\eqref{eq:sketch4} comes from the buckets $\Gamma_i$ that are heavy (i.e., contain many $B'\in N_\Is(A)$ inside them) and for such heavy buckets, we show that  %
 \begin{align}
     \label{eq:sketch5}
     \E_{x \sim \NB} \big [h(x) \mid B' \in \Gamma_i \big] \cdot \Pr_{B' \sim \mathcal{W}_I(A)} \big [B' \in  \Gamma_i \big]\geq \sigma \cdot \E_{\substack{x \in \univ \setminus A\\\correctCircuitA,\;\algoIJKW}}[\correctCircuitA(x) \ne g(x)],
 \end{align}
 for some universal constant $\sigma <1$. At this point, observe that the RHS of Eq.~\eqref{eq:sketch5} is exactly the quantity we wanted to upper bound, i.e., Eq.~\eqref{eq:sketch0} and we have shown that the LHS is at most $\alpha$. This shows that $\mu \leq O(\alpha)$ and concludes the proof of the theorem. We now make this sketch more rigorous below.

We start with the first point of showing that the probability of outputting $\bot$ is small. In order to show this, we will need to use that if $(A,B)$ is a $(\gamma,\eta)$-good edge (which is implied by $(\gamma,\eta,\alpha)$-excellence), we can bound the number of $x \in \univ \setminus A$ for which less than about a $\gamma/2$-fraction of the sets $B' \in N_\Is(A,x)$ have $\E\big[\algoIJKW(B')\big|_A = g^{k/2}(A)\big] \geq \eta$ (these are the $x$'s on which $\correctCircuitA$ is likely to output $\bot$). This claim is formalized as follows.

\begin{lemma}[Analogue of {\cite[Lemma~3.9]{impagliazzo2010uniform}}]
\label{lem:bound-algorithm-answers}
Let $\beta\in [0,1]$ and $\lambda = 2 e^{-\beta k/24}$. Fix a $(\gamma , \eta)$-good edge $(A,B) \in \Is$. Let  $F\subseteq \univ \setminus A$ be the set of elements $x$ such that strictly less than a $(\gamma - \lambda)/2$ fraction of $B' \in N_\Is(A,x)$~satisfy
$
\Pr_{\algoIJKW} \! \big[\algoIJKW(B')\big|_A = g^{k/2}(A)\big] \geq \eta
$.
Then the density $\rho(F) \eqdef |F|/| \univ \setminus A| < \beta$. %
\end{lemma}
\begin{proof}
Assume toward a contradiction that $\rho(F) = \beta$ (the case where $\rho(F) > \beta$ follows by fixing some $F' \subseteq F$ with $\rho(F') = \beta$).  Define the set
$$
W = \big\{B' \in N_{I}(A) : \E \! \big[\algoIJKW(B')\big|_A = g^{k/2}(A)\big] \geq \eta \big\}.
$$
Observe that 
\begin{align}
\label{eq:excellenteq1}
&\Pr_{\substack{x \in \univ \setminus A \\ B' \supseteq \{x\} \cup A}} 
\left[x \in F \wedge B' \in W \right]
= \Pr_{x \in \univ \setminus A} 
\left[x \in F \right]\cdot \Pr_{\substack{x \in \univ \setminus A \\ B' \supseteq \{x\} \cup A}} 
\left[B' \in W \mid  x \in F \right]
< \beta \cdot (\gamma - \lambda)/2.
\end{align}
On the other hand, let
\begin{equation*}
    W' = \{B' \in W : \frac{|(B' \setminus A) \cap F|}{|B' \setminus A|} \geq \beta/2 \},
\end{equation*}
and observe that since $W'\subseteq W$, we have
\begin{align}
    \label{eq:excellenteq2}
\Pr_{\substack{x \in \univ \setminus A \\ B' \supseteq \{x\} \cup A}} 
\left[x \in F \wedge B' \in W \right]
&\geq \Pr_{\substack{x \in \univ \setminus A \\ B' \supseteq \{x\} \cup A}} \left[x \in F \wedge B' \in W' \right] \notag\\
&= \Pr_{B'\in N_{I}(A)}[B' \in W'] \cdot \Pr_{\substack{B' \in N_\Is(A) \\ x \in B' \setminus A}} \left[x \in F  \mid  B' \in W' \right],
\end{align}
where the equality uses the following simple fact: selecting $x$ uniformly from $\univ \setminus A$  then picking a uniform $B'\in N_\Is(A)$ that contains $\{x\}\cup A$ is equivalent to  selecting $B'$ uniformly from $N_\Is(A)$ then picking a uniform $x\in B' \setminus A$. To complete the proof, we argue next that the last expression is larger than the bound in Eq.~\eqref{eq:excellenteq1}.

Since $(A,B)$ is $(\gamma,\eta)$-good, we get that that $\frac{|W|}{|N_{I}(A)|} \geq \gamma$. Moreover, using that $\rho(F) = \beta$ and our choice of $\lambda$, by the Hoeffding bound (Lemma \ref{lem:hoeffding}) the density of $W \setminus W'$ inside $N_\Is(A)$ is at most $\lambda$. Consequently, we  have that $\frac{|W'|}{|N_{I}(A)|} \geq \gamma - \lambda$. Using this estimate along with Eq.~\eqref{eq:excellenteq2},
\begin{align}
 \Pr_{\substack{x \in \univ \setminus A \\ B' \supseteq \{x\} \cup A}} 
\left[x \in F \wedge B' \in W \right]\geq \Pr_{B'\in N_{I}(A)}[B' \in W'] \cdot 
\Pr_{\substack{B' \in N_\Is(A) \\ x \in B'\setminus A}} 
\left[x \in F   \mid B' \in W' \right]  \geq (\gamma - \lambda) \cdot \beta/2,
\end{align}
which stands in contradiction to Eq.~\eqref{eq:excellenteq1}.
\end{proof}

We now prove \Cref{lem:excellent-implies-correctness}, which shows that selecting an excellence edge leads to correctness.

\begin{proof}[Proof of \Cref{lem:excellent-implies-correctness}]
Let $\beta$, $\lambda$, and $T$ be as in the statement of the lemma, and suppose that $(A,B)$ is a $(\gamma,\eta,\alpha)$-excellent edge. We show that
\[\E_{\substack{x \sim \univ \setminus A,\\\correctCircuitA ,\;\algoIJKW}}[\correctCircuitA = g(x)] \geq  1 - \beta - (1-\eta(\gamma-\lambda)/2)^T - 16 \alpha. \]

We first decouple the errors that stem from failing to find a consistent neighbour (in which case thee algorithm outputs $\bot$ and aborts) and those in which a consistent neighbour was found yet still the algorithm output incorrectly. To this end, by a union bound, we have
\begin{equation}\label{eq:excel_implies_correct}
\Pr_{\substack{x \in \univ \setminus A\\\correctCircuitA,\;\algoIJKW}}[\correctCircuitA(x) \ne g(x)] \leq 
\Pr_{\substack{x \in \univ \setminus A\\\correctCircuitA,\;\algoIJKW}}[\correctCircuitA(x) = \bot]
+ \Pr_{\substack{x \in \univ \setminus A \\  y \sim \correctCircuitA(x),\;\algoIJKW}}[y \ne g(x) \mid y \ne \bot].
\end{equation}
We will first bound the first term of the RHS of Eq.~(\ref{eq:excel_implies_correct}). 

Let $F$ be the subset of $\univ \setminus A$ such that $x \in F$ iff strictly less than a $(\gamma - \lambda)/2$ fraction of $B' \in N_\Is(A,x)$ satisfy $\E[\algoIJKW(B')|_A = g^{k/2}(A)] \geq \eta$. We have from \Cref{lem:bound-algorithm-answers} that $\rho(F) < \beta$, where $\rho(F)$ denotes the measure of $F$ inside $\univ \setminus A$. 

We now compute the probability that $\correctCircuitA(x) = \bot$ for an arbitrary but fixed $x \not\in F$. Let
$W$ be the subset of $N_\Is(A,x)$, such that 
$B' \in W$ iff $\E[\algoIJKW(B')|_A = g^{k/2}(A)] \geq  \eta$. From the assumption that $x \not\in F$, we have that $\Pr_{B' \in N_\Is(A,x)}[B' \in W] \geq (\gamma - \lambda)/2$. Let $E_i$ be the event that on the $i$-th iteration of $\correctCircuitA$, $v'|_A \ne g^{k/2}(A)$, where $v'$ is sampled in Step 2.2 from the definition of the circuit $\correctCircuitA$. 
For a fixed $i \in \{1, \ldots, T\}$, and a fixed $x$, we have 
\begin{align*}
    \Pr[E_i] & = \Pr_{\correctCircuitA, \algoIJKW}[\algoIJKW(B')|_A \neq g^{k/2}(A) \mid B' \in N_{\Is}(A, x)] \\
    & = 1 - \Pr_{\correctCircuitA, \algoIJKW}[ \algoIJKW(B')|_A = g^{k/2}(A) \mid B' \in N_\Is(A, x)] \\
    & \leq 1 - \Pr_{\correctCircuitA, \algoIJKW}[ \algoIJKW(B')|_A = g^{k/2}(A) \mid  B' \in W'] \cdot \Pr_{B' \in N_{\Is}(A, x)}[B' \in W] \\
    & \leq 1 - \eta \cdot \frac{\gamma - \lambda}{2}.
\end{align*}
Using that in each iteration of $\correctCircuitA$ a fresh selection of $B' \sim N_\Is(A,x)$ and $v' \sim \algoIJKW(B')$ is made, we get that %
$$
\Pr_{\correctCircuitA,\;\algoIJKW}[\correctCircuitA(x) = \bot \mid x \notin F] 
\;\leq\; 
\Pr_{\correctCircuitA,\;\algoIJKW}[E_1 \wedge \cdots \wedge E_T] \;\leq\; (1 - \eta(\gamma-\lambda)/2)^T.
$$

Putting together the previous estimates, 
\begin{align*}
\Pr_{\substack{x \in \univ \setminus A\\\correctCircuitA,\;\algoIJKW}}[\correctCircuitA(x) = \bot]
& \leq \Pr_{\substack{x \in \univ \setminus A\\\correctCircuitA,\;\algoIJKW}}[x \in F]
+ \Pr_{\substack{x \in \univ \setminus A\\\correctCircuitA,\;\algoIJKW}}[\correctCircuitA(x) = \bot \mid  x \not\in F] \\
&\leq \beta + (1 - \eta(\gamma-\lambda)/2)^T.
\end{align*}

We now bound the second term of the RHS of Eq.~(\ref{eq:excel_implies_correct}). Assuming that the output of the circuit $\correctCircuitA(x)$ is not $\bot$, we have that at some iteration $\correctCircuitA(x)$ picks $B' \sim N_\Is(A,x)$ and $v' \sim \algoIJKW(B')$ such that $v'|_A = g^{k/2}(A)$.

For convenience, for $x \in \univ \setminus  A$ let us define $h(x)$ to be the conditional probability that $\correctCircuitA$ produces an incorrect answer on $x$ (over the randomness of $\correctCircuitA$ and $\algoIJKW$), given that it does not output $\bot$. Our goal is to show that 
\begin{equation}\label{eq:IJKW_error_bound}
\Pr_{\substack{x \in \univ \setminus A \\  y \sim \correctCircuitA(x),\;\algoIJKW}}[y \ne g(x) \mid y \ne \bot] = \E_{x \sim \univ \setminus A}[h(x)]
\end{equation}
is small.  To show this, it will be convenient to use the following notation. For~$B \in N_\Is(A)$,\\ 

\noindent $p_\mathsf{cons}(B) = \Pr_{\algoIJKW}\big[\algoIJKW(B)|_A = g^{k/2}(A)\big]$.\\ 
\noindent $p_{\mathsf{tot}}(A) = \sum_{B \in N_\Is(A)} p_{\mathsf{cons}}(B)$.\\ 
\noindent $p_{\mathsf{tot}}(x) = \sum_{B \in N_\Is(A,x)} p_{\mathsf{cons}}(B)$.\\ 
\noindent $p_{\mathsf{used}}(x,B) = p_{\mathsf{cons}}(B)/p_{\mathsf{tot}}(x)$.\\
\noindent $p_{\mathsf{err}}(x,B) = \Pr_{v \sim \algoIJKW(B)}[v|_x \neq g(x) \mid v|_A = g^{k/2}(A)]$.\\

\noindent Note that all these values  depend on $A$. It is not hard to see that 
\begin{align}
h(x) = \sum_{B \in N_\Is(A,x)} p_{\mathsf{err}}(x,B) \cdot p_{\mathsf{used}}(x,B).
\end{align}
Moreover, since
$$
p_{\mathsf{err}}(x,B)  = \frac{\Pr_{v \sim \algoIJKW(B)}[v|_x \neq g(x) \wedge v|_A = g^{k/2}(A)]}{\Pr_{\algoIJKW}[\algoIJKW(B)|_A = g^{k/2}(A)]} = 
\frac{\Pr_{v \sim \algoIJKW(B)}[v|_x \neq g(x) \wedge v|_A = g^{k/2}(A)]}{p_\mathsf{cons}(B)},
$$
we have
$$
h(x) = \frac{\sum_{B'' \in N_\Is(A,x)} \Pr_{ \algoIJKW}[\algoIJKW(B'')|_x \neq g(x) \wedge \algoIJKW(B'')|_A = g^{k/2}(A)]}{\sum_{B'' \in N_\Is(A,x)}P_{\algoIJKW}[\algoIJKW(B'')|_A = g^{k/2}(A)]},
$$
or equivalently,
\begin{equation}\label{eq:hx_alternate}
h(x) = \frac{\sum_{B'' \in N_\Is(A,x)} \Pr_{ v'' \sim \algoIJKW(B'')}[v''|_x \neq g(x) \wedge v''|_A = g^{k/2}(A)]}{p_{\mathsf{tot}}(x)}.
\end{equation}

Next, we upper bound the quantity on the RHS of Equation (\ref{eq:IJKW_error_bound}). For convenience, let 
$$\mu = \E_{x \sim \univ\setminus A}[h(x)].$$

As we are interested in bounding the probability that $\correctCircuitA$ outputs an erroneous value when it samples an excellent edge, we would like to relate $\mu$ to the excellence parameter $\alpha$. We do this by first relating $\alpha$ and then $\mu$ to the following intermediate expression:

\begin{equation}\label{eq:intermediate}
\E_{B' \sim \mathcal{W}_I(A)} \Big [ \E_{x \sim \NB}[h(x)] \Big ] = \frac{1}{p_{\mathsf{tot}(A)}} \cdot  \sum_{B' \in N_\Is(A)} p_{\mathsf{cons}}(B') \E_{x \sim \NB}[h(x)],
\end{equation}
where $\mathcal{W}_I(A)$ is the distribution supported on $N_\Is(A)$ with each $B'\in N_\Is(A)$ being sampled with probability $p_\mathsf{cons}(B')/p_{\mathsf{tot}}(A)$. 

On the one hand, using Equation (\ref{eq:hx_alternate}), we can rewrite the RHS of Equation (\ref{eq:intermediate}) as:
$$
\frac{1}{p_{\mathsf{tot}}(A)} \cdot \sum_{B' \in N_\Is(A)} p_{\mathsf{cons}}(B') \cdot \frac{1}{k/2} \cdot \sum_{x \in \NB} \frac{1}{p_{\mathsf{tot}}(x)} \cdot \sum_{B'' \in N_\Is(A,x)} \Pr_{\algoIJKW}[\algoIJKW(B'')|_x \neq g(x) \wedge \algoIJKW(B'')|_A = g^{k/2}(A)].
$$
Note that in the expression above, every fixed $(x, B'')$, where $x \in \univ\setminus A$ and $B'' \in N_\Is(A, x)$, contributes a value %
$$
\frac{1}{p_{\mathsf{tot}}(A)} \cdot \left ( \sum_{B' \in N_\Is(A,x)} p_{\mathsf{cons}}(B') \right) \cdot \frac{1}{k/2} \cdot \frac{1}{p_{\mathsf{tot}}(x)} \cdot \Pr_{\algoIJKW}[\algoIJKW(B'')|_x \neq g(x) \wedge \algoIJKW(B'')|_A = g^{k/2}(A)], 
$$
since it appears in the sum with a corresponding factor $p_{\mathsf{cons}}(B')$ for each $B'$ in $N_\Is(A,x)$. Since the corresponding sum of $p_{\mathsf{cons}}(B')$ is precisely $p_{\mathsf{tot}}(x)$, it follows that every $(x,B'')$ appears in Equation (\ref{eq:intermediate}) with a contribution of
\begin{equation}\label{eq:edge_contrib}
    \frac{1}{p_{\mathsf{tot}}(A)} \cdot \frac{1}{k/2} \cdot \Pr_{\algoIJKW}[\algoIJKW(B'')|_x \neq g(x) \wedge \algoIJKW(B'')|_A = g^{k/2}(A)],
\end{equation}
and Equation (\ref{eq:intermediate}) is precisely the sum of these contributions over all $(x, B'')$.

On the other hand, using that the edge $(A,B)$ from the statement of the lemma is $\alpha$-excellent, we know that
\begin{eqnarray} \nonumber
\alpha & \geq & \E_{B' \sim \mathcal{W}_I(A)} \left [ \E_{x \sim \NB}[p_{\mathsf{err}}(x,B') ] \right ] \\ \nonumber
& = & \frac{1}{p_{\mathsf{tot}}(A)} \cdot \sum_{B' \in N_\Is(A)} p_\mathsf{cons}(B') \E_{x \sim \NB}[p_{\mathsf{err}}(x,B')] \\ \nonumber
& = &  \frac{1}{p_{\mathsf{tot}}(A)} \cdot \sum_{B' \in N_\Is(A)} \frac{1}{k/2} \sum_{x \in \NB} \Pr_{\algoIJKW}[\algoIJKW(B')|_x \neq g(x) \wedge \algoIJKW(B')|_A = g^{k/2}(A)]. \nonumber
\end{eqnarray}
In the expression above, each edge $(x,B')$ also contributes a value equal to that in Equation (\ref{eq:edge_contrib}). 
Consequently, we get from the discussion above that
\begin{equation}\label{eq:alpha_upper_bound}
\alpha \;\geq\; \E_{B' \sim \mathcal{W}_I(A)} \Big [ \E_{x \sim \NB}[h(x)] \Big ].
\end{equation}
Next, we move on to relate $\mu = \E_{x \sim \univ \setminus A}[h(x)]$ to the RHS of this inequality. Unlike in the classical case, this calculation is more involved as the $B'$ in~\Cref{eq:alpha_upper_bound} is not sampled uniformly from $N_\Is(A)$ but from $\mathcal{W}_\Is(A)$ (where sampling $B'$ is proportional to $p_{\mathsf{cons}}(B')$). To deal with this, we partition the elements $B' \in N_\Is(A)$ into buckets based on their value $p_{\mathsf{cons}}(B') \in [0,1]$. Let $\ell = 5 \cdot \log (1/(\gamma \cdot \eta))$, and for $i \in \{0, 1, \ldots, \ell\}$, set 
$$
\Gamma_i = \{B' \in N_\Is(A) \mid p_{\mathsf{cons}}(B') \in (2^{-i - 1}, 2^{-i}]\,\}.
$$
We also define the exceptional bucket $\Gamma_{\ell + 1}$ as follows:
$$
\Gamma_{\ell + 1} = \{B' \in N_\Is(A) \mid p_{\mathsf{cons}}(B') \leq 2^{-\ell -1}\,\}. 
$$
Note that the buckets are disjoint, and that $N_\Is(A) = \bigcup_{i = 0}^{\ell + 1} \Gamma_i$. For $B' \in N_\Is(A)$, let $\mu(B') = \E_{x \sim \NB}[h(x)]$. Then
\begin{equation}\label{eq:prob_decomposition}
\begin{aligned}
\E_{B' \sim \mathcal{W}_I(A)} \Big [ \E_{x \sim \NB}[h(x)] \Big ] &= \E_{B' \sim \mathcal{W}_I(A)}\big [\mu(B') \big ] \\
&= \sum_{i = 0}^{\ell + 1} \E_{B' \sim \mathcal{W}_I(A)} \big [\mu(B') \mid B' \in \Gamma_i \big] \cdot \Pr_{B' \sim \mathcal{W}_I(A)} \big [B' \in  \Gamma_i \big].
\end{aligned}
\end{equation}

Our goal is to lower bound the expression in~\Cref{eq:prob_decomposition}. Towards that, we aim to bound each term in the expression individually. The following simple claim will be useful. It shows that, in our analysis, for $0 \leq i \leq \ell$, we can replace sampling a $B'$ from $\Gamma_i$ according to $\mathcal{W}_\Is(A)$ with sampling a uniformly random $B' \sim \Gamma_i$. %

\begin{claim}\label{c:transfer_dist}
For every $0 \leq i \leq \ell$,
$$
\E_{B' \sim \mathcal{W}_I(A)}\big [ \mu(B') \mid B' \in \Gamma_i \big ] \;\geq\; \frac{1}{2} \cdot \E_{B' \sim \Gamma_i}\big [\mu(B') \big].
$$
\end{claim}

\begin{proof}[Proof of Claim \ref{c:transfer_dist}]
Indeed,
\begin{align*}
\E_{B' \sim \mathcal{W}_I(A)}\big [ \mu(B') \mid B' \in \Gamma_i \big ] & =  \frac{\sum_{B' \in \Gamma_i} \frac{p_{\mathsf{cons}}(B')}{p_\mathsf{tot}(A)} \cdot \mu(B')}{(1/p_{\mathsf{tot}}(A)) \cdot \sum_{B' \in \Gamma_i} p_{\mathsf{cons}}(B')} \nonumber \\
& =  \frac{
\sum_{B' \in \Gamma_i} p_{\mathsf{cons}}(B') \cdot \mu(B')}
{\sum_{B' \in \Gamma_i} p_{\mathsf{cons}}(B')}  \nonumber \\
& \geq \frac{\sum_{B' \in \Gamma_i} 2^{-i-1} \cdot \mu(B')}{|\Gamma_i| \cdot 2^{-i}} \nonumber \\
& =  \frac{1}{2} \cdot \left ( \frac{1}{|\Gamma_i|} \cdot \sum_{B' \in \Gamma_i}\mu(B') \right ) \nonumber \\
& =  \frac{1}{2} \cdot \E_{B' \sim \Gamma_i} \big [  \mu(B') \big ] . \qedhere \nonumber
\end{align*}
\end{proof}

Notice that this bound does not work for $\Gamma_{\ell+1}.$ However, it suffices for us to show that $\Pr_{B' \sim \mathcal{W}_\Is(A)}[B' \in \Gamma_{\ell+1}]$ is ``small" and can be omitted while determining the lower bound for~\Cref{eq:prob_decomposition}. In fact, we are able to claim something slightly stronger as shown next. For $0 \leq i \leq \ell$, we say that bucket $\Gamma_i$ is \emph{large} if $|\Gamma_i| \geq (\gamma \cdot \eta)^5 \cdot |N_\Is(A)|$. Otherwise, we say that it is \emph{small}. For convenience, let $w_i = \sum_{B' \in \Gamma_i} p_{\mathsf{cons}}(B')$. Recall that $(A,B)$ is a $(\gamma, \eta)$-good edge, and therefore
\begin{equation}\label{eq:ptot_wi}
\sum_{i = 0}^{\ell + 1} w_i \;=\; p_{\mathsf{tot}}(A) \;\geq\; \gamma \cdot \eta \cdot |N_\Is(A)|.
\end{equation}

\begin{claim}\label{c:smallGammai} The following upper bounds hold.
\begin{itemize}
    \item[\emph{(}i\emph{)}] For $0 \leq i \leq \ell$, if $\Gamma_i$ is small then
$$
\Pr_{B' \sim \mathcal{W}_I(A)}[B' \in \Gamma_i] \leq (\gamma \cdot \eta)^4.
$$
\item[\emph{(}ii\emph{)}] Moreover, in the special case where $i = \ell + 1$, we have
$$
\Pr_{B' \sim \mathcal{W}_I(A)}[B' \in \Gamma_{\ell + 1}] \leq (\gamma \cdot \eta)^4.
$$
\end{itemize}
\end{claim}

\begin{proof}[Proof of Claim \ref{c:smallGammai}]
For the proof of Item (\emph{i}), we rely on Equation (\ref{eq:ptot_wi}) and on the smallness of~$\Gamma_i$:
$$
\Pr_{B' \sim \mathcal{W}_I(A)}[B' \in \Gamma_i] = \frac{w_i}{p_{\mathsf{tot}}(A)} \leq \frac{|\Gamma_i| \cdot 2^{-i}}{p_{\mathsf{tot}}(A)} \leq \frac{|\Gamma_i|}{p_{\mathsf{tot}}(A)} \leq \frac{(\gamma \cdot \eta)^5 \cdot |N_\Is(A)|}{\gamma \cdot \eta \cdot |N_\Is(A)} = (\gamma \cdot \eta)^4.
$$

For the proof of Item (\emph{ii}), we rely on Equation (\ref{eq:ptot_wi}) an on the upper bound on $p_{\mathsf{cons}}(B')$ for $B' \in \Gamma_{\ell + 1}$:
$$
\Pr_{B' \sim \mathcal{W}_I(A)}[B' \in \Gamma_{\ell + 1}] = \frac{w_{\ell + 1}}{p_{\mathsf{tot}}(A)} \leq \frac{|\Gamma_{\ell+1}| \cdot 2^{-\ell - 1}}{p_{\mathsf{tot}}(A)} \leq \frac{|N_\Is(A)| \cdot 2^{-\ell - 1} }{p_{\mathsf{tot}}(A)} \leq \frac{2^{-\ell - 1}}{\gamma \cdot \eta}  \leq (\gamma \cdot \eta)^4,
$$
where the last inequality uses $\ell = 5 \cdot (\log (1/\gamma \cdot \eta))$.
\end{proof}

Intuitively, Claim \ref{c:smallGammai} shows that the only significant terms in Equation (\ref{eq:prob_decomposition}) are the ones coming from buckets $\Gamma_i$ with $0 \leq i \leq \ell$ that are large. Moreover, since there are $O(\log(1/\gamma \eta))$ terms, the combined probability weight of the insignificant terms is also small. 

\begin{claim}\label{c:prob_in_large_bucket}
A random $B' \sim \mathcal{W}_I(A)$ is likely to belong to a large bucket $\Gamma_i$ with $0 \leq i \leq \ell$, i.e.,
$$
\Pr_{B' \sim \mathcal{W}_I(A)}[B'~\text{is in a large bucket}~\Gamma_i~\text{for some}~0 \leq i \leq \ell\,] \geq 1/2.
$$
\end{claim}

\begin{proof}[Proof of Claim \ref{c:prob_in_large_bucket}]
Using Claim \ref{c:smallGammai} and a union bound over buckets,
$$
\Pr_{B' \sim \mathcal{W}_I(A)}[B'~\text{is in a small bucket or}~B' \in \Gamma_{\ell + 1}] \leq (\ell + 2) \cdot (\gamma \eta)^4 \leq 6 \cdot \log(1/(\gamma \eta)) \cdot (\gamma \eta)^4 \leq 1/2,
$$
where the last inequality uses the assumption of the lemma that $\gamma , \eta \leq 1/10$.
\end{proof}

Finally, the next claim establishes that when $0 \leq i \leq \ell$ and $\Gamma_i$ is large, $\E_{B' \sim \Gamma_i}\big [\mu(B') \big] = \Omega(\mu)$, provided that $\mu$ is not too small and $k$ is large enough. (Recall that if $\mu$ is small, specifically less than the excellence parameter $\alpha$, we are already done.) 

\begin{claim}\label{c:IJKW_like_bound}
Suppose that $\mu \geq \alpha$. Let $0 \leq i \leq \ell$ and $\Gamma_i$ be a large bucket. Then
$$
\E_{B' \sim \Gamma_i} \big [  \mu(B') \big ] \;\geq\; \frac{\mu}{2}  \cdot \left ( 1 - \frac{2e^{-k \alpha/12}}{(\gamma \eta)^5}  \right ) \;\geq\; \frac{\mu}{4}. 
$$
\end{claim}

\begin{proof}[Proof of Claim \ref{c:IJKW_like_bound}]
Let 
$$
\mathsf{Bad} = \{B' \in N_\Is(A) \mid \mu(B') < \mu/2\}.
$$
Then, by the Hoeffding bound, $\Pr_{B' \sim N_\Is(A)}[B' \in \mathsf{Bad}] \leq 2e^{-k\mu/12} \leq 2e^{-k\alpha/12}$, using that $\mu \geq \alpha$. Note that, since $\Gamma_i$ is large, $|\Gamma_i| \geq (\gamma \eta)^5 \cdot |N_\Is(A)|$. For convenience, let $\lambda = 2e^{-k\alpha/12}$. Then
$$
\frac{|\mathsf{Bad}|}{|\Gamma_i|} \leq \frac{\lambda \cdot |N_\Is(A)|}{(\gamma \eta)^5 |N_\Is(A)|} = \frac{\lambda}{(\gamma \eta)^5}.
$$
Consequently,
\begin{eqnarray}
\E_{B' \sim \Gamma_i} \big [  \mu(B') \big ] & \geq & \E_{B' \sim \Gamma_i} \big [ \mu(B') \mid B' \notin \mathsf{Bad} \big ] \cdot \Pr_{B' \sim \Gamma_i} \big [ B' \notin \mathsf{Bad} \big ] \nonumber \\
& = & \E_{B' \sim \Gamma_i \setminus \mathsf{Bad}} \big [ \mu(B') \big ] \cdot \frac{|\Gamma_i| - |\mathsf{Bad}|}{|\Gamma_i|} \nonumber \\
& \geq & \frac{\mu}{2} \cdot \left ( 1 - \frac{|\mathsf{Bad}|}{|\Gamma_i|} \right ) \;\geq\;  \frac{\mu}{2} \cdot \left ( 1 - \frac{\lambda}{(\gamma \eta)^5} \right ),\nonumber
\end{eqnarray}
where the second inequality used the definition of $\mathsf{Bad}$.  Claim \ref{c:IJKW_like_bound} follows using the value of $\lambda$ and the hypothesis of the lemma.
\end{proof}

We are ready to conclude the proof of Lemma \ref{lem:excellent-implies-correctness}. If $\mu \leq \alpha$, we are done. So we assume from now on that $\mu \geq \alpha$. From Equations (\ref{eq:alpha_upper_bound}) and (\ref{eq:prob_decomposition}), we have
\begin{eqnarray}
\alpha & \geq &\E_{B' \sim \mathcal{W}_I(A)} \Big [ \E_{x \sim \NB}[h(x)] \Big ]\nonumber\\
& = & \sum_{i = 0}^{\ell + 1} \E_{B' \sim \mathcal{W}_I(A)} \big [\mu(B') \mid B' \in \Gamma_i \big] \cdot \Pr_{B' \in \mathcal{W}_I(A)} \big [B' \in  \Gamma_i \big] \nonumber \\
(\text{Omitting terms + Claim}~ \ref{c:transfer_dist}) & \geq & \sum_{\substack{0 \leq i \leq \ell \\ \Gamma_i~\text{is large}}} \frac{1}{2} \cdot \E_{B' \sim \Gamma_i} \big [\mu(B') \big ] \cdot \Pr_{B' \sim \mathcal{W}_I(A)} \big [B' \in \Gamma_i \big] \nonumber \\
(\text{By Claim}~\ref{c:IJKW_like_bound}) & \geq & \sum_{\substack{0 \leq i \leq \ell \\ \Gamma_i~\text{is large}}} \frac{\mu}{8} \cdot \Pr_{B' \sim \mathcal{W}_I(A)} \big [B' \in \Gamma_i \big] \nonumber \\ 
& = & \frac{\mu}{8} \cdot \Pr_{B' \sim \mathcal{W}_I(A)} \big [  B'~\text{is in a large bucket}~\Gamma_i~\text{for some}~0 \leq i \leq \ell \big ] \nonumber \\
(\text{By Claim}~\ref{c:prob_in_large_bucket}) & \geq & \frac{\mu}{16}. \nonumber
\end{eqnarray}
This shows that $\mu \leq 16 \alpha$, which completes the proof of Lemma \ref{lem:excellent-implies-correctness}.
\end{proof}

\subsubsection{Excellent edges are abundant}
\label{sec:excellent-is-abundant}
 In this section, we prove \Cref{lem:excellent-is-abundant}, which loosely speaking, shows that if the algorithm $\algoIJKW$ non-trivially agrees with $g^k$, then there are many edges that satisfy the excellence condition.
Recall that an edge is said to be $(\eta , \gamma, \alpha)$-{\em excellent} if it is: (1) $(\eta, \gamma)$-good, and (2) satisfies
\begin{equation*}
    \E_{\substack{B' \sim \mathcal{W}_\Is(A)}} \left[ \ErrCons(A,B')
    \right] \leq \alpha,
\end{equation*}
where $\mathcal{W}_\Is(A)$ gives each edge $(A,B')$ weight according to its probability of being correct on $A$ (see \cref{def:excellent}). Our first lemma shows that the first condition holds, i.e., that under the assumption that our algorithm $\algoIJKW$ computes $g^k$ with probability at least $\eps$, there are many good edges. %
\begin{lemma}
\label{lem:Agoodgoodsetmany}
If $\E_{B\in  \mathcal{S}_k}\left[ \mathsf{Corr}(B) \right] \geq \eps$, for any $0 < \eta,\gamma,\xi$ such that $\eta + \gamma + \xi = \eps$, we have that at least a $\xi$-fraction of $(A,B) \in \Is$ are $(\gamma,\eta)$-good.
\end{lemma}
\begin{proof}
We prove the contrapositive statement in the lemma. In this direction, let
\begin{align}\label{eq:assumption-proof-good}
|\big\{(A,B)\in \Is: B\in N_\Is(A) \text{ and } (A,B) \text{ is } (\gamma,\eta)\text{-good}\big\}| < \xi \cdot  |I|.
\end{align} 

 We define the set
\[G = \left\{A \in  \mathcal{S}_{k/2} : 
\E_{B' \in N_\Is(A)}[\mathsf{Corr}(B') \geq \eta] \geq \gamma 
\right\}\]
of $A$s who have at least an $\gamma$-fraction of $\eta$-correct neighbors.

By our assumption in \Cref{eq:assumption-proof-good}, we have that $|G| < \xi |\mathcal{S}_{k/2}|$ and thus
\begin{align}\label{eq:good1}
\E_{A \in \mathcal{S}_{k/2}}[A \in G]\cdot \E_{\substack{A \in \mathcal{S}_{k/2}, \\B \in N_{I}(A) }}\left[  \mathsf{Corr}(B) \middle| A \in G \right] < \xi \cdot 1 = \xi.
\end{align}
We also have that
\begin{align}\label{eq:good2}
    &\E_{\substack{A \in \mathcal{S}_{k/2} \\ B \in N_{I}(A)}}\left[A \not\in G \wedge  \mathsf{Corr}(B) < \eta \right]\cdot 
\E_{\substack{A \in \mathcal{S}_{k/2} \\ B \in N_{I}(A)} }\left[ \mathsf{Corr}(B) \middle| \substack{A \not\in G \\     \mathsf{Corr}(B) < \eta}
    \right]  \leq 1 \cdot \eta = \eta,
\end{align}
and
\begin{align}\label{eq:good3}
    &\E_{\substack{A \in \mathcal{S}_{k/2} \\ B \in N_{I}(A)}}\left[A \not\in G \wedge \mathsf{Corr}(B) \geq \eta \right]\cdot 
\E_{\substack{A \in \mathcal{S}_{k/2} \\ B \in N_{I}(A)} }\left[ \mathsf{Corr}(B) \middle| \substack{A \not\in G \\     \mathsf{Corr}(B) \geq \eta}
    \right]  < \gamma \cdot 1 = \gamma,
\end{align}
where the first inequality above used the fact that conditioned on $A\notin G$, then the probability of a uniformly random $B'\in N_\Is(A)$ being $\eta$-correct is at most $\gamma$ (by definition of $G$).

We have that 
\begin{align*}
    \E_{B\in  \mathcal{S}_k}\left[ \mathsf{Corr}(B) \right] 
    = \E_{\substack{A \in \mathcal{S}_{k/2},\\ B \in N_{I}(A) }}\left[  \mathsf{Corr}(B) \right] < \text{ Eq.}~(\ref{eq:good1}) +  \text{Eq.}~(\ref{eq:good2}) + \text{Eq.}~(\ref{eq:good3}) < \xi+\eta+\gamma = \eps,
\end{align*}
where the first equality above used that we can obtain the uniform distribution on $B\in S_k$ by uniformly picking $A\in S_{k/2}$ and then considering a random $B\in N_\Is(A)$; and the last equality is by the assumption of the lemma. This concludes the proof of the statement. 
\end{proof}

Next we will use \cref{lem:Agoodgoodsetmany} to strengthen the foregoing conclusion and show that if $\algoIJKW$ computes $g^k$ with probability at least $\eps$,
then not only is the number of good edges large, but rather the number of \emph{excellent} edges is also large.
Formally, we prove that if $\E_{B\in  \mathcal{S}_k}\left[ \mathsf{Corr}(B) \right] \geq \eps$ then at least an $(\varepsilon/3 - \frac{62208}{\alpha^3 \cdot \varepsilon^5} \cdot e^{-\frac{\alpha}{96}\cdot k})$-fraction of the edges $(A,B) \in \Is$ are $(\eps/3,\eps/3,\alpha)$-excellent. Note that the fraction of edges with this property is at least $\varepsilon/6$ if $k = \frac{100}{\alpha} \cdot (15 + 3 \log(1/\alpha) + 6 \log (1/\varepsilon))$, so this does show that a noticeable fraction of edges are excellent as long as $k$ is not too small.

Compared with its counterpart in \cite{impagliazzo2010uniform}, in the proof of \Cref{lem:excellent-is-abundant} we face additional difficulties due to the asymmetries in the definition of excellent edges, which in this paper refer to the more involved distribution $\mathcal{W}_\mathcal{I}(A)$.

\begin{proof}[Proof of \Cref{lem:excellent-is-abundant}] Consider random choices of $A \sim \mathcal{S}_{k/2}$ and $B \sim N_\Is(A)$. We would like to lower bound
$$
\Pr_{A, B}[(A,B)~\text{is}~(\varepsilon/3, \varepsilon/3, \alpha)\text{-excellent}].
$$
In order to show that an $(\eps/3, \eps/3)$-good edge $(A,B)\in \Is$ is also  $(\eps/3, \eps/3, \alpha)$-excellent, we need to show that
$$
\E_{B' \sim \mathcal{W}_I(A)}[\mathsf{ErrCons}(A,B')] \leq \alpha.
$$

It will be useful to introduce the following probability space and event. In addition to sampling $A \sim \mathcal{S}_{k/2}$ and $B \sim N_\Is(A)$, independently sample $B' \sim N_\Is(A)$, $v' = \algoIJKW(B')$, and $x \sim \NB$. Let $\mathsf{Err}(B',v') = \{x \in \NB \mid v'|_x \neq g(x)\}$ be the set of $x$s in $B'\setminus A$ for which $v'|_x$ disagrees with $g(x)$, and $\mathsf{err}(B',v') = |\mathsf{Err}(B',v')|/(k/2)$. We introduce the following event. \\

\noindent $\calE(A,B,B',v',x)$: The following conditions hold:\\~\\
\indent $v'|_A = g^{k/2}(A)$.\\
\indent $\mathsf{err}(B', v') > \alpha/4$.\\
\indent $v'|_x \neq g(x)$.\\

\noindent By symmetry, when analysing $\Pr[\calE]$ we can select $B'$ first (accordingly sampling  $v'$ and $x$ depending on it), followed by choices of a random $A$ contained in $B'$ and a random $B$ that contains $A$. Note that
\[
\Pr[\calE] = \frac{1}{|\calS_k|} \sum_{B' \in \calS_k} \Pr[\calE ~\mid~ B'~\text{fixed}],
\]
 by ``$B'$ fixed'' we mean the event that the random choice gives a fixed $B' \in \mathcal{S}_k$. We now decompose each term in the sum above according to the possible values of $\lambda = \mathsf{err}(B', v')$ for $v' \sim \algoIJKW(B')$. Since for $\calE$ to hold this value must be larger than $\alpha/4$, we get for each fixed $B'$:
$$
\Pr[\calE ~|~ B'~\text{fixed}] = \sum_{\lambda > \alpha/4} \Pr[\calE  ~|~ B'~\text{is fixed} \wedge \mathsf{err}(B', v') = \lambda] \cdot \Pr_{v'}[\mathsf{err}(B', v') = \lambda].
$$
For a fixed $B'$ and any choice of $v'$ with $\mathsf{err}(B',v') = \lambda$, if $\mathsf{Err}(B',v') \cap A \neq \emptyset$ then it cannot be the case that $v'|_A = g^{k/2}(A)$. Consequently, for $\calE$ to hold $A$ must avoid the set $\mathsf{Err}(B',v')$, and in addition, $x$ must be selected from $\mathsf{Err}(B',v')$. It follows from the Hoeffding bound using $\lambda > \alpha/4$~that
$$
\Pr[\calE ~|~ B'~\text{is fixed} \wedge \mathsf{err}(B', v') = \lambda] \leq  2 \cdot e^{-\frac{\alpha}{96} \cdot k} \cdot \lambda.
$$
Putting together these estimates, and using that $\lambda \leq 1$, we have 
\begin{eqnarray}
\Pr[\calE] & \leq & \frac{1}{|\mathcal{T}|} \sum_{B' \in \mathcal{T}} \sum_{\lambda > \alpha/4} 2 \cdot e^{-\frac{\alpha}{96} \cdot k} \cdot  \Pr_{v' \sim \algoIJKW(B')}[\mathsf{err}(B', v') = \lambda] \nonumber \\
& \leq & \frac{1}{|\mathcal{T}|} \sum_{B' \in \mathcal{T}} 2 \cdot e^{-\frac{\alpha}{96} \cdot k} \cdot \sum_{\lambda > \alpha/4} \Pr_{v' \sim \algoIJKW(B')}[\mathsf{err}(B', v') = \lambda] \nonumber \\
& \leq & \frac{1}{|\mathcal{T}|} \sum_{B' \in \mathcal{T}} 2 \cdot e^{-\frac{\alpha}{96} \cdot k}  \leq  2 \cdot e^{-\frac{\alpha}{96}  \cdot k}. \nonumber
\end{eqnarray}

Next, notice that
\begin{align*}
& \E_{B'' \sim \mathcal{W}_I(A)}[\mathsf{ErrCons}(A,B'')] \nonumber\\
& \quad = \frac{1}{p_{\mathsf{tot}}(A)} \sum_{B'' \in N_\Is(A)} p_{\mathsf{cons}}(B'') \E_{x \sim \NBp} \left [\Pr_{v'' \sim \algoIJKW(B'')} \left [v''|_x \neq g(x) \;\mid\; v''|_A = g^{k/2}(A) \right] \right] \\
& \quad = \frac{1}{p_{\mathsf{tot}}(A)} \sum_{B'' \in N_\Is(A)}  \E_{x \sim \NBp} \left [\Pr_{v'' \sim \algoIJKW(B'')} \left [v''|_x \neq g(x) \wedge v''|_A = g^{k/2}(A) \right] \right]\\
& \quad =  \frac{1}{p_{\mathsf{tot}}(A)} \sum_{B'' \in N_\Is(A)} \frac{1}{(k/2)} \sum_{x \in \NBp} \Pr_{v'' \sim \algoIJKW(B'')} \left [v''|_x \neq g(x) \wedge v''|_A = g^{k/2}(A) \right] \\
& \quad = \frac{1}{p_{\mathsf{tot}}(A)} \sum_{B'' \in N_\Is(A)} \Pr_{\substack{v'' \sim \algoIJKW(B''),\\x \sim \NBp}}[v''|_A = g^{k/2}(A) \wedge x \in \mathsf{Err}(B'',v'')] \\
& \quad = \frac{1}{p_{\mathsf{tot}}(A)} \sum_{B'' \in N_\Is(A)} p_{\mathsf{cons}}(B'') \Pr_{v'' \sim \algoIJKW(B''),\;x \sim \NBp}[x \in \mathsf{Err}(B'',v'') \mid v''|_A = g^{k/2}(A)] \displaybreak  \\ 
& \quad = \E_{B'' \sim \mathcal{W}_I(A)}\left [\Pr_{\substack{v'' \sim \algoIJKW(B''),\\x \sim \NBp}}\Big[x \in \mathsf{Err}(B'',v'') \mid v''|_A = g^{k/2}(A)\Big ] \right ],
\end{align*}
where in the second equality we use the fact that
\[p_{\mathsf{cons}}(B'') \Pr_{v'' \sim \algoIJKW(B'')} \left [v''|_x \neq g(x) \;\mid\; v''|_A = g^{k/2}(A) \right]  = \Pr_{v'' \sim \algoIJKW(B'')} \left [v''|_x \neq g(x) \wedge v''|_A = g^{k/2}(A) \right].\]

This motivates the following definition. We say that a set $B''$ is $A$-heavy if 
$$
\Pr_{\substack{v'' \sim \algoIJKW(B''),\\x \sim \NBp}}\Big [x \in \mathsf{Err}(B'',v'') \mid v''|_A = g^{k/2}(A)\Big ] > \alpha/2.
$$
Then, in order to show that an $(\varepsilon/3, \varepsilon/3)$-good edge $(A,B)$ is $(\varepsilon/3, \varepsilon/3, \alpha)$-excellent, it suffices to prove that
$$
\Pr_{B'' \sim \mathcal{W}_I(A)}[B''~\text{is}~A\text{-heavy}] \leq \alpha/2.
$$

\noindent Consider the following event, which only depends on $A$ and $B$, but can also be considered over the probability space of event $\calE$:\\

\noindent $\calE_1(A,B)$: The following conditions hold:\\~\\ 
\indent $(A,B)$ is $(\varepsilon/3, \varepsilon/3)$-good.\\ 
\indent $\Pr_{B'' \sim \mathcal{W}_I(A)}[B''~\text{is}~A\text{-heavy}] > \alpha/2.$ (Note that this is a property of $A$.)\\

\noindent We want to show that the event $\calE_1$ happens with small probability. In other words, we show that $\Pr[\calE_1] = \Pr[\calE]/\Pr[\calE \mid \calE_1]$ is small by arguing that $\Pr[\calE \mid \calE_1]$ is not too small. We make use of the following claim, whose proof is deferred:

\begin{claim}\label{c:claim_star}
For every $(\varepsilon/3, \varepsilon/3)$-good edge $(A,B)$, 
\begin{center}
    if $\Pr_{B'' \sim \mathcal{W}_I(A)}[B''~ \text{is}~A\text{-heavy}] > \alpha/2$ then $\Pr_{B' \sim N_\Is(A)}[B'~ \text{is}~A\text{-heavy} \wedge p_{\mathsf{cons}}(B') \geq \varepsilon^*] > \frac{\alpha}{108} \cdot \varepsilon^3$,
\end{center}
where $\varepsilon^* \eqdef (\alpha/8)\cdot (\varepsilon^2/9)$ and $p_{\mathsf{cons}}(B') = \Pr_{v' \sim \algoIJKW(B')}[v'|_A = g^{k/2}(A)]$.
\end{claim}

Assuming this claim, we proceed as follows. Under event $\calE_1$, it follows from Claim \ref{c:claim_star} that 
$$
\Pr_{B' \sim N_\Is(A)}[B'~\text{is}~A\text{-heavy} \wedge p_{\mathsf{cons}}(B') \geq \varepsilon^*] \geq (\alpha/108) \cdot \varepsilon^3.
$$
Therefore, conditioning on $\calE_1$,  with probability at least $(\alpha/108) \cdot \varepsilon^3$ over the choices of $A$, $B$ and~$B'$ we get
\begin{equation}\label{eq:part1}
   p_{\mathsf{cons}}(B') \geq \varepsilon^* \quad \text{and} \quad \Pr_{v' \sim \algoIJKW(B'),\;x \sim \NB}[x \in \mathsf{Err}(B',v') \mid v'|_A = g^{k/2}(A) ] > \alpha/2. 
\end{equation}
We can write the latter probability  as follows:
$$
\Pr\big[x \in \mathsf{Err}(B',v') \wedge \mathsf{err}(B',v') \leq \alpha/4 \mid v'|_A = g^{k/2}(A)\big] + 
\Pr\big[x \in \mathsf{Err}(B',v') \wedge \mathsf{err}(B',v') > \alpha/4 \mid v'|_A = g^{k/2}(A)\big].
$$
Since the leftmost probability is at most $\alpha/4$, it follows that for a $B'$ of this form
$$
\Pr_{\substack{v' \sim \algoIJKW(B'),\\x \sim \NB}}[x \in \mathsf{Err}(B',v') \wedge \mathsf{err}(B',v') > \alpha/4 \mid v'|_A = g^{k/2}(A) ] \geq \alpha/4.
$$
Note that $x \in \mathsf{Err}(B',v')$ is equivalent to $v'|_x \neq g(x)$. Using this and rewriting the probability inequality above using the expression $\Pr[F_1 \mid F_2] = \Pr[F_1 \wedge F_2] / \Pr[F_2]$,
\begin{align*}
&\Pr_{\substack{v' \sim \algoIJKW(B'),\\x \sim \NB}}[v'|_x \neq g(x) \wedge \mathsf{err}(B',v') > \alpha/4 \wedge v'|_A = g^{k/2}(A)]\\
&=\Pr_{v' \sim \algoIJKW(B')}[v'|_x \neq g(x) \wedge \mathsf{err}(B',v') 
> \alpha/4 \mid v'|_A = g^{k/2}(A) ]
\Pr_{v' \sim \algoIJKW(B')}[v'|_A = g^{k/2}(A) ]
\\
&\geq (\alpha/4) \cdot p_{\mathsf{cons}}(B')\\
&\geq (\alpha/4) \cdot (\alpha/8) \cdot (\varepsilon^2/9),
\end{align*}
where we have used $p_{\mathsf{cons}}(B') \geq \varepsilon^*$ and $p_{\mathsf{cons}}(B') = \Pr_{v' \sim \algoIJKW(B')}[v'|_A = g^{k/2}(A)]$. Overall, combining this probability lower bound  and the probability that the conditions in Equation (\ref{eq:part1}) hold, it follows that
$$
\Pr[\calE \mid \calE_1] \geq (\alpha/108) \cdot \varepsilon^3 \cdot (\alpha/4) \cdot (\alpha/8) \cdot (\varepsilon^2/9) = \frac{\alpha^3 \cdot \varepsilon^5}{31104}. 
$$
Since $\Pr[\calE_1] \leq \Pr[\calE]/\Pr[\calE \mid \calE_1]$, 
$$
\Pr[\calE_1] \leq \frac{62208 \cdot e^{-\frac{\alpha}{96}\cdot k}}{\alpha^3 \cdot \varepsilon^5}.
$$

Finally, from the definition of event $\calE_1$ and the discussion above  we get 
$$
\Pr_{\substack{A \sim \mathcal{S}_{k/2},\\B \sim N_\Is(A)}}[(A,B)~\text{is}~(\varepsilon/3,\varepsilon/3,\alpha)\text{-excellent}] \geq \Pr_{A,\;B}[(A,B)~\text{is}~(\varepsilon/3,\varepsilon/3)\text{-good}] - \Pr_{A,\;B}[\calE_1].
$$
This implies using Lemma \ref{lem:Agoodgoodsetmany} and our probability estimate for $\calE_1$ that the probability that a random edge $(A,B)$ is $(\varepsilon/3, \varepsilon/3, \alpha)$-excellent is at least 
$$
\varepsilon/3 - \frac{62208 \cdot e^{-\frac{\alpha}{96}\cdot k}}{\alpha^3 \cdot \varepsilon^5}.
$$

In order to complete the argument, it remains to establish Claim \ref{c:claim_star}.\\

\begin{proof}[Proof of Claim \ref{c:claim_star}]
Let $(A,B)$ be an $(\varepsilon/3, \varepsilon/3)$-good edge, and assume that
$$
\Pr_{B'' \sim \mathcal{W}_I(A)}[B''~\text{is}~A\text{-heavy}] > \alpha/2.
$$
We need to prove that
$$
\Pr_{B' \sim N_\Is(A)}[B'~\text{is}~A\text{-heavy} \wedge p_{\mathsf{cons}}(B') \geq \varepsilon^*] > (\alpha/108) \cdot \varepsilon^3,
$$
where $\varepsilon^* = (\alpha/8)\cdot (\varepsilon^2/9)$. 

Note that the  probability of interest can be rewritten as 
\begin{align}
\label{eq:claimmainstatement426}    
\Pr_{B' \sim N_\Is(A)}[B'~\text{is}~A\text{-heavy} \mid p_{\mathsf{cons}}(B') \geq \varepsilon^*] \cdot \Pr_{B' \sim N_\Is(A)}[p_{\mathsf{cons}}(B') \geq \varepsilon^*].
\end{align}
Since $\varepsilon^* \leq \varepsilon/3$ and $(A,B)$ is $(\varepsilon/3, \varepsilon/3)$-good, the rightmost probability is at least $\varepsilon/3$. We lower bound the other probability next. 

Let $N_\Is(A, \geq \!\varepsilon^*) = \{B' \in N_\Is(A) \mid p_{\mathsf{cons}}(B') \geq \varepsilon^*\}$, and define $N_\Is(A, < \! \varepsilon^*)$ in a similar way. On the one hand,
\begin{align}
\label{eq:Aheavyconslarge}
\begin{aligned}
\Pr_{B' \sim N_\Is(A)}[B'~\text{is}~A\text{-heavy} \mid p_{\mathsf{cons}}(B') \geq \varepsilon^*] & =   \Pr_{B' \sim N_\Is(A, \geq \varepsilon^*)}[B'~\text{is}~A\text{-heavy}]  \\
& =   \frac{1}{|N_\Is(A,\geq \! \varepsilon^*)|} \cdot \sum_{B' \in N_\Is(A, \geq  \varepsilon^*)} \mathbf{1}_{[B'~\text{is}~A\text{-heavy}]}.
\end{aligned}
\end{align}

On the other hand, using the assumption of the claim,
\begin{eqnarray}
\alpha/2 & < & \frac{1}{p_{\mathsf{tot}}(A)} \left [ \sum_{B' \in N_\Is(A, \geq  \varepsilon^*)} p_{\mathsf{cons}}(B') \cdot \mathbf{1}_{[B'~\text{is}~A\text{-heavy}]} + \sum_{B' \in N_\Is(A, <  \varepsilon^*)} p_{\mathsf{cons}}(B') \cdot \mathbf{1}_{[B'~\text{is}~A\text{-heavy}]}  \right ] \nonumber \\
& \leq & \frac{1}{p_{\mathsf{tot}}(A)} \left [  \sum_{B' \in N_\Is(A, \geq  \varepsilon^*)} \mathbf{1}_{[B'~\text{is}~A\text{-heavy}]} + \sum_{B' \in N_\Is(A, <  \varepsilon^*)} p_{\mathsf{cons}}(B')   \right ]. \nonumber
\end{eqnarray}
This yields
$$
\sum_{B' \in N_\Is(A, \geq  \varepsilon^*)} \mathbf{1}_{[B'~\text{is}~A\text{-heavy}]} > (\alpha/2) \cdot p_{\mathsf{tot}}(A) - \sum_{B' \in N_\Is(A, <  \varepsilon^*)} p_{\mathsf{cons}}(B').
$$
In turn, thanks to our choice of $\varepsilon^*$,
\begin{eqnarray}
\sum_{B' \in N_\Is(A, <  \varepsilon^*)} p_{\mathsf{cons}}(B')  \leq  |N_\Is(A, < \! \varepsilon^*)| \cdot \varepsilon^* \leq |N_\Is(A)| \cdot \frac{\alpha}{8} \cdot \frac{\varepsilon^2}{9}< \frac{\alpha}{4} \cdot p_{\mathsf{tot}}(A), \nonumber
\end{eqnarray}
where the last inequality uses that $(A,B)$ is $(\varepsilon/3, \varepsilon/3)$-good, which yields $p_{\mathsf{tot}}(A) \geq |N_\Is(A)| \varepsilon^2/9$. As a consequence,
$$
\sum_{B' \in N_\Is(A, \geq  \varepsilon^*)} \mathbf{1}_{[B'~\text{is}~A\text{-heavy}]} > (\alpha/4) \cdot p_{\mathsf{tot}}(A) \geq \frac{\alpha}{4} \cdot \frac{\varepsilon^2}{9} \cdot |N_\Is(A)|.
$$
Overall, the probability we want to lower bound in Eq.~\eqref{eq:Aheavyconslarge} is at least
$$
\frac{1}{|N_\Is(A,\geq \! \varepsilon^*)|} \cdot  \frac{\alpha}{4} \cdot \frac{\varepsilon^2}{9} \cdot |N_\Is(A)| \geq \frac{\alpha}{4} \cdot \frac{\varepsilon^2}{9}.
$$
Combining our estimates, Eq.~\eqref{eq:claimmainstatement426} can be lower bound by $\frac{\alpha}{4} \cdot \frac{\varepsilon^2}{9}\cdot \varepsilon/3=\alpha/108\cdot \varepsilon^3$, which proves the claim. 
\end{proof}

This completes the proof of Lemma \ref{lem:excellent-is-abundant}. 
\end{proof}

\subsubsection{Extension to quantum circuits}

The next statement shares part of the terminology from Theorem \ref{thm:ijkw}, and we refer to the beginning of Section \ref{sec:IJKW} for more details. The part of the statement about uniformity refers to a fixed sequence of functions $g \colon \{0,1\}^n \to \{0,1\}$ indexed by $n$.

\begin{theorem}[Local list decoding for quantum circuits]\label{thm:ijkw_quantum}
There exists a universal constant $C \geq 1$ for which the following holds. Let $n \geq 1$ be a positive integer,  $k$ be an even integer, and let $\varepsilon, \delta > 0$ satisfy 
\begin{equation}\label{eq:IJKW_relation_eps_delta_quantum}
k \;\geq\; C \cdot \frac{1}{\delta} \cdot \left [ \log \! \left( \frac{1}{\delta} \right )  + \log \! \left  ( \frac{1}{\varepsilon} \right ) \right ].
\end{equation}
If $\algoIJKW$ is a quantum circuit of size at most $s$ defined over $\mathcal{S}_{n,k}$ with $k$ output bits such that
\begin{align}\label{eq:assumption-ijkw_quantum}
\E_{B \sim \calS_{n,k},\; \algoIJKW}\left[\algoIJKW(B) = g^k(B)\right] \eqdef \E_{B \sim \calS_{n,k},\; \algoIJKW}\left[\norm{\Pi_{g^{k}(B)}\algoIJKW\ket{B,0^m}}^2 \right]  
 \;\geq\; \eps,
\end{align}
then there is a quantum circuit $\mathcal{B}$ of size $\mathsf{poly}(n,k,s,\log(1/\delta), 1/\varepsilon)$ such that
\begin{equation}\label{eq:IJKW_conclusion_quantum}
\E_{x \sim \{0,1\}^n,\; \mathcal{B}}\Big [\mathcal{B}(x) = g(x)\Big ] \eqdef \E_{x \sim \{0,1\}^n,\; \mathcal{B}}\Big[ \norm{\Pi_{g(x)}\mathcal{B}\ket{x,0^{m'}}}\Big ] \;\geq\; 1 - \delta.
\end{equation}
Moreover, a quantum circuit $\mathcal{B}$ of this form can be constructed with noticeable probability from a quantum circuit $\algoIJKW$ as above by a uniform sequence of quantum circuits, a statement which is formalised as follows. For any choice of constructive functions $\varepsilon = \varepsilon(n)$, $\delta = \delta(n)$, and $k = k(n)$ satisfying the conditions of the theorem, and for any constructive function $s = s(n)$, there is a uniform family $\{\mathcal{C}_n^{\mathsf{IJKW}}\}_{n \geq 1}$ of quantum circuits $\mathcal{C}^{\mathsf{IJKW}}_n$ of size $\mathsf{poly}(n,k,s, \log(1/\delta), 1/\varepsilon)$ for which the following holds. If $\mathcal{C}^{\mathsf{IJKW}}_n$ is given as input a string $\mathsf{code}(\algoIJKW)$ describing a quantum circuit $\algoIJKW$ with the properties above, then when its output is measured it produces with probability $\zeta = \Omega(\varepsilon^2)$ a string $\mathsf{code}(\mathcal{B})$ describing a quantum circuit $\mathcal{B}$ of the desired size and success probability.  
\end{theorem}

\begin{proof}
The proof is divided into two parts: \emph{existence} and \emph{uniformity}.
First, we argue that if there is quantum circuit $\algoIJKW$ of the above form, then a corresponding quantum circuit $\mathcal{B}$ exists. Then, for a fixed choice of constructive functions $\varepsilon$, $\delta$, and $k$ that depend on $n$ and satisfy the conditions of the result, and for any constructive function $s = s(n)$, we provide a deterministic algorithm $D$ that, given $1^n$, runs in time $t = \mathsf{poly}(n,k,s,\log(1/\delta), \varepsilon)$ and outputs the description of a quantum circuit $\mathcal{C}_n^{\mathsf{IJKW}}$ as in the statement.

Let $\algoIJKW$ be any quantum circuit of size $s$ satisfying the hypothesis of the theorem. Note that in the proof of Theorem \ref{thm:ijkw} the corresponding circuit $\algoIJKW$ is accessed in a classical way. By Lemma \ref{lem:quantum_as_inherently_random}, the analysis of the success probability of the circuit $\mathcal{B}^{\mathcal{O}}$ constructed in the proof of Theorem \ref{thm:ijkw} remains valid when the oracle $\mathcal{O}$ is replaced by the circuit $\algoIJKW$. However, in the quantum case, instead of getting an inherently random \emph{oracle} circuit $\mathcal{B}^\mathcal{O}$ of size $\mathsf{poly}(n,k,\log(1/\delta),1/\varepsilon)$, we obtain a quantum circuit $\mathcal{B}$ of size $\mathsf{poly}(n,k,s, \log(1/\delta),1/\varepsilon)$, where the size overhead comes from replacing each oracle gate $\mathcal{O}$ by a copy of the quantum circuit $\algoIJKW$ of size $s$. Finally, recall that $\univ = \{0,1\}^n$, and note that the random input $y$ of $\mathcal{B}^{\mathcal{O}}$ appearing in Equation \ref{eq:IJKW_conclusion} can be assumed to be part of the quantum computation of $\mathcal{B}$ by standard techniques. Consequently, we obtain a quantum circuit $\mathcal{B}$ for which Equation \ref{eq:IJKW_conclusion_quantum} holds.  

Next, we argue that after fixing functions $s(n)$, $k(n)$, $\varepsilon(n)$, and $\delta(n)$, there is a quantum circuit $\mathcal{C}^{\mathsf{IJKW}}_n$ that given the \emph{code} of a good quantum circuit $\algoIJKW$ outputs with probability $\Omega(\varepsilon^2)$ the \emph{code} of a quantum circuit $\mathcal{B}$ with the desired properties. The circuit $\mathcal{C}^{\mathsf{IJKW}}_n$ is simply the quantum analogue of the algorithm $\mathcal{D}$ from Construction \ref{const:decoder}. Let $T(n) = \mathsf{poly}(n,k,s,\log(1/\delta), 1/\varepsilon)$ be fixed as in the proof of Theorem \ref{thm:ijkw}. Then, given the code of $\algoIJKW$ and using that $n$, $k$, and $T(n)$ are fixed, $\mathcal{C}^{\mathsf{IJKW}}_n$ proceed as follows. It uses its internal randomness (simulated in a quantum way) to compute as $\mathcal{D}$ in Step 1, then 
it 
simulates $\algoIJKW$ using a \emph{universal quantum circuit} (see Section \ref{sec:prelim-quantum}) of size $\mathsf{poly}(T)$ in order to sample $w \sim \algoIJKW(B)|_A$,\footnote{This is where we use that all parameters are fixed, meaning that we can instantiate a universal quantum circuit for quantum computations containing a fixed number of gates.} and finally it outputs the \emph{description} of a quantum circuit $\mathcal{B}$ that computes as $C_{A,w}$, where $C_{A,w}$ incorporates the code of $\algoIJKW$. Since by the proof of Theorem \ref{thm:ijkw} algorithm $\mathcal{D}$ outputs a good circuit $\mathcal{B}^\mathcal{O}$ with probability $\Omega(\varepsilon^2)$ over its internal randomness and the randomness of $\algoIJKW$ (whenever $\algoIJKW$ satisfies the conditions of the theorem), it  follows that this is also true for $\mathcal{C}^{\mathsf{IJKW}}_n$ and the description $\mathsf{code}(\mathcal{B})$ that it generates from $\mathsf{code}(\algoIJKW)$. 

Observe that  $\mathcal{C}^{\mathsf{IJKW}}_n$ is defined over inputs of length $\mathsf{poly}(s)$, which represent the string $\mathsf{code}(\algoIJKW)$. In addition, $\mathcal{C}^{\mathsf{IJKW}}_n$ has size $\mathsf{poly}(T)$. This includes the time it takes to simulate $\algoIJKW$, and the time it requires to print an explicit description of $\mathcal{B}$, which contains $\mathsf{poly}(T)$ many gates.

Finally, note that the \emph{code} of the quantum circuit $\mathcal{C}^{\mathsf{IJKW}}_n$ is fully explicit, given a choice of parameters. In other words, there is a \emph{deterministic} algorithm that when given $1^n$ runs for at most $t = \mathsf{poly}(T) = \mathsf{poly}(n,k,s,\log(1/\delta), 1/\varepsilon)$ steps and prints $\mathcal{C}^{\mathsf{IJKW}}_n$.
\end{proof}

\subsection{Self-reducibility in the quantum setting}\label{sec:self_reduc}

In this section, we explain how to use the downward and random self-reducibility of a language $L$ to help us to produce a sequence of circuits computing $L$. In more detail, we show how to design a small quantum circuit $B_n$ to compute $L$ on $n$ bit inputs from a large quantum circuit $P_{n - 1}$ computing $L$ on $n-1$ bit inputs and a collection of small quantum circuits $A_1, \ldots, A_t$ with the following guarantee: some $A_i$ offers a good approximation of $L$ over $n$ bit inputs.
We will implement this idea with respect to the language $L^\star$ provided by Theorem \ref{t:TV_language}. Note that the downward self-reducibility and random-self-reducibility of this language holds with respect to \emph{classical} computation, while here we will rely on these structural properties in the context of \emph{quantum} circuits. This is not an issue for the following reasons.
In the case of downward self-reducibility, given a quantum circuit $P_{n - 1}$ for $L^\star_{n - 1}$ (i.e.~$L^\star$ restricted to $n-1$ bit inputs), we observe that its success probability on every input can be amplified to $1 - \mathsf{negl}(n)$. As a consequence, the classical reduction, which makes $\mathsf{poly}(n)$ classical queries, will obtain correct and consistent answers with overwhelming probability, even if $P_{n - 1}$ is a quantum~circuit.
On the other hand, in the case of random self-reducibility, on an input $x$ for $L^\star$, the reduction is implemented by a classical algorithm that makes $n^a$ queries to a classical oracle $A$, where each query is uniformly distributed over $\{0,1\}^n$ (but different queries can be correlated). If the oracle $A$ answers all queries according to a fixed function $\widetilde{f}_n \colon \{0,1\}^n \to \{0,1\}$ that is $1/n^b$-close to $L_n^\star$, the reduction gives a correct answer on $x$ with high probability. Now note that Definition \ref{def:random-selfreducible} does not offer a guarantee when $\widetilde{f}_n$ is not a classical oracle, e.g., if it is defined from a quantum circuit that does not provide a deterministic output. Fortunately, in order to quantize this reduction, we can reduce the analysis to the classical case. This is only possible because here we are in the regime where the correlation of interest is of the form $1 - o(1)$, while for instance in hardness amplification (Section \ref{sec:IJKW}) we consider correlations that are $o(1)$.
In more detail, here is one possible way of doing this. In our proof, $A$ is some quantum circuit that computes $L^\star_n$ with probability $p$, in the sense that the expected agreement between $A$ and $L_\star^n$ on a random input $x$ and the measurement of $A$'s ouput is $p$. It turns out that in our analysis we can take $p \geq 1 - 1/n^{c}$ for a convenient constant $c > a + b$. This allows us to show that there is a large set $S \subseteq \{0,1\}^n$ with $\Pr_x[x \in S] \geq 1 - 1/n^{a + b + 1}$ such that $A$ is correct on each $x \in S$ with high probability. Moreover, by amplifying the success probability of $A$, we can get a quantum circuit $A'$ that is correct on each $x \in S$ with probability at least, say, $1 - 2^{-n}$. Finally, since the reduction makes at most $n^a$ queries and each of them is uniform, by a union bound, with high probability all queries land in $S$. This can be used to show that with high probability all (classical) queries answered by the quantum circuit $A'$ agree with a classical function that is $1/n^b$-close to $L^\star_n$, which guarantees the correctness of the reduction also in our context. Using the ideas described above, it is not hard to implement the result explained at the beginning of this section. We provide the details next.

We start with the following: suppose we have a quantum circuit $U$ that computes a random self-reducible language $L$ with high probability for a uniformly random input, then one can construct a quantum circuit $U^*$ that, with high probability, computes $L$ on every $x\in \01^n$.

\begin{lemma}[Random self-reducibility and quantum circuits]\label{l:amplification-quantum-rsr}
  Let $L:\01^*\rightarrow \01$ be a random self-reducible language with parameters $a,b,c$ (as described in \Cref{def:random-selfreducible}). For every $n$, suppose we have the description of a quantum circuit $U$ such that 
  \begin{align}\label{eq:assumption-rsr}
    \E_{x \in \01^n} \left[ \norm{\Pi_{L(x)} U\ket{x,0^ q}  }^2 \right]\geq
    1-\frac{1}{n^{k}},
  \end{align}
  for some $k \geq 2b + a$.
  
There is a $O(|U| \cdot \poly(n))$-size quantum circuit $U^*$ that satisfies
  \[
    \norm{\tilde{\Pi}_{x} U^*\ket{0,x,0^{q^*}}}^2 \geq 1- 2^{-2n+1} \qquad \text{ for every } x\in \01^n,
  \]
  where $\tilde{\Pi}_{x} = \kb{L(x)} \otimes \kb{x} \otimes \kb{0^{q^*}}$ and $q^* = \poly(n)$.
\end{lemma}

\begin{proof}
Let us define $Y = \{x \in \01^n : \norm{\Pi_{L(x)} U \ket{x,0^q}}^2 \geq \frac{2}{3}\}$. It follows from \Cref{eq:assumption-rsr}, that $|Y| \geq (1-\frac{3}{n^{k}})2^n$. In this case, we can consider the circuit $U^{\textup{amp}}$ that, on input $x$, computes $U$ on the input $O(n)$ times in parallel and answers with the majority of the outputs. It follows that in this case, for every $x \in Y$, we have that $\norm{\Pi_{L(x)} U^{\textup{amp}} \ket{x,0^{O(qn)}}}^2 \geq 1-2^{-n}$.

Let $g:\01^*\rightarrow \01^*$ be the polynomial-time computable {\em classical} function from \Cref{{def:random-selfreducible}} in the definition of the random self-reducibility of $L$. Before defining $U^*$, we first define the unitary $\tilde{U}_1$ 
  \begin{align}
      \label{eq:tildeU1}
    \widetilde{U}_1:\ket{x,0}\rightarrow \frac{1}{ \sqrt{2^{n^c}}}\underbrace{\sum_{r \in \01^{n^c}} \ket{x,r,0} \bigotimes_{i \in [n^a]}
  \ket{g(i,x,r),0}}_{:=\ket{\chi}}.
  \end{align} 

  We now use the circuit $U^{\textup{amp}}$ to compute
  $L$ for every $g(i,x,r)$, i.e., we apply $\widetilde{U}_2 = \id \otimes (U^{\textup{amp}})^{\otimes n^a}$ to the state $\ket{\chi}$ to obtain
  \begin{align}
      \label{eq:tildeU2}
    \ket{\psi}=\frac{1}{\sqrt{2^{n^c}}}\sum_{r \in \01^{n^c}} \ket{x,r,0} \bigotimes_{i \in [n^a]}
  \big(U^{\textup{amp}}\ket{g(i,x,r),0}\big).
  \end{align}

Let us fix some random $r$.
  Observe that by our assumption on $g$ (i.e. $g(i,x,r) \sim \calU_n$), it follows from a union bound that there exists some $i \in [n^a]$ such that $g(i,x,r) \not\in Y$ with probability at most $\frac{3}{n^{k-a}}$. 
  Assuming that for every $i \in [n^a]$ we have $g(i,x,r) \in Y$, it follows that for every   $i \in [n^a]$ 
  \[
  \Pr_r\left[\Big\|\Pi_{L(g(i,x,r))}
  U\ket{g(i,x,r),0}\Big\|^2  \right] 
  \geq 1 - 2^{-n}.
\] 

    Let $\widetilde{U}_3$ be a unitary that implements the   classical circuit $h:\01^*\rightarrow \01$ from \Cref{def:random-selfreducible}. The action of $\widetilde{U}_3$ on $\ket{\psi}$ can be written as
    \begin{align}
  \label{eq:first-output}
  \ket{\phi}=\frac{1}{\sqrt{2^{n^c}}}\sum_r \tilde{U}_3 \left(\ket{x,r,0} \bigotimes_{i \in [n^a]}
  U^{\textup{amp}}\ket{g(i,x,r),0}\right).
  \end{align}
  
Notice that if we assume that for every $i \in [n^a]$ we have $g(i,x,r) \in Y$, and that $U^{\textup{amp}}$ is correct for every $g(i,x,r)$, we have from \Cref{def:random-selfreducible} 
that 
\[\norm{\Pi_{L(x)}\ket{\phi}}^2 \geq 1 - 2^{-2n}.\]
  
It follows from a union bound that for every $x \in \01^n$, \[\norm{\Pi_{L(x)}\ket{\phi}}^2 \geq 1 - \frac{3}{n^{k-a}} - \frac{n^a}{2^{n}} - \frac{1}{2^{2n}} \geq 1 - \frac{1}{\poly(n)}.\]
  
    We can pick $U^*$ as the algorithm that runs $\tilde{U}_3\tilde{U}_2\tilde{U}_1$ in parallel $O(n)$ times and answers with the majority. Finally, to remove any garbage from the computation, we can copy the output register into a separate register and uncompute $U^*$ and still compute $L(x)$ with overwhelming probability.
\end{proof}

 We now show that if $L$ is downward self-reducible, then we can construct a quantum circuit $U^*$ that computes $L$ on inputs of size $n$ from a quantum circuit $U_{n-1}$ that computes $L$ on inputs with size $n-1$.

\begin{theorem}[Downward-self-reducibility of $L^\star$ and quantum circuits]\label{lem:dsr_quantum}
Let $s_P \colon \mathbb{N} \to \mathbb{N}$ be a constructive function. Let $L^\star$ be the language from Theorem \ref{t:TV_language}. There is a sequence $\{\mathcal{C}_n^{\mathsf{DR}}\}_{n \geq 1}$ of deterministic circuits $\mathcal{C}_n^{\mathsf{DR}}$ for which the following holds:
\begin{itemize}
    \item[\emph{(}i\emph{)}] \emph{Input:} Each circuit $\mathcal{C}_n^{\mathsf{DR}}$ gets as input $1^n$ and a string $\mathsf{code}(P_{n - 1})$ that describes a quantum circuit $P_{n-1}$ of size $\leq s_P(n-1)$.
    \item[\emph{(}ii\emph{)}] \emph{Uniformity and Size:} Each circuit $\mathcal{C}^{\mathsf{DR}}_n$ is of size $S(n) = \mathsf{poly}(n, s_P(n-1))$, and there is a deterministic algorithm that when given $1^n$ prints $\mathsf{code}(\mathcal{C}^{\mathsf{DR}}_n)$ in time $\mathsf{poly}(S(n))$.
    \item[\emph{(}iii\emph{)}] \emph{Output and Correctness:} If $P_{n - 1}$ computes $L^\star$ on inputs of length $n - 1$ then $\mathcal{C}^{\mathsf{DR}}_n$ outputs the description of a quantum circuit $P_n$ of size $\mathsf{poly}(n, s_P(n-1))$ that computes $L^\star$ on inputs of length $n$.
\end{itemize}
\end{theorem}

\begin{proof}
 By item ($i$), we have that
 $$
  \Big\|{\Pi}_{\Lstar(x)} P_{n-1}\ket{x, 0^{q}}\Big\|^2 \geq 2/3.
 $$
   We first amplify the success probability of $P_{n-1}$  from $2/3$ to $1-\negl(n)$. For this, we perform a standard majority vote on the outputs of $O(n)$ parallel copies of $P_{n-1}$ and construct $P_{n-1}^*$ of size $O(n\cdot s_P(n-1))$ such that %
   \begin{align}
        \label{eq:correctnessofselfred1ampl}
    \Big\|\widetilde{\Pi}_{x} P^*_{n-1}\ket{0, x,0^{q}}\Big\|^2 \geq 1 - \frac{1}{\negl(n)},
  \end{align} %
where $q=O(\poly(n))$ and $\widetilde{\Pi}_{x} = \kb{L^\star(x)} \otimes \kb{x} \otimes \kb{0^q}$. 
  Now, $P_n$ simulates the {\em
  classical} circuit $A^{L^\star}$ from \Cref{def:downward-selfreducible} that computes $L^{\star}(x)$ and answers $A$'s queries to $L^\star$ on instances of size $n-1$ by simulating $P_{n-1}^*$. 
  To remove garbage, $P_n$ copies the output of $A^{L^\star}$ onto the output qubit, and then {\em uncomputes} $A^{L^\star}$ (by also uncomputing the calls to $P_{n-1}^*$). Since each one of the polynomially many queries is correct with probability at least {$1 - 1/\negl(n)$}, we have that the output of $A^{L^\star}$ is correct with probability at least  $1 -  1/\negl(n) \geq 2/3$. Additionally, as $A$ is a $\poly(n)$ time circuit and $P_{n-1}^*$ is of size $O(n \cdot s_P(n-1))$, $P_n$ is of size $\poly(n, s_P(n-1))$, and therefore ($iii$) holds. 
  
  In order to show ($ii$), observe that the circuit $\calC_n^{\textsf{DR}}$ first generates $\code(P_{n-1}^*)$ of size $O(n \cdot s_P(n-1))$ by using $\code(P_{n-1})$. It then converts the classical circuit $A^{L^\star}$ into it's reversible form and replaces these reversible gates with corresponding unitary descriptions. It also replaces the queries made by $A^{L^\star}$ with $\code(P_{n-1}^*)$. All of this can be computed by $\calC_n^{\textsf{DR}}$ using $\poly(n, |\code(P_{n-1})|) = \poly(n, s_P(n-1))$ gates.
\end{proof}

\begin{theorem}[Self-reducibility of $L^\star$ and quantum circuits]\label{lem:quantum_self_reducibility_lemma}
Let $s_A, s_P, t \colon \mathbb{N} \to \mathbb{N}$ be constructive functions. Moreover, let $L^\star$ be the language from Theorem \ref{t:TV_language}, and $a_\star, b_\star$ be the associated constants. There is a sequence $\{\mathcal{C}^{\mathsf{SR}}_n\}_{n \geq 1}$ of quantum circuits $\mathcal{C}^{\mathsf{SR}}_n$ for which the following holds:
\begin{itemize}
    \item[\emph{(}i\emph{)}] \emph{Input:} Each circuit $\mathcal{C}^{\mathsf{SR}}_n$ gets as input $1^n$ and strings $\mathsf{code}(P_{n-1})$, $\mathsf{code}(A_1), \ldots, \mathsf{code}(A_{t(n)})$, where $P_{n-1}$ is a quantum circuit of size $\leq s_P(n-1)$ and each $A_i$ is a quantum circuit of size~$\leq s_A(n)$.
    \item[\emph{(}ii\emph{)}] \emph{Uniformity and Size:} Each circuit $\mathcal{C}^{\mathsf{SR}}_n$ is of size $S(n) = \mathsf{poly}(n,t(n),s_A(n),s_P(n-1))$, and there is a deterministic algorithm that when given $1^n$ prints $\mathsf{code}(\mathcal{C}^{\mathsf{SR}}_n)$ in time $\mathsf{poly}(S(n))$.
    \item[\emph{(}iii\emph{)}] \emph{Output and Correctness:} Assume that $P_{n - 1}$  computes $L^\star$ on inputs of length $n - 1$, and that there exists $i \in [t(n)]$ such that
    \begin{align}
    \label{eq:goodAi}
    \Pr_{x \sim \{0,1\}^n,\,A_i}[A_i(x) = L^\star(x)] \;\geq\; 1 - n^{-2b_\star-a_\star}.
    \end{align}
    Then with probability at least $1 - 1/500n^2$ over its output measurement, $\mathcal{C}^{\mathsf{SR}}_n$ generates the description $\mathsf{code}(B_n)$ of a quantum circuit $B_n$ of size $\mathsf{poly}(n,s_A(n))$ that correctly computes $L^\star$ on inputs of length $n$. In other words, for every $x \in \{0,1\}^n$,
    $$
    \Pr_{B_n}[B_n(x) = L^\star(x)] \geq 2/3.
    $$
\end{itemize}
\end{theorem}

\noindent \textbf{Note.} Notice that from \Cref{lem:dsr_quantum}, we could achieve a circuit of size $\poly(n,s_P(n-1))$ that computes $L^*$ on inputs of size $n$. The non-trivial aspect of \Cref{lem:quantum_self_reducibility_lemma} is to achieve a circuit $P_n$ whose size depends only on $n$ and $s_A(n)$ and is {\em independent} of $s_{P}(n-1)$.

\begin{proof}
First, given $\code(P_{n-1})$, we use~\Cref{lem:dsr_quantum} to construct a circuit $P_n$ with $|\code(P_n)| = \poly(n, s_P(n-1))$ such that for every $x \in \01^n$ and for $q' = \poly(n)$, we have
\begin{equation}
\label{eq:closepn}
    \big\|\widetilde{\Pi}_{x} P_n\ket{0, x, 0^{q'}}\big\|^2 \geq 1- \negl(n)
\end{equation}
where $\widetilde{\Pi}_{x} = \kb{\Lstar(x)} \otimes \kb{x} \otimes \kb{0^{q'}}$. In the remainder of this proof, the usage of $\widetilde{\Pi}_{x}$ will denote the process of checking if the output qubit is $\ket{\Lstar(x)}$, the input qubits remain as $\ket{x}$ and all auxiliary qubits are set to $0$.
We proceed to construct $B_n$ in three steps using the random and downward self-reducibility of $\Lstar$. \\

\noindent\underline{Step 1:} Notice that, the theorem statement provides no guarantees on how well $A_\ell$ performs for $\ell \neq~i$. To identify the circuits (among $\{A_1,\ldots,A_{t(n)}\}$) which compute $\Lstar$ correctly with probability at least $1 - 1/\poly(n)$, under the promise that there exists one,  we carry out the following test. Let $R = O(\log (t(n)/\eta))$ where $\eta = 1/\poly(n)$. For every $\ell \in [t(n)]$, pick uniformly random $x^1_\ell, \ldots ,x^R_\ell \in \01^n$, and for every $r\in [R]$ run $A_\ell$ and $P_n$ (constructed at the start of the proof) separately on two copies of $\ket{0,x^r_\ell,0^{q'}}$, measure the first qubit and let the outputs be $b_{\ell}^r,c^r_\ell\in \01$ respectively. Consider the~set 
$$
\mathcal{J} = \left\{\ell \in [t(n)]:  \sum_{r=1}^R [b_{\ell}^r = c^r_\ell] \geq \frac{3R}{4}\right\}.
$$
Note that it is possible that there is an $\ell \in \mathcal{J}$ that passes the test but $A_\ell$ does not compute~$\Lstar$ correctly. However, it suffices for us to show that, with high probability (over the randomness of sampling $x^i_j$s) if $\ell\in \mathcal{J}$, then the quantum circuit $A_\ell$ computes~$\Lstar$. Towards this end, we define  $t_\star = 2b_\star + a_\star + 1$ and show that, supposing $\Pr [A_\ell(x)=P_n(x)] < 1 - n^{-t_\star}$ (where the probability is taken over uniformly random $x\in \01^n$ and randomness in $A_{\ell},P_n$) then with high probability $\ell \notin \mathcal{J}$. To prove this, first notice that if $\Pr [A_\ell(x)=P_n(x)]\leq  1 -  n^{-t_\star}$, then applying the Chernoff bound (Theorem~\ref{lem:chernoff}) with $\mu < 1 -  n^{-t_\star}$ and $\delta = O(1)$, we have
$$
\Pr\Big[\frac{1}{R}\sum_r[b^r_\ell=c^r_\ell]\geq \frac{3}{4}\Big]\leq e^{-O(R)},
$$
This implies that with probability $\geq 1-e^{-O(R)}$ (over the random samples $x^i_j$s), if  $\ell \in \mathcal{J}$, then we have $\Pr_x [A_\ell(x)=P_n(x)]\geq 1- 1/n^{t_\star}$. Moreover, along with~\Cref{{eq:closepn}}, this implies that  with probability $\geq 1-e^{-O(R)}$, if $\ell\in \mathcal{J}$ then
\begin{align*}
\Pr_x[A_\ell(x) = \Lstar(x)] & \geq \Pr_x [A_\ell(x)=\Lstar(x) | P_n(x) = \Lstar(x)] \cdot \Pr_x [P_n(x) = \Lstar(x)] \\
& \geq (1 - n^{-t_\star})\left( 1 - {\negl(n)} \right) \geq 1 - n^{-2b_\star-a_\star}
\end{align*}
where the second inequality uses the fact that the conditional probability refers to the event that $A_\ell(x) = P_n(x)$.  Hence, with probability at least $1-e^{-O(R)}$, if $\ell \in \mathcal{J}$ then  $\Pr [A_\ell(x)=\Lstar(x)] \geq 1- n^{-2b_\star-a_\star}$. 
By taking a union bound over all $\ell \in \mathcal{J}$, we now have
$$
\Pr\Big[\exists \ell\in\mathcal{J}: \Pr [A_\ell(x)=\Lstar(x)] < 1-n^{-2b_\star-a_\star}\Big]  \leq t(n)\cdot e^{-O(R)} \leq \eta. 
$$
Hence with probability at least $1-\eta$, every $\ell\in \mathcal{J}$ satisfies  $\Pr [A_\ell(x)=\Lstar(x)] \geq 1-n^{-2b_\star-a_\star}$.

  Before proceeding to the next step, we argue that with probability $\geq 1 -e^{-O(R)}$, $\mathcal{J}$ is non-empty.   By assumption, we have that for $i \in [t(n)]$ it follows that $\Pr_x[A_i(x) = \Lstar(x)] \geq (1 - n^{-2b_\star-a_\star}$). 
  Notice that since $\Pr_x[P_n(x) = \Lstar(x)] \geq 1 - \negl(n)$, we have that $\Pr_{x_1,...,x_R}[\exists r \in [R], c^r_i \ne \Lstar(x^r_i)] \leq \negl(n)$ by a union bound. Conditioned that for all $r$ we have $c^r_i = \Lstar(x^r_i)$, we have that  the probability that $A_i$ fails the test is very small i.e., 
  $$
\Pr\Big[\frac{1}{R}\sum_r[b^r_i=c_i^r]\leq \frac{3}{4}\Big] = 
\Pr\Big[\frac{1}{R}\sum_r[b^r_i=\Lstar(x^ r_i)]\leq \frac{3}{4}\Big] 
\leq e^{-O(R)},
$$
where the inequality follows by applying the Chernoff bound (Theorem~\ref{lem:chernoff}) on Eq.~\eqref{eq:goodAi},
Hence, with probability at least $1 - e^{-O(R)} -\negl(n) = 1 - e^{-O(R)}$, $\mathcal{J}$ contains $i$.\\

\noindent\underline{Step $2$:} In Step $1$, we showed that circuits $A_\ell$ for $\ell \in \mathcal{J}$ compute $\Lstar$ on \emph{average} $x$. Now we use the random self-reducibility property of $\Lstar$ to obtain a circuit that performs well for \emph{every} $x\in \01^n$ and not just for a uniformly random $x$. In this  direction, for every $\ell \in \mathcal{J}$, given $\code(A_\ell)$, we use~\Cref{l:amplification-quantum-rsr} to construct the circuit $A_\ell^*$ (with $|\code(A_\ell^*)| = \poly(|\code(A_\ell)|,n) = \poly(s_A(n), n)$). Then, for every $\ell \in \mathcal{J}$ such that $\E_x [ A_\ell(x) = \Lstar(x)] \geq 1 - n^{-2b_\star -a_\star}$ we also have %
 \begin{align}
 \label{eq:verygoodAell}
  \big\|\widetilde{\Pi}_{x}  A_\ell^*\ket{0,x,0^{\bar{q}}}\big\|^2 \geq 1- 2^{-2n+1}, \quad \text{for every } x \in \01^n.
 \end{align}

We now prove item ($iii$). Pick $\eta=1/500n^2$ and note that with probability at least $ 1 - \eta = 1 - 1/500n^2$, every $\ell \in \mathcal{J}$ satisfies Eq.~\eqref{eq:verygoodAell}. Hence, picking an arbitrary $\ell \in \mathcal{J}$ and setting $B_n = A_{\ell}^*$ gives us the desired quantum circuit with $|\code(B_n)| = |\code(A_{\ell}^*)| \leq \poly(s_A(n), n)$. \\
Finally, observe that Step $1$ can be described with $\poly( t(n) \cdot |\code(A_\ell^*)| + |\code(P_n)|) = \poly(n, t(n), s_P(n-1), s_A(n))$ gates. Similarly, Step $2$ can be described with $\poly(T, t(n), s_P(n-1), s_A(n)) = \poly(n, t(n), s_P(n-1), s_A(n))$ gates. Step $3$ uses $\poly(n, s_A(n))$ gates. Putting these together, $|\code(\calC_n^{\textsf{SR}})| = \poly(n, t(n), s_P(n-1), s_A(n))$. This proves item $(ii)$ and completes the proof of the theorem.
\end{proof}

\section{A conditional PRG against uniform quantum circuits}\label{sec:prg-construction}
In this section, we put together the results of Section \ref{sec:tech-tools}, and show that if polynomial-space classical algorithms cannot be simulated in sub-exponential time by quantum algorithms, then there exists a pseudorandom generator secure against uniform quantum computations.

\begin{theorem}[Conditional $\PRG$ against uniform quantum computations]\label{t:main_PRG_restated}
Suppose that $\mathsf{PSPACE} \nsubseteq \mathsf{BQSUBEXP}$. In other words, there is a language $L \in \mathsf{PSPACE}$ and $\gamma > 0$ such that $L \notin \mathsf{BQTIME}[2^{n^\gamma}]$. Then, for some choice of constants $\alpha \geq 1$ and $\lambda \in (0,1/5)$, there is an infinitely often $(\ell, m, s, \varepsilon)$-generator $G = \{G_n\}_{n \geq 1}$, where 
$\ell(n) \leq n^\alpha$, $m(n) = \lfloor 2^{n^\lambda} \rfloor$,  $s(m) = 2^{n^{2\lambda}} \geq \mathsf{poly}(m)$ \emph{(}for any polynomial\emph{)}, and $\varepsilon(m) = 1/m.$
\end{theorem}

In the proof given below, one can even take a larger constant $\lambda$ closer to $1$. However, since we are not optimizing the choice of the constant $\alpha$ in the seed length, this is inessential.

\begin{proof}
Let $L^\star \subseteq \{0,1\}^*$ be the special language from Theorem \ref{t:TV_language}. For each $n \geq 1$, let $f_n \colon \{0,1\}^n \to \{0,1\}$ be the indicator Boolean function which agrees with the set $L^\star \cap \{0,1\}^n$. We  use these functions to define a set of candidate generators $\{G^{\alpha, \lambda}\}_{\alpha,\lambda}$, where $G^{\alpha, \lambda} = \{G^{\alpha, \lambda}_n\}_{n \geq 1}$ with parameters $\ell$, $m$, $s$, and $\varepsilon$ as in the statement of the result.  We then argue that if none of them is an infinitely often $(\ell, m, s, \varepsilon)$-generator, then for every $\gamma > 0$ we have $L^\star \in \mathsf{BQTIME}[2^{n^\gamma}]$. Since this language is complete for $\mathsf{PSPACE}$, as a consequence we get $\mathsf{PSPACE} \subseteq \mathsf{BQSUBEXP}$, contradicting the hypothesis of the theorem.

Each generator $G^{\alpha,\lambda} = \{G^{\alpha, \lambda}_n\}_{n \geq 1}$ for a large enough $\alpha \geq 1$ is defined as follows. Let $d_\star \geq 1$ be the constant appearing in Theorem \ref{t:TV_language}, and $a_\star, b_\star, c_\star \geq 1$ be the constants from Definition \ref{def:random-selfreducible} associated with the random-self-reducibility of $L^\star$. Next, for every $n \geq 1$ we set
$$
 n_1(n) \eqdef k \cdot n ~~(\text{for}\;k(n) \eqdef 2n^{2b_\star + a_\star + d_\star + 2}), \quad n_2(n) \eqdef kn + k, \quad \ell(n) \eqdef (kn + k)^2 \leq n^{\alpha}, \quad m(n) \eqdef \lfloor 2^{n^\lambda} \rfloor,
$$
and introduce functions
$$
g_n \colon \{0,1\}^{n_1} \to \{0,1\}^k, \quad h_n \colon \{0,1\}^{n_2} \to \{0,1\}, \quad G_n^{\alpha, \lambda} \colon \{0,1\}^{\ell} \to \{0,1\}^m,
$$
which are defined as follows:
\begin{eqnarray}
    g_n & \equiv &  f_n^k,~\text{i.e.},~g(x^1, \ldots, x^k) \eqdef (f_n(x^1), \ldots, f_n(x^k))~\text{for every}~\vec{x} \in \{0,1\}^{kn}, \nonumber \\
    h_n(\vec{x}, r) & \eqdef & g_n(\vec{x}) \cdot r = \sum_{i =1}^k f(x^i) \cdot r~\;(\mathsf{mod}~2),~\text{where}~r \in \{0,1\}^k,~\text{and} \nonumber \\
    G^{\alpha, \lambda}_n(z) & \eqdef & \mathsf{NW}^{h_n}(z)~\text{instantiated with}~m(n)~\text{output bits}. \nonumber
\end{eqnarray}
This completes the definition of the generator $G^{\alpha, \lambda} \eqdef \{G^{\alpha, \lambda}_n\}_{n \geq 1}$.

We argue next that if $G^{\alpha, \lambda}$ is {\em not} an infinitely often $(\ell, m, s, \varepsilon)$-generator, then $L^\star \in \mathsf{BQTIME}[2^{n^{3 \lambda}}]$. First, note that it has the correct \emph{stretch}. Moreover, $f_n$ can be computed in (deterministic) time $O(2^{n^{d_\star}})$, $g_n$ and $h_n$ can be computed in time $\mathsf{poly}(n) \cdot 2^{n^{d_\star}}$, and the sets in the combinatorial design can each be computed in time $\mathsf{poly}(n)$. Given these time bounds, it follows that $G_n^{\alpha, \lambda}$ can be computed in time $m(n) \cdot \mathsf{poly}(n) \cdot 2^{n^{d_\star}} = O(2^{\ell(n)})$, by our choice of $\ell(n)$. In other words, the generator also satisfies the \emph{uniformity and running time} requirements. Thus, if $G^{\alpha, \lambda}$ is not an infinitely often $(\ell, m, s, \varepsilon)$-generator it must be the case that its output distributions violate the \emph{pseudorandomness} condition. We explore this in what follows.

Let $\mathcal{D}_m \equiv G^{\alpha, \lambda}_n(\mathcal{U}_\ell)$ be the distribution induced by the output of the generator on a random input seed. Since the generator is \emph{not} infinitely often pseudorandom for parameters $s(m) = 2^{n^{2\lambda}}$ and $\varepsilon(m) = 1/m$, there is a deterministic algorithm $A(1^m)$ that runs in time $s(m)$, outputs a ``distinguisher'' quantum circuit $D_m$ over $m$ input variables and of size at most $s(m)$, and for \emph{every large enough} $n$ (say, $n \geq \kappa \in \mathbb{N}$),
 \[ 
\left |  \Pr_{x \sim \{0,1\}^m}[D_m(x) = 1] - \Pr_{y \sim \mathcal{D}_m}[D_m(y) = 1]  \right | > \eps(m).
 \]

Our next step is to argue that algorithm $A$ can be used to define a \emph{uniform} family of quantum circuits $Q_n$ of size at most $O(2^{n^{3\lambda}})$ that correctly decide $L^\star$ on $n$-bit inputs. In other words, according to our notation, for every $x \in \{0,1\}^n$ we have $\Pr_{Q_n}[Q_n(x) = f_n(x)] \geq 2/3$. By the discussion above, this completes the proof of the theorem.\\

\noindent \textbf{The quantum circuit $\bm{Q_n(x)}$.} We now describe each quantum circuit $Q_n$ and argue about its correctness. This circuit computes in $1 + (n - \kappa + 1) + 1$ sequential stages, which are delegated to sub-circuits $Q'_{0}$ followed by $Q'_\kappa, Q'_{\kappa + 1}, \ldots, Q'_n$ and a final sub-circuit $E_n$:\\

\noindent \emph{Initialization:} $Q'_0$  on input $\ket{0^{q_0}}$ prints the description $\mathsf{code}(P_{\kappa - 1})$ of a quantum circuit $P_{\kappa - 1}$ of size $O(1)$  that computes the Boolean function $f_{\kappa - 1}$ corresponding to $L^\star$ over $(\kappa - 1)$-bit strings. 

\vspace{0.3cm}

\noindent \emph{Core Stages:} For every $j \in \{\kappa, \ldots, n\}$, the circuit $Q'_j$ expects as input a description $\mathsf{code}(P_{j - 1})$ of a quantum circuit $P_{j - 1}$ that computes $f_{j - 1}$. The goal of $Q'_j$ is to output with high probability the code of a circuit $P_j$ that computes $f_j$.

\vspace{0.3cm}

\noindent \emph{Final Computation:} Quantum circuit $E_n$ expects input strings $\mathsf{code}(P_n)$ and $x \in \{0,1\}^n$, and outputs the simulation of $P_n$ on $x$.

\vspace{0.3cm}

\noindent (We omit in this description the amplification of the success probability of $Q_n$ from $1/2 + \Omega(1)$ to $\geq 2/3$ (see e.g.~Section \ref{sec:self_reduc}, where more elaborate amplifications are discussed.)\\

Next, we formalise this idea, analysing the size and uniformity of quantum circuits $Q'_j$ and $E_n$, the size of the involved circuits $P_j$, and the success probability of each $Q'_j$. This will allow us to bound the size of $Q_n$ and to show that it correctly computes $f_n$.

In order to define these circuits, we will make use of the uniform families of quantum circuits from Section \ref{sec:tech-tools} and of the uniform family of quantum circuits $\{D_{m(n)}\}_{n \geq 1}$ that violate the pseudorandomness of the generator whenever $n \geq \kappa$. Our main goal is to prove the following lemma.

\begin{lemma}\label{lem:main_PRG}
There exist universal constants $C_U \geq C_Q \geq C_P \geq 1$ for which the following holds. Let $s_P(n) \eqdef 2^{C_P \cdot n^{2 \lambda}}$ for every $n \geq 1$. For every $j \in \{\kappa, \ldots, n\}$, there is a  quantum circuit $Q'_j$ such~that:
\begin{itemize}
    \item[\emph{(}i\emph{)}] \emph{Input:} $Q'_j$ expects as input the description of a quantum circuit $P_{j - 1}$ of size $\leq s_P(j-1).$
    \item[\emph{(}ii\emph{)}] \emph{Size and Uniformity:} $Q'_j$ is a circuit of size $\leq 2^{C_Q  \cdot j^{2 \lambda}}$. Moreover, there is a deterministic algorithm that when given $1^j$ prints $\mathsf{code}(Q'_j)$ in time $\leq 2^{C_U \cdot j^{2 \lambda}}$.
    \item[\emph{(}iii\emph{)}] \emph{Output and Correctness:} If the input circuit $P_{j - 1}$ correctly computes $f_{j - 1}$, then with probability at least $1 - 1/100j^2$ over its output measurement, $Q'_j$ generates the description of a circuit $P_j$ of size $\leq s_P(j)$ that  computes $f_j$.
\end{itemize}
\end{lemma}

Assuming Lemma \ref{lem:main_PRG}, we can complete the proof of Theorem \ref{t:main_PRG_restated} as follows. By the definition of $Q_n$ and its components, it follows from a union bound over all measurements in the core stages of $Q_n$ that the probability that the string $\mathsf{code}(P_n)$ output by $Q'_n$ does not describe a quantum circuit $P_n$ that computes $f_n$ is a most
$$
\sum_{j = \kappa}^n 1/100j^2 \;\leq\; \frac{1}{100} \cdot \sum_{j \geq 1} \frac{1}{j^2} \;=\; \frac{1}{100} \cdot \frac{\pi^2}{6} \;\leq\; \frac{1}{50}.  
$$
In this case, for every fixed $x \in \{0,1\}^n$, $\Pr_{P_n}[P_n(x) = f_n(x)] \geq 2/3$. Since in the final computation stage of $Q_n$ it uses $E_n$ to simulate the input circuit $P_n$ on $x$, we get by a union bound that on each input $x$,
$$
\Pr_{Q_n}[Q_n(x) \neq f_n(x)] \;\leq\; 1/50 + 1/3 \;<\; 2/5. 
$$
In other words, $Q_n$ computes $f_n$. The size of $Q_n$ is given by the sum of the sizes of each component. First, we can assume that $\mathsf{size}(Q'_0) \leq 2^{C_P \cdot (\kappa - 1)^{2 \lambda}}$ provided that $C_P$ is a large enough constant, given that $\kappa$ is constant. In addition, we have
$$
\sum_{j = \kappa}^n \mathsf{size}(Q'_j) \;\leq \; n \cdot \mathsf{size}(Q'_n) \;\leq\; n \cdot 2^{C_Q \cdot n^{2\lambda}}.
$$
Finally, the size of $E_n$ is at most polynomial in the size of the input circuit $P_n$. Overall, we get that $\mathsf{size}(Q_n) = 2^{O(n^{2 \lambda})} = O(2^{n^{3\lambda}})$, as desired. The uniformity of $Q_n$ follows from the uniformity of its components (the code of $Q'_0$ can be obtained using an exhaustive computation, since $\kappa$ is constant.). It follows from this discussion that $L^\star \in \mathsf{BQTIME}[2^{n^{3\lambda}}]$.

We now proceed to prove Lemma \ref{lem:main_PRG}, which finishes our proof.

\begin{proof}[Proof of \Cref{lem:main_PRG}]

 It will be evident from our proof that large enough constants $C_P$, $C_Q$, and $C_U$ can be chosen so that the argument works. Furthermore, from the proof given below it will be clear that every sub-circuit of $Q'_j$ can be uniformly constructed in deterministic time that is polynomial in the size of the sub-circuit. From this and using that it is easy to describe how these sub-circuits are connected, we get that the sequence $\{Q'_j\}_{j \geq \kappa}$ of quantum circuits $Q'_j$ is uniform. 

Let $j \in \{\kappa, \ldots, n\}$ be fixed, and assume that $Q'_j$ is given as input a string $\mathsf{code}(P_{j - 1})$ representing a quantum circuit $P_{j - 1}$ of size $\leq s_P(j - 1)$ that computes $f_{j - 1}$. First, $Q'_j$ invokes the deterministic circuit $\mathcal{C}^\mathsf{DR}_j$ from Lemma \ref{lem:dsr_quantum} on $\mathsf{code}(P_{j - 1})$, obtaining from it a string $\mathsf{code}(F_j)$ describing a quantum circuit $F_j$ of size $\mathsf{poly}(s_P(j-1))$ that computes $f_j$. Note that $F_j$ might contain more than $s_P(j)$ gates, so it cannot be used as the output circuit $P_j$. 
However, we show next that we can use $F_j$ to uniformly construct the circuit $P_j$ that computes $f_j$ with size $s_P(j)$. For that we need four steps: 
\\

\noindent \emph{1. From a large circuit for $f_j$ to a smaller approximate circuit for $h_j$ \emph{(}with non-trivial probability\emph{)}.} Circuit $Q'_j$ takes the code of $F_j$ and produces from it a string $\mathsf{code}(H_j)$ describing  a quantum circuit $H_j$ that computes the function $h_j$ defined above. Note that the definition of $H_j$ from $F_j$ is  elementary, and that $\mathsf{size}(H_j) = \mathsf{poly}(\mathsf{size}(F_j)) = \mathsf{poly}(s_P(j - 1))$. Recall that $h_j$ is defined over $\ell(j) = \mathsf{poly}(j)$ input bits, and that the corresponding function $\mathsf{NW}^{h_j}$ from above produces $m(j) = \lfloor 2^{j^{\lambda}} \rfloor$ output bits. We now invoke Lemma \ref{lem:nw_quantum} with  parameters associated with the Nisan-Wigderson generator for index value $j$: function $h_j$ and its corresponding quantum circuit $H_j$, stretch $m(j) = \lfloor 2^{j^{\lambda}} \rfloor$, and distinguisher circuit $D_{m(j)}$ of size $s(j) = 2^{j^{2\lambda}}$ and advantage $\gamma(j) = 1/m(j)$. From this lemma and our choice of parameters, it follows that $Q'_j$ has access to a quantum circuit $\mathcal{C}^{\mathsf{NW}}$ of size $\mathsf{poly}(m(j),s_{P}(j-1),s(j)) = \mathsf{poly}(s_P(j - 1)) = 2^{O(j^{2\lambda})}$ such that, when given $\mathsf{code}(H_j)$, $\mathsf{code}(D_{m(j)})$, the input length of $h_j$ and the stretch value $m(j)$, it outputs with probability at least
$\Omega(\gamma(j)/m(j)^2) = \Omega(1/m(j)^3) = \Omega(2^{-3 \cdot j^{\lambda}})$ the encoding $\mathsf{code}(A_j)$ of a quantum circuit $A_j$ of size $O(m(j)^2 \cdot s(j)) = O(2^{1.01 \cdot j^{2\lambda}})$ 
such that
\begin{equation}\label{eq:NW_reconstruction_good}
\Pr_{x,\,r,\,A_j}[A_j(x,r) = h_j(x)] \;\geq\; \frac{1}{2} + \frac{\gamma(j)}{2m(j)} \;=\; \frac{1}{2} + \frac{1}{2m(j)^2} \;\geq\; \frac{1}{2} + 2^{-3 \cdot j^{\lambda}}.
\end{equation}
Crucially, note that the size of $A_j$ does not depend on the sizes of $F_j$ and $P_{j - 1}$. In its next steps, $Q'_j$ tries to compute from $A_j$ a quantum circuit for $f_j$ of size $\leq 2^{C_P \cdot j^{2\lambda}}$, while maintaining its total number of gates $\leq 2^{C_Q \cdot j^{2\lambda}}$ (we will boost the success probability of $Q'_j$ later in the proof). Note that the constant $C_Q$ might depend on $C_P$ once we fix a large enough $C_P$, and indeed this is needed for this plan to work.\
~\\

\noindent \emph{2. From a circuit approximating $h_j$ to a non-trivial quantum circuit for $g_j$ \emph{(}with probability $1$\emph{)}.} Given the string $\mathsf{code}(A_j)$ produced by $\mathcal{C}^{\mathsf{NW}}$  in the step above, $Q'_j$ proceeds as follows. It instantiates the corresponding deterministic circuit $\mathcal{C}^{\mathsf{GL}}$ from Lemma \ref{lem:gl}, which computes from $\mathsf{code}(A_j)$ a string $\mathsf{code}(B_j)$ describing a quantum circuit $B_j$,
with
 $\mathsf{size}(B_j) = \mathsf{poly}(j,k(j)) \cdot \mathsf{size}(A_j) = O(2^{1.02 \cdot j^{2 \lambda}})$.
If we assume that $A_j$ satisfies Equation \eqref{eq:NW_reconstruction_good}, it follows from \Cref{lem:gl} that
\begin{eqnarray}\label{eq:GL_is_good}
    \Pr_{x,\,B_j}[B_j(x) = g_j(x)] \;\geq\; \frac{(2^{-3 \cdot j^{\lambda}})^{3}}{2} \;\geq\; 2^{-10 \cdot j^\lambda}. 
\end{eqnarray}
Moreover, note that $\mathsf{size}(\mathcal{C}^{\mathsf{GL}}) = \mathsf{poly}(j , k(j), \mathsf{size}(A_j)) = 2^{O(j^{2\lambda})}$.\\

\noindent \emph{3. From a non-trivial circuit for $g_j$ to an excellent circuit for $f_j$ \emph{(}with non-trivial probability\emph{)}.}
Given the string $\mathsf{code}(B_j)$ produced by $\mathcal{C}^{\mathsf{GL}}$ in the step above, and assuming for now that $B_j$ satisfies Equation \ref{eq:GL_is_good}, $Q'_j$ proceeds as follows.\footnote{To be more formal $Q'_j$ proceeds as we describe independently of the assumption that $B_j$ satisfies Equation \eqref{eq:GL_is_good}, but we keep this assumption for simplicity of exposition.} Let $\varepsilon'(j) \eqdef 2^{-10 \cdot j^\lambda}$ and $\delta(j) \eqdef j^{-2b^\star + a_\star}$, and consider $k(j)$ and the size bound for $B_j$ established above. Then, for this choice of parameters as a function of $j$, it is not hard to check that Equation \ref{eq:IJKW_relation_eps_delta_quantum} in Theorem \ref{thm:ijkw_quantum} holds. Indeed,
$$
k(j) = 2j^{2b_\star + a_\star + d_\star + 2}, \quad \text{while} \quad \frac{1}{\delta(j)} \cdot \left [ \log \left ( \frac{1}{\delta(j)} \right ) + \log \left(\frac{1}{\varepsilon'(j)} \right ) \right ] \;\ll\; j^{2b_\star + a_\star + \lambda + 1} \;\ll\; j^{2b_\star + a_\star + 2}. 
$$
Let $\mathcal{C}^{\mathsf{IJKW}}$ be the quantum circuit provided by Theorem \ref{thm:ijkw_quantum} for our choice of parameters. We rewrite Equation \ref{eq:GL_is_good} more explicitly as follows:
\begin{equation}\label{eq:GL_is_good_explicit}
    \Pr_{x^1, \ldots, x^k \sim (\{0,1\}^j)^k ,\,B_j}\Big [B_j\big (x^1, \ldots, x^k \big) = \big (f_j(x^1), \ldots, f_j(x^k)\big )\Big ] \;\geq\; 2^{-10 \cdot j^{\lambda}},
\end{equation}
where by assumption $0 < \lambda < 1/5$. Now note that there is a mismatch between the expression above and the assumption in Equation \ref{eq:assumption-ijkw_quantum}, because there we sample a $k$-tuple of $j$-bit strings according to $\mathcal{S}_{j, k}$,\footnote{Remember from \Cref{def:ksets} that we define $\mathcal{S}_{j, k} = \{S \subseteq \01^j : |S| = k\}$.} i.e., there is no repetition of strings and we assume a canonical order of the tuple when using it as an input string of length $k \cdot j$ bits. To remedy this situation, $Q'_j$ will not invoke $\mathcal{C}^{\mathsf{IJKW}}$ directly on $\mathsf{code}(B_j)$, as we explain next.

Define the quantum circuit $A'_j$ that attempts to compute $g_j$ on $\mathcal{S}_{j,k}$ based on $B_j$ as follows:
\begin{itemize}
    \item[1.] On an input $\vec{x} \in \mathcal{S}_{j,k}$, 
    \item[2.] Sample a random permutation $\pi \colon [k] \to [k]$.
    \item[3.] Permute the $k$ strings in $\vec{x}$ according to $\pi$, and let $\vec{y} \eqdef \pi(\vec{x})$ be the corresponding $kj$-bit string.
    \item[4.] Output $B_j(\vec{y})$.
\end{itemize}
\begin{claim}\label{cl:PRG_proof_claim} The following holds:
$$
\Pr_{\vec{x} \sim \mathcal{S}_{j,k},\,A'_j}\big [A'_j(\vec{x}) = g_j(\vec{x})\big ] \;\geq\; \frac{2^{-10 \cdot j^\lambda}}{2}.
$$
\end{claim}
\begin{proof}[Proof of Claim \ref{cl:PRG_proof_claim}] Indeed, using the definition of $A'_j$, and recalling the definition of $g_j$,
\begin{eqnarray}
    \Pr_{\vec{x} \sim \mathcal{S}_{j,k},\,A'_j}\big [A'_j(\vec{x}) = g_j(\vec{x})\big ] & = & \Pr_{\substack{\vec{x} \sim \mathcal{S}_{j,k}\\\vec{y} = \pi(\vec{x}),\,B_j}}\big[ B_j(\vec{y}) = g_j(\vec{y}) \big]  \nonumber \\
    & = & \Pr_{\vec{x} \sim \{0,1\}^{jk}\,,B_j}\big [B_j(\vec{x}) = g_j(\vec{x}) \mid \text{strings in}~\vec{x}~\text{are distinct}\,\big ] \nonumber \\
    & \geq &  \Pr_{\vec{x} \sim \{0,1\}^{jk}\,,B_j}\big [B_j(\vec{x}) = g_j(\vec{x})~\wedge~\text{strings in}~\vec{x}~\text{are distinct}\,\big ] \nonumber \\
   (\Pr[E_1 \wedge E_2] \geq \Pr[E_1] - \Pr[\neg E_2]) & \geq & \Pr_{\vec{x} \sim \{0,1\}^{jk}\,,B_j}\big [B_j(\vec{x}) = g_j(\vec{x}) \big ] - \Pr_{\vec{x} \sim \{0,1\}^{jk}} \big [ \exists i_1 \neq i_2~\text{s.t.}~x^{i_1} = x^{i_2} \big ]\nonumber \\
    & \geq & 2^{-10 \cdot j^\lambda} - k^2 \cdot 2^{-j} \;\geq\; \frac{2^{-10 \cdot j^\lambda}}{2}, \nonumber 
\end{eqnarray}
where the last inequality used that $\lambda < 1$ and $k = k(j) = \mathsf{poly}(j)$.
\end{proof}
Circuit $Q'_j$ constructs $\mathsf{code}(A'_j)$ from $\mathsf{code}(B_j)$, then invokes $\mathcal{C}^{\mathsf{IJKW}}$ on $\mathsf{code}(A'_j)$. Assuming Equation \ref{eq:GL_is_good_explicit} holds, $\mathcal{C}^{\mathsf{IJKW}}$ outputs with probability $\Omega(\varepsilon'(j)^2) = \Omega(2^{-20 \cdot j^\lambda})$ a string $\mathsf{code}(B'_j)$ describing a quantum circuit $B'_j$ of size 
\begin{equation}\label{eq:PRG_size_bound}
\mathsf{size}(B'_j) = \mathsf{poly}(j,k(j), \mathsf{size}(B_j),\log(1/\delta(j)),1/\varepsilon'(j)) = \mathsf{poly}(\mathsf{size}(B_j)) =   2^{O(j^{2\lambda})}
\end{equation}
such that
$$
\Pr_{x \sim \{0,1\}^j, B'_j} \big [ B'_j(x) = f_j(x)  \big ] \;\geq\; 1 - \delta(j) \;=\; 1 - j^{-2b_\star - a_\star}.
$$
Note that $\mathsf{size}(\calC^{\mathsf{IJKW}}) = 2^{O(j^{2\lambda})}$.\\

\noindent \emph{Summary of Steps 1--3.} By composing Steps 1--3 described above, we get that for every sufficiently large constant $C_1$ (such that Equation \ref{eq:PRG_size_bound} holds) there is a constant $C_2 > C_1$ and a quantum circuit $Q'_j$ of size $2^{C_2 \cdot j^{2 \lambda}}$ that when given a description $\mathsf{code}(P_{j - 1})$ of a quantum circuit $P_{j-1}$ of size $2^{C_1 \cdot (j - 1)^{2 \lambda}}$ that computes $f_{j - 1}$, outputs with probability $\zeta(j) = \Omega(2^{-23 \cdot j^{\lambda}})$ the description of a quantum circuit $\widetilde{P_j}$ of size $\leq 2^{C_1 \cdot j^{2\lambda}}$ such that
$$
\Pr_{x \sim \{0,1\}^j, \widetilde{P_j}} \big [ \widetilde{P_j}(x) = f_j(x)  \big ] \;\geq\; 1 - j^{-2b_\star -a_\star}.
$$
(Crucially, the size of $\widetilde{P_{j}}$ does not depend on the exponent $C_2$ nor on the size of $P_{j - 1}$, thanks to the results from Section \ref{sec:tech-tools} and the existence of distinguisher circuits $D_{m(j)}$ for every $j \geq \kappa$.) In order to complete the proof of Lemma \ref{lem:main_PRG}, it remains for us to (1) amplify the success probability that $Q'_j$ generates an almost-correct circuit $\widetilde{P_{j}}$; and  (2) convert an almost-correct circuit $\widetilde{P_{j}}$ into a quantum circuit $P_j$ that computes $f_j$ on each input with probability at least $2/3$. These two goals are achieved next.\\

\noindent \emph{4. Amplifying the success probability of $Q'_j(\mathsf{code}(P_{j - 1}))$ and generating a correct circuit $P_j$ for $f_j$.} Our final version of  $Q'_j$ works as follows. This circuit takes its (classical) input $\mathsf{code}(P_{j - 1})$ and runs Steps 1--3 for a total of $t(j) \eqdef \mathsf{poly}(j, 1/\zeta(j)) = 2^{O(j^{\lambda})}$ times, obtaining from this a collection $\widehat{P}_1, \ldots, \widehat{P}_{t(j)}$ of candidate quantum circuits such that,
$$
\text{with probability}~\geq 1 - 1/500j^2,~\text{there is}~i \in [t(j)]~\text{s.t.}~\Pr_{x \sim \{0,1\}^j, \widehat{P_i}} \big [ \widehat{P_i}(x) = f_j(x)  \big ] \;\geq\; 1 - j^{-2b_\star -a_\star}.
$$
Now $Q'_j$ instantiates the corresponding circuit $\mathcal{C}^{\mathsf{SR}}$ from \cref{lem:quantum_self_reducibility_lemma} on inputs $1^j$,  $\mathsf{code}(\widehat{P}_i)_{i \in [t(j)]}$, and $\mathsf{code}(P_{j-1})$  in order to output with high probability a (single) circuit $P_j$ that computes $f_j$. Note that despite the blowup in the size of $Q'_j$ and $P_j$ due to the amplification, our size requirements for them are  maintained. In particular,  we have from \cref{lem:quantum_self_reducibility_lemma} that the size of the output circuit $P_j$ does not depend on $\mathsf{size}(P_{j-1})$.\\

This finishes the construction of  circuits $Q'_j$ satisfying the conditions of Lemma \ref{lem:main_PRG}, which completes the proof of Theorem \ref{t:main_PRG_restated}.
\end{proof}
\end{proof}

\bibliographystyle{alpha}	
\bibliography{refs}	

\appendix

\section{On trivial quantum learning algorithms} \label{app:Fsampling}
In this section, we explain in more detail that there are two quantum learners with different parameters that are ``trivial'', in the sense that they do not really exploit the structure of a concept class $\mathfrak{C}$. First observe that, even classically, there is always a ``brute-force'' learner that works for {\em all} possible functions: query the input function $f:\01^n\rightarrow \{0,1\}$ on all inputs, then store the outcomes in a lookup table which is used as the exponentially large output hypothesis.

The learner above also works in the quantum setting. However, notice that quantumly there exists a second ``trivial'' learner coming from \emph{Fourier sampling}. Before we describe this learner, we briefly discuss the basics of Fourier analysis on the Boolean cube (and refer the reader to~\cite{o2014analysis} for more details).

Given the space of functions $f:\01^n\rightarrow \mathbb{R}$, define the inner product between two functions in this space as $\langle f,g\rangle=\E_x [f(x)\cdot  g(x)]$ where the expectation is taken uniformly from $x\in \01^n$. In this space, one can define a set of \emph{orthonormal basis functions} as follows: for $S\in \01^n$, define $\chi_S(x)=(-1)^{S\cdot x}$ where $S\cdot x=\sum_i S_i\cdot x_i$. Hence, every function $f:\01^n\rightarrow \mathbb{R}$ can be written \emph{uniquely} as
$$
f(x)=\sum_S\widehat{f}(S)\chi_S(x),
$$
where $\widehat{f}(S)=\E_x [f(x)\cdot \chi_S(x)]$ is called a \emph{Fourier coeffient} of $f$. Moreover, it is not hard to see that by Parseval's identity, for every Boolean function $f:\01^n\rightarrow \{-1,1\}$,
$$
\sum_S\widehat{f}(S)^2=\E_x[f(x)^2]=1,
$$
hence the set of of squared Fourier coefficients $\{\widehat{f}(S)^2\}_S$ of a Boolean function $f$ forms a probability~distribution. 

It is well known that in the quantum learning model, given one uniform quantum example $\frac{1}{\sqrt{2^n}}\sum_x\ket{x,f(x)}$ for $f:\01^n\rightarrow \{-1,1\}$, with probability $1/2$ we can \emph{Fourier sample}, i.e., sample from the distribution $\{\widehat{f}(S)^2\}_S$.  Indeed, given $\frac{1}{\sqrt{2^n}}\sum_x\ket{x,f(x)}$, apply the one-qubit Hadamard gate on the last register and measure the last qubit: with probability $1/2$ we get the $1$ outcome, in which case the resulting state is $\frac{1}{\sqrt{2^n}}\sum_xf(x)\ket{x}\ket{1}$, then apply the $n$-qubit Hadamard transform on the first $n$ qubits to obtain the state $\sum_S\widehat{f}(S)\ket{S}\ket{1}$. Measuring this state produces a sample~$S$ from the distribution $\{\widehat{f}(S)^2\}_S$.  Let $\bm{S}$ be the random variable that outputs $S \subseteq [n]$ with probability~$\widehat{f}(S)^2$.

\begin{claim}\label{cl:fourier_sampling} For every $0 \leq \gamma \leq 1$ and $f:\01^n\rightarrow \{-1,1\}$, we have 
$$
\Pr_{\bm{S}}\big[\,|\hat{f}(\bm{S})| \geq \gamma \cdot 2^{-n/2}\,\big] \;\geq\; 1 - \gamma^2. 
$$
\end{claim}

\begin{proof}
It is enough to show that $p \eqdef \Pr_{\bm{S}}[\,\hat{f}(\bm{S})^2 \geq \varepsilon \cdot 2^{-n}\,] \geq 1 - \varepsilon$. Note that with probability $1 - p$ over the choice of $\bm{S}$, we have $\hat{f}(\bm{S})^2 < \varepsilon \cdot 2^{-n}$. But then
$$
\sum_{S \colon \hat{f}(S)^2 \,<\, \varepsilon \cdot 2^{-n}} \hat{f}(S)^2 \;=\;1 - p.
$$
This in turn implies that 
$$
2^n \cdot (\varepsilon \cdot 2^{-n}) \;\geq\; \sum_{S \colon \hat{f}(S)^2 \,<\, \varepsilon \cdot 2^{-n}} \varepsilon \cdot 2^{-n} \;\geq\; 1 - p,
$$
which completes the proof.
\end{proof}

Let $\mathfrak{F}_n$ be the class of all Boolean functions $f \colon \{0,1\}^n \to \{0,1\}$.
Claim \ref{cl:fourier_sampling} implies that $\mathfrak{F}_n$ admits a quantum learning algorithm $\calA$ under the uniform distribution with the following property: for every  $f \in \mathfrak{F}_n$, $\calA$ runs in polynomial time and outputs with probability at least $0.249$ a Boolean circuit $C$ that computes $f$ with probability $\frac{1}{2} + \Omega(2^{-n/2})$. In order to see this, we use Fourier sampling to obtain a set $\bm{S}$, and then with probability $\frac{1}{2}$ we let $C = \chi_S$ and with probability $\frac{1}{2}$ we let $C = \neg \chi_S$. By Claim \ref{cl:fourier_sampling}, we are guaranteed that with probability at least $0.999$, we pick some $S$ such that $\chi_S$ or $\neg \chi_S$ computes $f$ with probability $\frac{1}{2} + \Omega(2^{-n/2})$, and we pick the correct one with probability $\frac{1}{2}$. Since the Fourier sampling routine described above samples from the correct distribution with probability $1/2$, we are done.

\end{document}